%% file: camera_ready.tex
\documentclass[twoside,leqno]{article}

\usepackage[letterpaper]{geometry}

\usepackage{ltexpprt}
\usepackage{hyperref}
\usepackage{latexsym,amssymb,amsmath,graphicx,enumerate,braket, algorithm, mathtools}
\usepackage{cite,comment,easy-todo}
\usepackage{tikz}
\usepackage{caption,subcaption}
\usepackage{enumitem}
\usepackage{thmtools}
\usepackage{thm-restate}
\usepackage{graphicx,color}
\definecolor{blue}{RGB}{0,82,147} 
\definecolor{red}{RGB}{202,033,063}
\definecolor{green}{RGB}{11,218,81} 
\definecolor{graphcolor}{RGB}{50,50,50}
\tikzstyle{dual}=[draw=graphcolor!10,fill=graphcolor!15,opacity=.5]
\usetikzlibrary{decorations.pathreplacing,calligraphy}
\usepackage{algorithm}
\usepackage{algorithmicx}
\usepackage[noend]{algpseudocode}
\usepackage[title]{appendix}

\hypersetup{
	colorlinks=true,
	citecolor=green,
	linkcolor=blue,  
	urlcolor=black,
}

\newcommand{\Zp}{{\mathbb{Z}}_+}

\newcommand{\cost}{{\sf cost}}
\newcommand{\wt}{{\sf wt}}
\newcommand{\lev}{{\sf lev}}
\newcommand{\spn}{{\rm span}}
\newcommand{\rank}{{\rm rank}}

\newcommand{\I}{\mathcal{I}}
\newcommand{\C}{\mathcal{C}}
\newcommand{\D}{\mathcal{D}}

\newcommand{\T}{\mathcal{T}}
\newcommand{\calS}{\mathcal{S}}

\makeatletter
\providecommand*{\cupdot}{%
  \mathbin{%
    \mathpalette\@cupdot{}%
  }%
}
\newcommand*{\@cupdot}[2]{%
  \ooalign{%
    $\m@th#1\cup$\cr
    \sbox0{$#1\cup$}%
    \dimen@=\ht0 %
    \sbox0{$\m@th#1\cdot$}%
    \advance\dimen@ by -\ht0 %
    \dimen@=.5\dimen@
    \hidewidth\raise\dimen@\box0\hidewidth
  }%
}

\providecommand*{\bigcupdot}{%
  \mathop{%
    \vphantom{\bigcup}%
    \mathpalette\@bigcupdot{}%
  }%
}
\newcommand*{\@bigcupdot}[2]{%
  \ooalign{%
    $\m@th#1\bigcup$\cr
    \sbox0{$#1\bigcup$}%
    \dimen@=\ht0 %
    \advance\dimen@ by -\dp0 %
    \sbox0{\scalebox{2}{$\m@th#1\cdot$}}%
    \advance\dimen@ by -\ht0 %
    \dimen@=.5\dimen@
    \hidewidth\raise\dimen@\box0\hidewidth
  }%
}
\makeatother

\newcommand{\myproblem}[1]{{\sf{#1}}}
\def\NP{\mathsf{NP}}

\newtheorem{observation}{Observation}[section]
\newtheorem{claim}{Claim}[section]
\newenvironment{myproof}[1]{\medskip                    
\noindent{\bf{\em #1.}}}{\quad $\Box$\medskip}  %%

\newcommand{\switch}[2]{#2} % activate this line for arXiv

\begin{document}

\newcommand\relatedversion{}
\renewcommand\relatedversion{\thanks{The full version of the paper can be accessed at \protect\url{https://arxiv.org/abs/1902.09310}}} % Replace URL with link to full paper or comment out this line

\title{\Large Arborescences, Colorful Forests, and Popularity\switch{\relatedversion}{}}
%\author{Corey Gray\thanks{Society for Industrial and Applied Mathematics.}
%\and Tricia Manning\thanks{Society for Industrial and Applied Mathematics.}}
\author{
Telikepalli Kavitha
\thanks{Tata Institute of Fundamental Research, Mumbai, India; \tt kavitha@tifr.res.in}
\and Kazuhisa Makino
\thanks{Research Institute for Mathematical Sciences, Kyoto University, Kyoto, Japan; \tt makino@kurims.kyoto-u.ac.jp}
\and Ildik\'o Schlotter
\thanks{Centre for Economic and Regional Studies, Budapest, Hungary; also at Budapest University of Technology and Economics, Budapest, Hungary; \tt schlotter.ildiko@krtk.hun-ren.hu}
\and Yu Yokoi
\thanks{Tokyo Institute of Technology, Tokyo, Japan; \tt yokoi@c.titech.ac.jp}
}
\date{}
\switch{}{\fancyfoot[C]{\thepage}}
\maketitle

% Copyright Statement
% When submitting your final paper to a SIAM proceedings, it is requested that you include
% the appropriate copyright in the footer of the paper.  The copyright added should be
% consistent with the copyright selected on the copyright form submitted with the paper.
% Please note that "20XX" should be changed to the year of the meeting.

% Default Copyright Statement
\switch{
\fancyfoot[R]{\scriptsize{Copyright \textcopyright\ 2024 by SIAM\\
Unauthorized reproduction of this article is prohibited}}
}{}
% Depending on which copyright you agree to when you sign the copyright form, the copyright
% can be changed to one of the following after commenting out the default copyright statement
% above.

%\fancyfoot[R]{\scriptsize{Copyright \textcopyright\ 20XX\\
%Copyright for this paper is retained by authors}}

%\fancyfoot[R]{\scriptsize{Copyright \textcopyright\ 20XX\\
%Copyright retained by principal author's organization}}

%\pagenumbering{arabic}
%\setcounter{page}{1}%Leave this line commented out.

\begin{abstract} \small\baselineskip=9pt Our input is a directed, rooted graph $G = (V \cup \{r\},E)$ where each vertex in~$V$ has a 
% wo->po
%weak ranking 
partial order preference over its incoming edges. 
The preferences of a vertex extend naturally to preferences over arborescences rooted at~$r$. 
We seek a \emph{popular} arborescence in $G$, i.e., one for which there is no  ``more popular'' arborescence. 
Popular arborescences have applications in liquid democracy or collective decision making; however, they need not exist in every input instance.
The \myproblem{popular arborescence} problem is to decide if a given input instance admits a popular arborescence or not.
We show a polynomial-time algorithm for this problem, whose computational complexity was not known previously.

Our algorithm is combinatorial, and can be regarded as a primal-dual algorithm. It searches for an arborescence along with its 
dual certificate,  a chain of subsets of $E$,  witnessing its popularity. In fact, our algorithm solves the more general \myproblem{popular common base} 
problem in the intersection of two matroids, where one matroid is the partition matroid defined by any partition $E = \bigcupdot_{v\in V} \delta(v)$ 
and the other is an arbitrary matroid $M = (E, {\cal I})$  of rank $|V|$, with each $v \in V$ having a 
%weak ranking 
% wo->po
partial order over elements in $\delta(v)$. 
We extend our algorithm to the case with forced or forbidden edges. 

We also study the related \myproblem{popular colorful forest} (or more generally, the \myproblem{popular common independent set}) problem where edges are partitioned into color classes, and the task is to find a 
colorful forest that is popular within the set of all colorful forests. %
For the case with weak rankings, we formulate the popular colorful forest polytope, 
and thus show that a minimum-cost popular colorful forest can be computed efficiently. %in polynomial time when edges have costs.
By contrast, we prove that it is $\NP$-hard to compute a minimum-cost popular arborescence, even when rankings are strict.\end{abstract}

\section{Introduction}
\label{sec:intro}
Let $G = (V \cup \{r\}, E)$ be a directed graph where the vertex $r$ (called the root) has no incoming edge.
Every vertex~$v \in V$ has a partial ordering 
% wo->po
$\succ_v$ (i.e., a preference relation that is irreflexive, antisymmetric and transitive) over its incoming edges,
as in this example from \cite{KKMSS20} 
where preference orders are strict rankings. 
Here $V = \{a, b, c, d\}$ and the preference orders of 
these four vertices on their incoming edges are as follows:

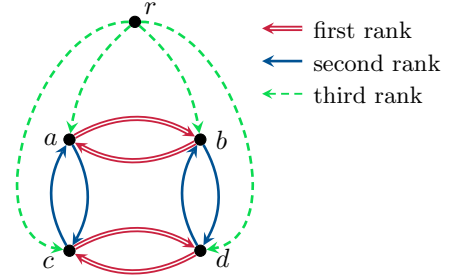
\begin{minipage}[c]{0.55\textwidth}
\centering
		\begin{align*}
                & (b,a) \succ_a (c,a) \succ_a (r,a) \\
                & (a,b) \succ_b (d,b) \succ_b (r,b)\\
                & (d,c) \succ_c (a,c) \succ_c (r,c)\\
                & (c,d) \succ_d (b,d) \succ_d (r,d).\\
                 \end{align*}
\end{minipage}
\begin{minipage}{0.35\textwidth}
\input{tikz_instance_intro}
\vspace{1mm}
\end{minipage}

We are interested in computing an \emph{optimal arborescence} rooted at $r$, where an arborescence is an acyclic subgraph of~$G$ in which each vertex $v \in V$ has a unique incoming edge. Our notion of optimality is a function of the preferences $(\succ_v)_{v\in V}$ of vertices for their incoming edges. 

Given any pair of arborescences $A$ and $A'$ in $G$, 
we say that $v \in V$ prefers $A$ to $A'$ if $v$ prefers its incoming edge in $A$ to its incoming edge in $A'$, i.e.,
$v$ prefers $A$ to $A'$ if $A(v) \succ_v A'(v)$ where $A(v)$ (resp., $A'(v)$) is $v$'s incoming edge in $A$ (resp., $A'$).   
Let $\phi(A,A')$ be the number of vertices that prefer $A$ to $A'$. 
We say that $A$ is \emph{more popular than} $A'$ if $\phi(A,A') > \phi(A',A)$.

\begin{Definition}
\label{def:pop-arb}
An arborescence $A$ is popular if $\phi(A,A') \ge \phi(A',A)$ for all arborescences $A'$.
\end{Definition}

Our notion of optimality is popularity, in other words, we seek a popular arborescence $A$ in~$G$.
So there is \emph{no} arborescence more popular than $A$, thus $A$ is maximal under the ``more popular than'' relation. 
The ``more popular than'' relation is not transitive and popular arborescences need not always exist. 

\begin{comment}
Consider the following example from \cite{KKMSS20} where $V = \{a, b, c, d\}$ and the preference orders of 
these four vertices on their incoming edges are as follows:

\begin{minipage}[c]{0.55\textwidth}
\centering
		\begin{align*}
                & (b,a) \succ_a (c,a) \succ_a (r,a) \\
                & (a,b) \succ_b (d,b) \succ_b (r,b)\\
                & (d,c) \succ_c (a,c) \succ_c (r,c)\\
                & (c,d) \succ_d (b,d) \succ_d (r,d).\\
                 \end{align*}
\end{minipage}
\begin{minipage}{0.35\textwidth}
\input{tikz_instance_intro}
\vspace{1mm}
\end{minipage}
\end{comment}

Consider the example from \cite{KKMSS20} illustrated above.
The arborescence $A = \{(r,a),(a,b),(a,c),(c,d)\}$ is not popular, since the arborescence 
$A' = \{(r,d),(d,c),(c,a),(a,b)\}$ is more popular. This is because $a$ and $c$ prefer~$A'$ to~$A$, while $d$ prefers~$A$ to~$A'$, 
and $b$ is indifferent between $A$ and $A'$. We can similarly obtain an arborescence $A'' = \{(r,b),(b,a),(b,d),(d,c)\}$ more popular 
than $A'$. It is easy to check that for any arborescence here, there is a more popular arborescence. Therefore this
instance has no popular arborescence.

Consider the above instance without the edge $(r,d)$. Vertex preferences are the same as in the earlier instance, except that vertex $d$
has no third-choice edge. It can be shown that this instance has two popular arborescences: $A = \{(r,a),(a,b),(a,c),(c,d)\}$ and 
%$A'' = \{(r,b),(b,a),(b,d),(d,c)\}$ 
$A''' = \{(r,b),(b,a),(a,c),(c,d)\}$ (Appendix~\ref{app:examples} has more details).

\medskip

\paragraph{The popular arborescence problem.}
Given a directed graph $G$ as described above, the \myproblem{popular arborescence} problem is 
to determine if $G$ admits a popular arborescence or not, and to find one, if so. 
The computational complexity of the popular arborescence problem was posed as an open problem at the Eml\'ekt\'abla
workshop~\cite{Kir19} in 2019 and the problem has remained open till now. Thus it is an intriguing open problem---aside from its mathematical interest and curiosity, it has applications in \emph{liquid democracy}, which is a voting scheme that allows a voter to delegate its vote to another voter.\footnote{A vertex $v$ delegating its vote to $u$ should be represented as the edge $(v,u)$; however as said in \cite{KKMSS20}, it will be more convenient to denote this delegation by $(u,v)$ so as to be consistent with downward edges in an arborescence.} 

\medskip

\paragraph{Popular branchings.} A special case of the popular arborescence problem is the \myproblem{popular branching} problem.
A branching is a directed forest in a digraph $G = (V, E)$ where %
each vertex has at most one incoming edge. Any branching in $G$ can be viewed as an arborescence in an auxiliary 
graph obtained by augmenting $G$ with a new vertex $r$ as the root and
adding the edge~$(r,v)$ for each $v \in V$ as the least-preferred incoming edge of~$v$.
%making $r$ an in-neighbor of every $v \in V$.  
%For any $v \in V$, let $(r,v)$ be its least preferred incoming edge.
So the problem of deciding whether the given instance $G$ admits a \myproblem{popular branching} or not 
reduces to the problem of deciding whether this auxiliary instance admits a \myproblem{popular arborescence} or not. An 
efficient algorithm for this special case of the \myproblem{popular arborescence} problem (where the root $r$ is an in-neighbor of every 
$v \in V$) was given in \cite{KKMSS20}. 

The applications of popular branchings in liquid democracy were discussed in \cite{KKMSS20}---as mentioned above, each voter can
delegate its vote to another voter; however delegation cycles are forbidden. A popular branching $B$ represents a cycle-free delegation process 
that is stable, and every root in~$B$ casts a weighted vote on behalf of all its descendants.
As mentioned in \cite{KKMSS20}, liquid democracy has been used for internal decision making at Google \cite{HL15} and political parties
such as the German {\em Pirate Party} or the Swedish party {\em Demoex}. We refer to \cite{sch22} for more details.

However, in many real-world applications, not all agents would be willing to be representatives, i.e., to be roots in a branching. 
Thus it cannot be assumed that {\em every} 
vertex is an out-neighbor of $r$, so it is only agents who are willing to be representatives that are out-neighbors of $r$ in our instance. 
Thus the \myproblem{popular arborescence} problem has to be solved in a general digraph $G = (V\cup\{r\}, E)$ rather than in one where every vertex is an 
out-neighbor of $r$. As mentioned earlier, the computational complexity of the \myproblem{popular arborescence} problem was open till now. 
We show the following result.

\begin{theorem}
    \label{thm:pop-arb}
    Let $G = (V \cup \{r\}, E)$ be a directed graph where each $v \in V$ has a partial order 
    % wo->po
    %weak ranking
    over its incoming edges. 
    There is a polynomial-time algorithm to solve the \myproblem{popular arborescence} problem in $G$.
\end{theorem}

%\noindent
\paragraph{Popular matchings and assignments.} The notion of popularity has been extensively studied in the domain of bipartite
matchings where vertices on one side of the graph have weak rankings 
% wo->po
(i.e., linear preference order with possible ties) over their neighbors.
The \myproblem{popular matching} problem is to decide if such a bipartite graph admits a \emph{popular matching}, i.e., 
a matching $M$ such that there is no matching more popular than $M$.

An efficient algorithm for the \myproblem{popular matching} problem was given almost 20 years ago~\cite{AIKM07}. Very recently (in 2022), 
the \myproblem{popular assignment} problem was considered~\cite{KKMSS22}. What is sought in this problem is a perfect matching that is popular 
within the set of perfect matchings---so the cardinality of the matching is more important than popularity here. 
It is easy to see that the \myproblem{popular assignment} problem is a generalization of the \myproblem{popular matching} problem
(a simple reduction from the \myproblem{popular matching} problem to the \myproblem{popular assignment} problem can be shown by adding some dummy vertices).
An efficient algorithm for the \myproblem{popular assignment} problem was given in \cite{KKMSS22}. 

\medskip

\paragraph{Popular common base problem.}
Observe that the \myproblem{popular arborescence} and \myproblem{popular assignment} problems are special cases of
the \myproblem{popular common base} problem in the intersection 
of two matroids, where one matroid is the partition matroid defined by any partition $E = \bigcupdot_{v\in V} \delta(v)$ and the other is an arbitrary matroid $M = (E, {\cal I})$  of rank $|V|$, and each $v \in V$ has a 
%weak ranking 
% wo->po
partial order $\succ_v$ over elements in $\delta(v)$. 
\begin{itemize}
    \item For any pair of common bases (i.e., common maximal independent sets) $I$ and $I'$ in the matroid intersection, we say that $v \in V$ 
    prefers $I$ to $I'$ if $v$ prefers the element in $I\cap\delta(v)$ to the element in $I'\cap\delta(v)$, i.e., $e \succ_v f$ where
    $I\cap\delta(v) = \{e\}$ and $I'\cap\delta(v) = \{f\}$. Let $\phi(I,I')$ be the number of vertices in~$V$ 
    that prefer $I$ to $I'$. The set $I$ is popular within the set of common bases if $\phi(I,I') \ge \phi(I',I)$ 
    for all common bases $I'$.
\end{itemize}
Arborescences are the common bases in the intersection of a partition matroid with a graphic matroid
(for any edge set $I \subseteq E$, $I\in\I$ if and only if $I$ has no cycle in the underlying undirected graph)
while assignments are common bases in the intersection of two partition matroids. %
In fact, our algorithm and the proof of correctness for Theorem~\ref{thm:pop-arb} work for the general \myproblem{popular common base} problem. 
\begin{theorem}
    \label{thm:pop-largest}
    A popular common base in the intersection of a partition matroid on $E = \bigcupdot_{v\in V} \delta(v)$ with any matroid 
    $M = (E, {\cal I})$  of rank $|V|$ can be computed in polynomial time.
\end{theorem}

%In general, the matroid intersection need not admit common bases, and in such a case, an alternative is a largest common independent set 
%that is popular among all largest common independent sets. This problem can be easily reduced to the popular common base problem 
%(see Appendix~\ref{sec:discussion}).

Interestingly, the \myproblem{popular common independent set} problem which asks for a common independent set that is popular 
in the set of all common independent sets (of all sizes) in the matroid intersection can be reduced to the \myproblem{popular common base} problem (see Section~\ref{sec:colorful}). 
Therefore, the following fact is obtained as a corollary to Theorem~\ref{thm:pop-largest}.
\begin{corollary}
    \label{cor:pop-common}
    A popular common independent set in the intersection of a partition matroid on $E = \bigcupdot_{v\in V} \delta(v)$ 
    with any matroid $M = (E, {\cal I})$ %
    can be computed in polynomial time.
\end{corollary}

All of the following problems fall in the framework of a popular common base (or common independent set) in the 
intersection of a partition matroid with another matroid: 
\begin{enumerate}
    \item Popular matchings~\cite{AIKM07}.
    \item Popular assignments~\cite{KKMSS22}.
    \item Popular branchings~\cite{KKMSS20}.
    \item Popular matchings with matroid constraints\footnote{This problem asks for a popular many-to-one matching in a bipartite graph $G = (A \cup B, E)$ where vertices in $A$ have weak rankings and the vertices that get matched to each $b \in B$ must form an independent set in a matroid $M_b$.\label{footnote:one}}~\cite{Kam17}.
\end{enumerate}
%(i)~popular matchings~\cite{AIKM07},
%(ii)~popular assignments~\cite{KKMSS22}, (iii)~popular branchings~\cite{KKMSS20}, and (iv)~popular matchings with matroid constraints\footnote{This problem asks for a popular many-to-one matching in a bipartite graph $G = (A \cup B, E)$ where vertices in $A$ have weak rankings and the vertices that get matched to each $b \in B$ must form an independent set in a matroid $M_b$.\label{footnote:one}}~\cite{Kam17}.
%
% wo->po
Since Corollary~\ref{cor:pop-common} holds for partial order preferences, it generalizes the tractability result in~\cite{Kam17} which assumes that preferences are weak rankings (note that the results in~\cite{Kam17} are based on the paper \cite{AIKM07}, which in turn strongly relies on weak rankings).
%for both the \myproblem{popular branching} problem in~\cite{KKMSS20} and for the \myproblem{popular assignment} problem in~\cite{KKMSS22}.
%
There are other problems that fall in this framework, e.g., %
the \myproblem{popular colorful forest} problem and the \myproblem{popular colorful spanning tree} problem---these 
are natural generalizations of the \myproblem{popular branching} problem and \myproblem{popular arborescence}
problem, respectively. The \myproblem{popular colorful forest} problem and \myproblem{popular color spanning tree} problem 
are new problems introduced in this paper.

\medskip

\paragraph{Popular colorful forests and popular colorful spanning trees.}
The input here is an undirected graph~$G$ where each edge has a color in 
$\{1,\ldots,n\}$. A forest $F$ is \emph{colorful} if each edge in $F$ has a distinct color. 
Colorful forests are the common independent sets of the partition matroid defined by color classes and the graphic matroid of $G$.
For each $i \in \{1,\ldots,n\}$, we assume
there is an agent~$i$ with a 
%weak ranking 
% wo->po
partial order $\succ_i$ over color~$i$ edges. 
Agent~$i$ prefers forest $F$ to forest $F'$ if either (i)~$F$ contains an edge colored $i$ while $F'$ has no edge 
colored $i$ or (ii)~both $F$ and $F'$ contain color~$i$ edges and $i$ prefers the color~$i$ edge in $F$ to the color~$i$ edge in $F'$.

A colorful forest $F$ is popular if $\phi(F,F') \ge \phi(F',F)$ for all colorful forests $F'$,
where $\phi(F,F')$ is the number of agents that prefer $F$ to $F'$. 
The \myproblem{popular colorful forest} problem is to decide if a given graph $G$ admits a popular colorful forest or not, and
to find one, if so. The motivation here is to find an optimal {\em independent} network (cycles are forbidden) with diversity,
i.e., there is at most one edge from each color class---as before, our definition of optimality is popularity.
The \myproblem{popular branching} problem is a special case of the \myproblem{popular colorful
forest} problem where all edges entering vertex~$i$ are colored $i$. 

A colorful spanning tree is a colorful forest with exactly one component.
In the \myproblem {popular colorful spanning tree} problem, \emph{connectivity} is more important than popularity, and we seek popularity
within the set of colorful spanning trees rather than popularity within the set of all colorful forests.

\medskip

\paragraph{Implications of Theorem~\ref{thm:pop-largest}.}
Along with the popular arborescence problem, our algorithm also solves the problems considered in \cite{AIKM07,Kam17,KKMSS20,KKMSS22}; furthermore, it also solves 
the \myproblem{popular colorful forest} and \myproblem{popular colorful spanning tree} problems.
The algorithms given in \cite{Kam17,KKMSS20,KKMSS22} for solving their respective problems are quite different from each other. Thus 
our algorithm provides a unified framework for all these problems and shows that there is one 
polynomial-time algorithm that solves all of them. 

In general, the matroid intersection need not admit common bases, and in such a case, an alternative is a largest common independent set 
that is popular among all largest common independent sets. This problem can be easily reduced to the popular common base problem 
(see \switch{the full version}{Appendix~\ref{sec:discussion}}).
Furthermore, %
along with some simple reductions, we can use our popular common base algorithm to find a popular solution under %
certain constraints. 

For example, we can find a common independent set that is popular subject to a size constraint (if a solution exists). We can further solve the problem under a category-wise size constraint:
%; this means that 
consider a setting where the set $V$ of voters is partitioned into categories, and for each category, there are lower and upper bounds on the number of voters who 
%are assigned to (roughly speaking) {\em real elements} 
(roughly speaking)
have an element in the chosen independent set belonging to them
(see \switch{the full version}{Appendix~\ref{sec:discussion}}).
In the liquid democracy %
application mentioned earlier, this %
translates to setting lower and upper bounds on the number of representatives taken from each category so as to ensure that there is diversity among
representatives.

\medskip

\paragraph{Popular common independent set polytope.}
% wo->po
If preferences are weak rankings, then we also give a formulation of an extension of the {\em popular common independent set polytope}, i.e.,
the convex hull of incidence vectors of popular common independent sets in our matroid intersection.

\begin{theorem}
  \label{thm:polytope}
  % wo->po
  If preferences are weak rankings, the popular common independent set polytope is a projection of a face of the matroid intersection polytope.
\end{theorem}

There are an exponential number of constraints in this formulation, however it admits an efficient separation oracle.
As a consequence, when there is a function $\cost: E \rightarrow \mathbb{R}$, a min-cost popular common independent set can be
computed in polynomial time by optimizing over this polytope, 
% wo->po
assuming that preferences are weak rankings.  
Unfortunately, such a result does not hold for the min-cost popular arborescence problem.

\begin{theorem}
    \label{thm:min-cost}
    Given an instance $G = (V \cup \{r\},E)$ of the popular arborescence problem where each vertex has a strict ranking over its incoming edges along with a function $\cost: E \rightarrow \{0,1,\infty\}$, it is $\NP$-hard to compute a min-cost
    popular arborescence in $G$. 
\end{theorem}
Nevertheless, finding a popular arborescence with forced/forbidden edges in an input instance with %weak rankings
% wo->po
partial order preferences
is polynomial-time solvable. This result allows us to recognize in polynomial time all those edges that are present in 
every popular arborescence and all those edges that are present in {\em no} popular arborescence.

\begin{theorem}
  \label{thm:forced-forbidden}
  For any instance $G = (V \cup \{r\},E)$ of the popular arborescence problem with a set $E^+ \subseteq E$ of forced edges
  and a set $E^- \subseteq E$ of forbidden edges, there is a polynomial-time algorithm to decide if there is a popular 
  arborescence $A$ with $E^+ \subseteq A$ and $E^- \cap A = \emptyset$ and to find one, if so.
\end{theorem}

In instances where a popular arborescence does not exist, we could relax popularity to {\em near-popularity} or ``low unpopularity''. 
A standard measure of unpopularity is the \emph{unpopularity margin}~\cite{McCutchen}, defined for any arborescence~$A$ as  $\mu(A)=\max_{A'} \phi(A',A)-\phi(A,A')$ where the maximum is taken over all arborescences~$A'$.
An arborescence~$A$ is popular if and only if $\mu(A)=0$. Unfortunately, finding an arborescence with minimum unpopularity margin is $\NP$-hard. 

\begin{theorem}
    \label{thm:min-unpop-margin}
    Given an instance $G = (V \cup \{r\},E)$ of the popular arborescence problem where each vertex has a strict ranking over its incoming edges, together with an integer~$k$, it is $\NP$-complete to decide whether $G$ contains an arborescence with unpopularity margin at most~$k$. 
\end{theorem}

\subsection{Background.}
\label{sec:background}
The notion of popularity was introduced by G{\"a}rdenfors~\cite{Gar75} in 1975 in bipartite graphs with two-sided strict preferences. 
In this model every stable matching~\cite{GS62} is popular, thus popular matchings always exist in this setting.
When preferences are \emph{one-sided}, popular matchings need not always exist. This is not very surprising 
given that popular solutions correspond to (weak) Condorcet winners~\cite{Con85,condorcet} and it is well-known in social choice theory
that such a winner need not exist. 

For the case when preferences are weak rankings, a combinatorial characterization of popular matchings was given in \cite{AIKM07} and this yielded an efficient algorithm to solve
the \myproblem{popular matching} problem in this case.
Note that the characterization in~\cite{AIKM07} does not generalize to partial order preferences, as argued in~\cite{KKMSS22arxiv}. 
Several extensions of the \myproblem{popular matching} problem have been considered such as random popular matchings~\cite{Mah06},
weighted voters~\cite{Mestre14}, capacitated objects~\cite{SM10}, popular mixed matchings~\cite{KMN09},
and popularity with matroid constraints~\cite{Kam17}. We refer to \cite{Cseh17} for a survey on results in popular matchings.

Popular spanning trees were studied in \cite{Darm13b,Darm16,DKP11a} where the incentive was to find a ``socially best'' spanning tree. 
However, in contrast to the popular colorful spanning tree problem, 
edges have no colors in their model and voters have rankings over the entire edge set. Many different ways to compare a pair of trees were studied
here, and most of these led to hardness results. 
Popular branchings, i.e., popular directed forests, in a directed graph %$G$ 
(where each vertex has
preferences as a partial order over its incoming edges) were studied in \cite{KKMSS20} where a polynomial-time algorithm was given for the \myproblem{popular branching} problem.
When preferences are weak rankings, polynomial-time algorithms for the min-cost popular branching problem and the $k$-unpopularity margin branching problem
were shown 
in~\cite{KKMSS20}; however these problems were shown to be $\NP$-hard for partial order preferences.
The popular branching problem where each vertex (i.e., voter) has a weight was considered in \cite{NT23}.

The \myproblem{popular assignment} algorithm from \cite{KKMSS22} 
solves the popular maximum matching problem in a bipartite graph, and works for partial order preferences.
It was also shown in~\cite{KKMSS22}
that the min-cost popular assignment problem is $\NP$-hard, even for strict rankings.

Many combinatorial optimization problems can be expressed as (largest) common independent sets in the intersection of two matroids. 
Interestingly, constraining one of the two matroids in the matroid intersection to be a partition matroid is not really
a restriction, because 
%as shown in \cite{HKL11}, 
any matroid intersection can be reduced to the case where one matroid is a partition matroid (see \cite[claims 104--106]{edmonds1970submodular}).
We refer to \cite{Goemans,Schrijver} for notes on matroid intersection and for the formulation of the matroid intersection polytope.

\subsection{An overview of our algorithm.}

For an arborescence $A$, we can naturally define a weight function $\wt_A: E \rightarrow \{-1, 0, 1\}$ such that for any arborescence $A'$ we have $\wt_A(A') = \phi(A',A) - \phi(A,A')$. 
Then a popular arborescence~$A$ is a max-weight arborescence in $G = (V\cup\{r\}, E)$ with this function $\wt_A$. Therefore, the \myproblem{popular arborescence} problem is the problem of finding $A\in{\cal A}_G$ such that $\max_{A'\in{\cal A}_G}\wt_A(A') = \wt_A(A) = 0$ where
${\cal A}_G$ is the set of all arborescences in~$G$.
Thus a popular arborescence $A$ is an optimal solution to the  max-weight arborescence LP with edge weights given by $\wt_A$.

\medskip

%\noindent{\bf Dual certificates.}
\paragraph{Dual certificates.}
We show that every popular arborescence $A$ has a dual certificate with a special structure; this corresponds to a \emph{chain} 
$\C =  \{ C_1,\ldots,C_p \}$ of subsets of $E$ with $\emptyset \subsetneq C_1 \subsetneq \cdots \subsetneq C_p = E$ and $\spn(A \cap C_i) = C_i$ for all $i$.\footnote{In the arborescence case, the set $\spn(A\cap C_i)$ is defined as $(A\cap C_i)\cup \set{e\in E: \text{$(A\cap C_i)+e$ contains a cycle}}$. 
}
Our algorithm to compute a popular arborescence is a search for such a chain $\C$ and arborescence $A$.
At a high level, this method is similar to the approach used in \cite{KKMSS22} for popular assignment, however our dual certificates are more complex
than those in \cite{KKMSS22}, and hence the steps in our algorithm (and its proof of correctness) become much more challenging.

Given a chain $\C$ of subsets of $E$, there is a polynomial-time algorithm to check if $\C$ corresponds to a dual certificate for some popular arborescence.
It follows from dual feasibility and complementary slackness that $\C$ is a dual certificate if and only if a certain subgraph $G_{\C} = (V\cup\{r\}, E(\C))$ 
admits an arborescence $A$ such that $\spn(A \cap C_i) = C_i$ for all $C_i \in \C$. 
If such an arborescence~$A$ exists in~$G_{\C}$, then it is easy to
show that $A$ is a popular arborescence in $G$ with $\C$ as its dual certificate. 

If $G_{\C}$ does not admit such an arborescence, then we need to update~$\C$. 
Since updating $\C$ changes~$E(\C)$,  we now seek an arborescence $A$ in the new graph~$G_{\C}$ such that $\spn(A \cap C_i) = C_i$
for all~$i$. If such an $A$ does not exist, then $\C$ is updated again. Note that updating $\C$ may increase $|\C|$.
When $|\C|$ becomes larger than $|V|$, we claim that $G$ has \emph{no} popular arborescence. 
Among other ideas, our technical novelty lies in the proof of this claim that is based on the strong exchange property of matroids.

\medskip

%\noindent{\bf Matroid Intersection.}
\paragraph{Matroid Intersection.}
Our algorithm holds in the generality of matroid intersection (where one of the matroids is a partition matroid); dual certificates for  
\myproblem{popular common bases} are exactly the same, i.e., chains that are described above.
We also show that a \myproblem{popular common independent set} has a dual certificate $\C = \{ C,E \}$ of length at most~2.
This leads to the polyhedral result given in Theorem~\ref{thm:polytope}. 

Our algorithm is quite different from the popular branching algorithm~\cite{KKMSS20} that (loosely speaking) first finds a maximum branching 
on {\em best} edges and then augments this branching with {\em second best} edges entering certain vertices. 
Indeed, as seen in Theorem~\ref{thm:polytope},
popular branchings or
popular common independent sets have a significantly simpler structure than popular common bases---the latter seem far tougher to characterize 
and analyze. Pleasingly, as we show here, there is a clean and compact algorithm to solve the \myproblem{popular common base} problem
 (see Algorithm~\ref{alg:pop-arb}).

For the sake of readability, we will describe our results for the \myproblem{popular common base} problem in terms of
the \myproblem{popular arborescence} problem and our results for the \myproblem{popular common independent set} problem
in terms of the \myproblem{popular colorful forest} problem.

\medskip

%\noindent{\bf Organization of the paper.}
\paragraph{Organization of the paper.} 
The rest of the paper is organized as follows. Section~\ref{sec:dual} describes dual certificates for popular arborescences. Section~\ref{sec:algo} presents
the popular arborescence algorithm and its proof of correctness. In Section~\ref{sec:colorful}, we discuss popular colorful forests and their polytope.
Section~\ref{sec:forbidden} provides the algorithm for the popular arborescence problem with forced/forbidden edges.

\switch{
Section~\ref{sec:hardness-mincost} shows the $\NP$-hardness of the min-cost popular arborescence problem and we refer to the
full version of our paper for the proof of Theorem~\ref{thm:min-unpop-margin}.}
{Our hardness results are proved in Sections~\ref{sec:hardness-mincost} and \ref{sec:hardness-min-unpop-margin}.}
\switch{
In Appendix~\ref{app:examples}, we present various examples and explain how our algorithm works on them.}% SODA version.
{Appendices~\ref{app:examples}, \ref{app:proofs}, and \ref{sec:discussion} respectively present examples of executions of our algorithm, omitted proofs, and extensions and related results.}

\section{Dual Certificates}\label{sec:dual}

In this section we show that every popular arborescence has a special dual certificate---this will be crucial 
in designing our algorithm in Section~\ref{sec:algo}.
Our input is a directed graph $G = (V \cup \{r\}, E)$ where the root vertex $r$ has no incoming edge, and
every vertex $v \in V$ has a 
% wo->po
%weak ranking 
partial order $\succ_v$ over its set of incoming edges, denoted by $\delta(v)$.
%; let $\delta(v)$ be the set of~$v$'s incoming edges.
For edges $e,f \in \delta(v)$, we write $e \sim_v f$ to denote that $v$ is indifferent between~$e$ and~$f$, i.e., $e \not\succ_v f$ and $f \not\succ_v e$.

Given an arborescence $A$, there is a simple method (as shown in \cite{KKMSS20}) to check if $A$ is popular or not.
We need to check that $\phi(A,A') \ge \phi(A',A)$ for all arborescences $A'$ in $G$.
For this, we will use the following function $\wt_A: E\to \{-1,0,1\}$. %

For any $v \in V$, let $A(v)$ be the unique edge in $A\cap \delta(v)$. For any $v \in V$ and $e\in \delta(v)$, let 
\[  \wt_A(e)=\begin{cases}
    \phantom{-} 1 & \text{if $e\succ_v A(v)$} \ \ \ \text{($v$ prefers $e$ to $A(v)$)};\\
    \phantom{-} 0&\text{if $e\sim_v A(v)$} \ \ \ \text{($v$ is indifferent between $e$ and $A(v)$)};\\
    -1 & \text{if $e\prec_v A(v)$} \ \ \ \text{($v$ prefers $A(v)$ to $e$)}.
\end{cases}\]

    It immediately follows from the definition of $\wt_A$ that
    we have $\wt_A(A') = \phi(A',A) - \phi(A,A')$ for any arborescence $A'$ in $G$. Thus $A$ is popular if and only if every arborescence in $G$
    has weight at most 0, where edge weights are given by $\wt_A$.

Consider the linear program problem \ref{LP1} below.  The constraints of \ref{LP1} describe the face of the matroid intersection polytope
corresponding to common bases. Recall that this is the intersection of the partition matroid 
on~$E = \bigcupdot_{v\in V} \delta(v)$ with the graphic matroid $M = (E, {\cal I})$ of $G$, whose rank is~$|V|$.
Here, $\rank:2^E\to \Zp$ is the rank function of $(E, \I)$, i.e, for any $S\subseteq E$, the value of $\rank(S)$  
is the maximum size of an acyclic subset of~$S$ in the graph $G$.

\begin{table}[H]
\begin{minipage}[t]{0.45\linewidth}\centering
\begin{align}
\label{LP1}
 \max \sum_{e \in E} \wt_A(e)\cdot x_e  &\mbox{\hspace*{0.1in}}\tag{LP1}\\   \notag
      \text{s.t.}\qquad\sum_{e \in \delta(v)}x_e & \, = \, 1  \mbox{\hspace*{0.5in}}\forall\, v \in V\\ \notag
                  \qquad \sum_{e\in S} x_e & \, \leq \, \rank(S) \mbox{\hspace*{0.1in}}\forall\, S \subseteq E\\ \notag
                        x_e  & \, \ge \, 0   \mbox{\hspace*{0.55in}}\forall\, e \in E. 
\end{align}
\end{minipage}
\hspace{0.4cm}
\begin{minipage}[t]{0.45\linewidth}\centering
  \begin{align}
  \label{LP2}
\quad\min \sum_{S \subseteq E}\rank(S)\cdot y_S & + \ \sum_{v\in V}\alpha_v  \mbox{\hspace*{0.1in}}\tag{LP2}\\ \notag
       \quad\text{s.t.}\quad ~\sum_{S: e\in S} y_S+\alpha_v & \, \ge \, \wt_{A}(e)  \mbox{\hspace*{0.1in}}\forall\,  e \in \delta(v),
       \forall v\in V\\ 
       \notag 
                  \qquad  y_S & \, \ge \, 0 \mbox{\hspace*{0.5in}}\forall\, S \subseteq E.
\end{align}
\end{minipage}
\end{table}

The feasible region of \ref{LP1} is the arborescence polytope of $G$. Hence \ref{LP1} is the max-weight arborescence LP 
in $G$ with edge weights given by $\wt_A$. The linear program %
\ref{LP2} is the dual LP in variables $y_S$ and $\alpha_v$ where $S \subseteq E$ and $v \in V$.

The arborescence $A$ is popular if and only if the optimal value of \ref{LP1} is at most $0$, more precisely, if the optimal value 
is exactly $0$, since $\wt_A(A) = 0$. Equivalently, $A$ is popular if and only if the optimal value of \ref{LP2} is~$0$. 
We will now show that \ref{LP2} has an optimal solution with some special properties.  For a popular arborescence~$A$,
a dual optimal solution that satisfies all these special properties (see Lemma~\ref{lem:duality}) will be called 
a {\em dual certificate} for~$A$.

The function $\spn:2^E\to 2^E$ of a matroid $(E, \I)$ is defined as follows:
\[\spn(S)=\set{e\in E: \, \rank(S+e)=\rank(S)}\quad \text{where}\ S\subseteq E.\]
In particular, if $S\in \I$, then $\spn(S)=S\cup \set{e\in E: S+e\not\in \I}$.

A {\em chain} $\C$ of length $p$ is a collection of $p$ distinct subsets of $E$ such that for each two distinct sets $C, C'\in \C$, we have either $C\subsetneq C'$ or $C'\subsetneq C$. 
That is, a chain has the form $\C = \{C_1, C_2,\dots, C_p\}$ where $C_1\subsetneq C_2\subsetneq \cdots \subsetneq C_p$.

Lemma~\ref{lem:duality} shows that \ref{LP2} always 
admits an optimal solution in the following special form.
The proof is based on basic facts on matroid intersection and linear programming, and we postpone it to the end of Section~\ref{sec:algo}. 
%and the fact that $\wt_A$ takes integer values at most $1$.

\begin{restatable}{lemma}{lempropduality}%[$\star$]
    \label{lem:duality}
    An arborescence $A$ is popular if and only if there exists a feasible solution $(\vec{y},\vec{\alpha})$ to \ref{LP2} 
    such that $\sum_{S\subseteq E} \rank(S)\cdot y_S+\sum_{v\in V}\alpha_v = 0$ and properties~1--4 are satisfied:
    \begin{enumerate}
        \item $\vec{y}$ is integral and its support $\C:=\set{S\subseteq E: y_S> 0}$ is a chain.
        \item Each $C\in \C$ satisfies $\spn(A\cap C)=C$.
        \item Every element in $\C$ is nonempty, and the maximal element in $\C$ is $E$.
        \item For each $C\in \C$, we have $y_{C}=1$.
        For each $v\in V$,  we have $\alpha_v = -|\set{C\in \C: A(v)\in C}|$.
    \end{enumerate}
\end{restatable}

For any chain $\C$, we will now define a subset $E(\C)$ of $E$ that will be used in our algorithm. The construction of $E(\C)$ is inspired by the construction of an analogous edge subset in the popular assignment algorithm~\cite{KKMSS22}.

For a chain  $\C=\{C_1, C_2,\cdots, C_p\}$ with $\emptyset\subsetneq C_1 \subsetneq \dots \subsetneq C_p=E$,  
define
\begin{align*}
&\lev_{\C}(e)=\text{the index $i$ such that $e\in C_i\setminus C_{i-1}$ }\quad~ \text{for any}\ e\in E,\\
&\lev^*_{\C}(v)=\max\set{\lev_{\C}(e): e\in \delta(v)}\qquad\qquad\qquad \text{for any}\ v\in V,
\end{align*}
where we let $C_0=\emptyset$. 
Thus every element in $E$ has a {\em level} in $\{1,\ldots,p\}$ associated with it, which
is the minimum subscript $i$ such that $e \in C_i$ (where $C_i\in \C$). 
Furthermore, each $v \in V$ has a $\lev^*_{\C}$-value which is the highest level of any element in $\delta(v)$.

Define $E(\C)\subseteq E$ as follows. For each $v \in V$, an element $e\in \delta(v)$ belongs to $E(\C)$ if one of the following two conditions holds:
\begin{itemize}
%\item $\lev_{\C}(e)=\lev^*_{\C}(v)$ and $e\succeq_v e'$ for every $e' \in \delta(v)$ such that $\lev_{\C}(e')=\lev^*_{\C}(v)$;
\item $\lev_{\C}(e)=\lev^*_{\C}(v)$ and there is no element $e' \in \delta(v)$ such that $\lev_{\C}(e')=\lev^*_{\C}(v)$ and $e'\succ_v e$;
%\item $\lev_{\C}(e)=\lev^*_{\C}(v)-1$ and $e\succeq_v e'$ for every $e' \in \delta(v)$ such that $\lev_{\C}(e')=\lev^*_{\C}(v)-1$, and moreover, $e\succ_v f$ for every $f \in \delta(v)$ with $\lev_{\C}(f)=\lev^*_{\C}(v)$.
\item $\lev_{\C}(e)=\lev^*_{\C}(v)-1$ and there is no element $e' \in \delta(v)$ such that $\lev_{\C}(e')=\lev^*_{\C}(v)-1$ and $e'\succ_v e$, and moreover, $e\succ_v f$ for every $f \in \delta(v)$ with $\lev_{\C}(f)=\lev^*_{\C}(v)$.

\end{itemize}
In other words,  $e \in \delta(v)$ belongs to $E(\C)$ if either (i)~$e$ is a maximal element in $\delta(v)$ with respect to $\succ_v$ 
among those in $\lev^*_{\C}(v)$ or (ii)~$e$ is a maximal element in $\delta(v)$ 
among those in $\lev^*_{\C}(v) - 1$ and $v$ strictly prefers $e$ to all elements in level $\lev^*_{\C}(v)$.
From Lemma~\ref{lem:duality}, we obtain the following useful characterization of popular arborescences.

\begin{restatable}{lemma}{lemchain}%[$\star$]
\label{lem:chain}
An arborescence  $A$ is popular if and only if there exists a chain
 $\C=\{C_1,\ldots, C_p\}$ such that 
 $\emptyset \subsetneq C_1 
 \subsetneq \dots \subsetneq C_p =E$, $A\subseteq E(\C)$,  and 
 $\spn(A \cap C_i) = C_i$ for all $C_i \in \C$.
\end{restatable}
The proof is given below. Recall that for a popular arborescence $A$, we defined its {\em dual certificate} as a dual optimal solution $(\vec{y}, \vec{\alpha})$ to \ref{LP2} that satisfies properties 1--4 in Lemma~\ref{lem:duality}.
As shown in the proof of Lemma~\ref{lem:chain}, we can obtain such a solution $(\vec{y}, \vec{\alpha})$ from a chain satisfying the properties in Lemma~\ref{lem:chain}. We therefore will also use the term {\em dual certificate} to refer to a chain as described in Lemma~\ref{lem:chain}.

\begin{myproof}{Proof of Lemma~\ref{lem:chain}}  % moved from Appendix
We first show the existence of a desired chain $\C$ for a popular arborescence $A$.
Since $A$ is popular,  
we know from Lemma~\ref{lem:duality} that there exists an optimal solution $(\vec{y}, \vec{\alpha})$ to \ref{LP2} such that
properties~1--4 hold, where $\C$ is the support of $y$. 
Since the properties $\emptyset \subsetneq C_1 
 \subsetneq \dots \subsetneq C_p =E$ and $\spn(A \cap C_i) = C_i~(\forall C_i \in \C)$ directly follow from  properties~3 and 2, respectively, it remains  to show that $A\subseteq E(\C)$.

Since $(\vec{y},\vec{\alpha})$ is a feasible solution of \ref{LP2}, we have $\sum_{S:e\in S}y_S+\alpha_v\geq \wt_A(e)$ 
for every $e\in \delta(v)$ with $v\in V$. By property~4, the left hand side can be expressed as
\[|\set{C_i\in \C:e\in C_i}|-|\set{C_i\in \C:A(v)\in C_i}|=
(p-\lev_{\C}(e)+1)-(p-\lev_{\C}(A(v))+1)=\lev_{\C}(A(v))-\lev_{\C}(e).\]
Thus it is equivalent to the condition that for every $e\in \delta(v)$:
\begin{equation}
\lev_{\C}(A(v))-\lev_{\C}(e)\geq \wt_A(e)=
\begin{cases}
    1 & \text{if $e\succ_v A(v)$};\\
    0&\text{if $e\sim_v A(v)$};\\
    -1 & \text{if $e\prec_v A(v)$.}
\end{cases}\label{eq:ineq}
\end{equation}
In particular, this holds for an edge $e'$ with $\lev_{\C}(e')=\lev^*_{\C}(v)$, 
and hence we have $\lev_{\C}(A(v))\geq \lev^*_{\C}(v)-1$. Since $\lev_{\C}(A(v))\leq \lev^*_{\C}(v)$  by $A(v)\in \delta(v)$, $\lev_{\C}(A(v))$ is either $\lev^*_{\C}(v)$ or $\lev^*_{\C}(v)-1$.
\begin{itemize}
\item If $\lev_{\C}(A(v))=\lev^*_{\C}(v)$, then for any $e\in \delta(v)$ with $\lev_{\C}(e)=\lev^*_{\C}(v)$,
the left hand side of \eqref{eq:ineq} is $0$, and hence it must be the case that either $A(v) \succ_v e$ or 
$A(v) \sim_v e$. 
Hence $A(v)$ is a maximal element 
in~$\set{e\in \delta(v): \lev_{\C}(e)=\lev^*_{\C}(v)}$ with respect to $\succ_v$.
\item If $\lev_{\C}(A(v))=\lev^*_{\C}(v)-1$, then we can similarly show that 
$A(v)$ is a maximal element in the set~$\set{e\in \delta(v): \lev_{\C}(e)=\lev^*_{\C}(v)-1}$ with respect to $\succ_v$.
Furthermore, in this case, for any $e\in \delta(v)$ with $\lev_{\C}(e)=\lev^*_{\C}(v)$, the left hand side of 
\eqref{eq:ineq} is $-1$, and hence $A(v)\succ_v e$ must hold. 
\end{itemize}
Therefore, in either case, we have $A(v)\in E(\C)$, which implies that  $A\subseteq E(\C)$.

\medskip

For the converse, suppose that  ${\cal C}=\{C_1, \dots , C_p\}$  is a chain such that  $\emptyset \subsetneq C_1 
 \subsetneq \dots \subsetneq C_p =E$, $A\subseteq E(\C)$, and $\spn(A \cap C_i)= C_i$ for all $C_i \in \C$. 
Define $\vec{y}$ by $y_{C_i}=1$ for every $C_i\in \C$ and $y_S=0$ for all $S \in 2^S\setminus \C$. 
We also define 
$\vec{\alpha}$ by  $\alpha_v=-|\set{C\in \C :A(v)\in C }|$ for any $v\in V$. 
Then $(\vec{y}, \vec{\alpha})$ satisfies properties~1-4 given in Lemma~\ref{lem:duality}, which also implies that  
the objective value is $0$.  
Thus it is enough to show that $(\vec{y}, \vec{\alpha})$ is a feasible solution to \ref{LP2}, because it implies that  $A$ is a popular arborescence by Lemma~\ref{lem:duality}.
Observe that constraint~\eqref{eq:ineq} is satisfied for every $v\in V$ and $e\in \delta(v)$, which follows from  $A\subseteq E(\C)$.
Since it is equivalent to   the constraint in \ref{LP2} for~$v \in V$ and $e \in \delta(v)$, 
the proof is completed.
\end{myproof}

\section{Our Algorithm}\label{sec:algo}

We now present our main result. The popular arborescence algorithm seeks to construct an arborescence~$A$ along with 
its dual certificate ${\cal C}=\{C_1, \dots , C_p\}$, which is a chain satisfying (i)~$\emptyset\subsetneq C_1 \subsetneq \dots \subsetneq C_p =E$,   (ii)~$A \subseteq E(\C)$,  and (iii)~$\spn(A \cap C_i) = C_i$ for all 
$C_i \in \C$. 

\begin{enumerate}
    \item[-] The existence of such a chain $\C$ means that $A$ is popular by Lemma~\ref{lem:chain}. 
    \item[-] Since a popular arborescence %
need not always exist, the algorithm also needs to detect when a solution does not exist.
\end{enumerate}

The algorithm starts with the chain $\C = \{E\}$ and repeatedly updates it. It always maintains $\C$ as a multichain, where a collection $\C=\{C_1,\cdots, C_p\}$ of indexed subsets of $E$ is called a {\em multichain} if $C_1\subseteq \cdots \subseteq C_p$. 
Note that it is a chain if all the inclusions are strict.
We will use the notations $\lev_\C$, $\lev^*_\C$, and $E(\C)$ also for multichains, which are defined in the same manner as for chains.  

During the algorithm, $\C=\{C_1, \ldots, C_p\}$ is always a multichain with $C_p=E$ and $\spn(C_i)=C_i$ for all $C_i \in \C$.
Note that when $\spn(C_i)=C_i$ holds, the condition (iii) for some~$A$ above is equivalent to $|A\cap C_i|=\rank(C_i)$. Furthermore, as explained later, any multichain can be  modified to a chain that satisfies (i) preserving the remaining conditions (ii) and (iii). Therefore,  we can obtain a desired chain if $|A\cap C_i|=\rank(C_i)$ is attained for all $C_i \in \C$ for some arborescence $A \subseteq E(\C)$ in the algorithm.

\medskip

\paragraph{Lex-maximal branching.}
In order to determine the existence of an arborescence $A\subseteq E(\C)$ that satisfies $|A \cap C_i| = \rank(C_i)$ for all $C_i \in \C$, the algorithm computes a {\em lex-maximal} branching $I$ in $E(\C)$. 
That is, 
$I$ is a branching whose $p$-tuple $(|I \cap C_1|, \ldots,|I\cap C_p|)$ is lexicographically maximum 
among all branchings in~$E(\C)$. If~$(|I \cap C_1|, \ldots,|I\cap C_p|) = (\rank(C_1),\ldots,\rank(C_p))$, then we can show that $I$ is a popular arborescence\footnote{Observe that the branching $I$ will be an {\em arborescence} since $|I\cap E| = |I \cap C_p| = \rank(C_p) = \rank(E) = |V|$.\label{footnote-arb}}; 
otherwise the multichain~$\C$ is updated. We describe the algorithm as Algorithm~\ref{alg:pop-arb}; 
recall that $\rank(E) = |V| = n$.

\begin{algorithm}
\caption{The popular arborescence algorithm}\label{alg:pop-arb}
\begin{algorithmic}[1]
\State Initialize $p =1$ and $C_1 = E$.
\Comment{Initially we set $\C = \{E\}$.}
\While{$p \le n$}
\State Compute the edge set $E(\C)$ from the current multichain $\C$.
\State Find a branching  $I\subseteq E(\C)$ that lexicographically maximizes $(|I\cap C_1|, \dots, |I\cap C_{p}|)$.
\If{$|I\cap C_i|=\rank(C_i)$ for every $i=1,\dots, p$} return $I$. \EndIf
\State Let $k$ be the minimum index such that $|I\cap C_k|<\rank(C_k)$.  
\State Update $C_k \leftarrow \spn(I \cap C_k)$.
\If{$k = p$} $p \leftarrow p+1$, $C_p \leftarrow E$, and $\C \leftarrow \C \cup \{C_p\}$. \EndIf
\EndWhile
\State Return ``$G$ has no popular arborescence''.  
\end{algorithmic}
\end{algorithm}

We include some examples in Appendix~\ref{app:examples} to illustrate the working of Algorithm~\ref{alg:pop-arb} on different input instances. 
% $G = (V\cup\{r\},E)$.
The following observation is important.

\begin{observation}
\label{obs1}
    During Algorithm~\ref{alg:pop-arb}, $\C$ is always a multichain and  $\spn(C_i)=C_i$ for all $C_i\in \C$.
\end{observation}
\begin{proof}
    When $C_k$ is updated, it becomes smaller but the inclusion $C_{k-1}\subseteq C_k$ is preserved.
Indeed, since $|I\cap C_{k-1}|=\rank(C_{k-1})$ by the choice of $k$,  we have $C_{k-1}\subseteq \spn(I\cap C_{k-1})\subseteq \spn(I\cap C_{k})$, for the set~$C_k$ before the update. Hence the updated value for~$C_k$, i.e.,  $\spn(I\cap C_{k})$, is still a superset of $C_{k-1}$, and thus $\C$ remains a multichain. 

Since any $C_i\in \C$ is defined in the form $\spn(X)$ for some $X\subseteq E$ (note that $E=\spn(E)$) and $\spn(\spn(X))=\spn(X)$ holds in general,
we have $\spn(C_i)=C_i$.
\end{proof}

Line~4 can be implemented in polynomial time by a max-weight branching algorithm \cite{Bock71,CL65,Edm67}
and, in the more general case of the intersection of two matroids, by the weighted matroid intersection algorithm~\cite{Frank81}. 
Hence Algorithm~\ref{alg:pop-arb} can be implemented in polynomial time.

\medskip

\paragraph{Correctness of the algorithm.}
Suppose that a branching $I$ is returned by the algorithm. Then $I$ is an arborescence (see Footnote~\ref{footnote-arb})
with $I\subseteq E(\C)$, where $\C$ is the current multichain. 
This implies that  $I \subseteq E(\C)$ and $|I \cap C_i| = \rank(C_i)$
for all $C_i \in \C$, and the latter implies $\spn(I \cap C_i) = C_i$
for all $C_i \in \C$ by Observation~\ref{obs1}.

In order to prove that $I$ is a popular arborescence, let us first prune the multichain $\C$ to a chain~$\C'$,
i.e., $\C'$ contains a single occurrence of each $C_i \in \C$; we will also remove any occurrence of $\emptyset$ from $\C'$.
Observe that $E(\C) \subseteq E(\C')$: indeed, if $C_i = C_{i+1}$ in $\C$, then no element~$e \in E $ can have $\lev_{\C}(e) = i+1$,
and hence no element gets deleted from $E(\C)$ by pruning $C_{i+1}$ from $\C$.
Thus $I \subseteq E(\C) \subseteq E(\C')$. 
This implies that $\C'=\{C'_1,\ldots, C'_{p'}\}$ satisfies 
 $\emptyset \subsetneq C'_1 \subsetneq \dots \subsetneq C'_{p'} =E$, $I\subseteq E(\C')$,  and 
 $\spn(I \cap C'_i) = C'_i$ for all $C'_i \in \C'$.\footnote{In fact, it will turn out that $\C=\C'$, i.e., the final $\C$ obtained by the algorithm itself is a dual certificate of $I$ if the algorithm returns an arborescence $I$. This fact follows from Lemma~\ref{prop:correctness} (with $\C'$ substituted for $\D$).} 
 Hence $I$ is a popular arborescence by Lemma~\ref{lem:chain}. 

We will now show that the algorithm always returns a popular arborescence, if $G$ admits one. 
Let $A$ be any popular arborescence in $G$ and let $\D = \{D_1,\ldots,D_q\}$ be a dual certificate for $A$.
\begin{claim}
    \label{clm:level}
     We have $q \le n$ where $|\D| = q$.
\end{claim}
\begin{proof}
    From the definition of dual certificate, we have $\emptyset \subsetneq D_1\subsetneq \cdots \subsetneq D_q=E$ and $\spn(D_i) = D_i$ for each $D_i$. 
    This implies 
    $0 < \rank(D_1) < \cdots < \rank(D_q)$. Since $\rank(D_q) = \rank(E) =|V|$, 
    we obtain  $q \le |V| = n$.
\end{proof}

The following crucial lemma shows an invariant of the algorithm that holds for the multichain $\C=\{C_1, \ldots, C_p\}$ constructed in the algorithm and a dual certificate $\D = \{D_1,\ldots,D_q\}$ of any popular arborescence~$A$. 
The proof will be given in this section. 
\begin{lemma}
\label{prop:correctness}
\emph{At any moment of Algorithm~\ref{alg:pop-arb}, $p\leq q$ and $D_i\subseteq C_i$ holds for $i=1,\dots, p$.}
\end{lemma}
If $p=n+1$ occurs in Algorithm~\ref{alg:pop-arb}, then Lemma~\ref{prop:correctness} implies $q \ge n+1$. This contradicts Claim~\ref{clm:level}. Hence it has to be the case that 
$G$ has no popular arborescence when $p = n+1$. Thus assuming Lemma~\ref{prop:correctness}, 
the correctness of Algorithm~\ref{alg:pop-arb} follows.

Before we prove Lemma~\ref{prop:correctness}, we need the following claim on $E(\C)$ and $E(\D)$.

\begin{claim}
\label{claim:1}
Assume $p\leq q$ and $D_i\subseteq C_i$ for $i=1,\dots, p$.
For each $e\in E$, if $\lev_{\C}(e)=\lev_{\D}(e)$ and $e\in E(\D)$, then $e\in E(\C)$. 
\end{claim}
\begin{proof} 
Suppose for the sake of contradiction that~$e$ fulfills the conditions of the claim, but $e\not\in E(\C)$. 
Let $e \in \delta(v)$. It follows from the definition of $E(\C)$ that there exists an element
$e'\in \delta(v)$ such that one of the following three conditions holds:
(a) $\lev_{\C}(e')\geq \lev_{\C}(e)+2$,
(b) $\lev_{\C}(e')=\lev_{\C}(e)+1$ and $e\not\succ_v e'$, or %
(c) $\lev_{\C}(e')=\lev_{\C}(e)$ and~$e' \succ_ve$.

Because $D_i\subseteq C_i$ for each $i \in \{1,\dots,p\}$, we have $\lev_{\D}(e')\ge \lev_{\C}(e')$. Since $\lev_{\D}(e)=\lev_{\C}(e)$, the
existence of such an $e' \in \delta(v)$ implies $e\not\in E(\D)$, a contradiction. Thus we have $e\in E(\C)$.
\end{proof}

The proof of Lemma~\ref{prop:correctness} will use the following fact, known as the strong exchange property, that is satisfied by any matroid.\footnote{The original statement in \cite{brualdi1969comments} claims this property only for pairs of bases (maximal independent sets), but it is equivalent to Fact~\ref{fact:exchange}. Indeed, if we consider the $\rank(E)$-truncation of the direct sum of $(E, \I)$ and a free matroid whose rank is $\rank(E)$, then 
the axiom in \cite{brualdi1969comments} applied to this new matroid implies Fact \ref{fact:exchange} for $(E, \I)$.}

\begin{fact}[Brualdi \cite{brualdi1969comments}]\label{fact:exchange}
    For any $X, Y\in \I$ and $e\in X\setminus Y$, if $Y+e\not\in \I$,  then there exists 
    an element 
    $f\in Y\setminus X$ such that 
    $X-e+f$ and $Y+e-f$ are in $\I$.
\end{fact}

Now we provide the proof of Lemma~\ref{prop:correctness}.
As mentioned above, this completes the proof of the correctness of our algorithm, and 
hence we can conclude Theorem~\ref{thm:pop-arb}.
Furthermore, we can conclude Theorem~\ref{thm:pop-largest} since Algorithm~\ref{alg:pop-arb} and its correctness proof hold in the generality of a common base in the intersection of the partition matroid  on the set $E = \bigcupdot_{v\in V} \delta(v)$ with any matroid $M = (E, {\cal I})$  of rank~$|V|$.

\begin{myproof}{Proof of Lemma~\ref{prop:correctness}}
Algorithm~\ref{alg:pop-arb} starts with $\C = \{E\}$. Then the conditions in Lemma~\ref{prop:correctness} %
hold at the beginning. We show by induction that they are preserved through the algorithm.

It is easy to see that the condition $p\leq q$ is preserved. 
Indeed, whenever Algorithm~\ref{alg:pop-arb} is going to increase $p$ (in line~8), it is the case that $p+1\leq q$ because $D_p \subseteq C_p \subsetneq E = D_q$ by the induction hypothesis.
Thus $p \le q$ is maintained in the algorithm. 

We now show that $D_i\subseteq C_i~(i=1,\dots, p)$ is maintained. Note that $\C$ is updated in lines 7 or 8. The update in line 8 (adding $C_p=E$) clearly preserves the condition. We complete the proof by showing that the update in line 7 also preserves the condition, i.e., we show the following statement.
\begin{itemize}
\item Let $\C=\{C_1,\dots, C_p\}$ be a multichain with $C_p=E$ 
such that $p\leq q$ and $D_i\subseteq C_i$ for $i=1,\dots,p$. Suppose the following two conditions hold.
\begin{enumerate}
\item $I$ is a lex-maximal common independent set subject to $I\subseteq E(\C)$.
\item $\spn(I\cap C_i)=C_i$ for $i=1,\dots, k-1$, and $\spn(I\cap C_k)\subsetneq C_k$.
\end{enumerate}
Then $D_k\subseteq \spn(I\cap C_k)$.
\end{itemize}
To show this statement, assume for contradiction that $D_k\not\subseteq \spn(I\cap C_k)$. 

We will first show the existence of distinct elements $e_1$ and $f_1$ such that $e_1, f_1 \in \delta(v_1)$ for some $v_1 \in V$
and~$f_1 \in A \setminus I$ while $e_1 \in I \setminus A$.
Then we will use the pair $e_1,f_1$ to show the existence of another pair $e_2, f_2$ such that
$e_2,f_2 \in \delta(v_2)$ where $f_2 \ne f_1$ and $f_2 \in A\setminus I$ while $e_2 \in I \setminus A$. 
In this manner, for any $t \in \Zp$ we will be able to show {\em distinct} elements $f_1,f_2,\ldots,f_t$ that belong to $A$. 
However $A$ has only $n$ elements, a contradiction. Then we can conclude that our assumption $D_k\not\subseteq \spn(I\cap C_k)$ is wrong. 
The following is our starting claim.

\begin{claim}
    \label{claim-e1-f1}
    There exists $v_1 \in V$ such that there are $e_1, f_1 \in \delta(v_1)$ satisfying the following properties:
    \begin{enumerate}
         \item $f_1 \in A  \setminus I $, 
         %$\lev_{\C}(f_1)=\lev_{\D}(f_1)$, 
         \ $I_1:=(I\cap C_k)+f_1\in\I$,\  $I_1\subseteq E(\C)$, and $\lev_{\C}(f_1)=k$,
        \item $e_1 \in I_1 \setminus A$ and $\lev_{\C}(e_1)=\lev_{\D}(e_1)\leq k$.
       \end{enumerate}
\end{claim}
\begin{proof}
Since $\D$ is a dual certificate of $A$, we have $\spn(A\cap D_k) =D_k$.  
So $D_k\not\subseteq \spn(I\cap C_k)$ implies that $\spn(A\cap D_k)\not\subseteq \spn(I\cap C_k)$.
Hence $A\cap D_k\not\subseteq \spn(I\cap C_k)$. %
So there exists $f_1$ such that $f_1\in A\cap D_k$ and  $f_1\not\in \spn(I\cap C_k)$.

Since $D_k\subseteq C_k$, we have $f_1\in D_k\subseteq C_k$. We also have $D_{k-1}\subseteq C_{k-1}=\spn(I\cap C_{k-1})\subseteq \spn(I\cap C_k)\not\ni f_1$. Hence $f_1\in C_k\setminus C_{k-1}$ and $f_1\in D_k\setminus D_{k-1}$, i.e.,
$\lev_{\C}(f_1)=\lev_{\D}(f_1) = k$.

\smallskip

Since $f_1\in A\subseteq E(\D)$ and $\lev_{\C}(f_1)=\lev_{\D}(f_1)$, we have $f_1\in E(\C)$ by Claim~\ref{claim:1}. 
As $I\subseteq E(\C)$, we then have  $I_1:=(I\cap C_k)+f_1\subseteq E(\C)$. 
Also, $I_1 \in \I$ by $f_1\not\in \spn(I\cap C_k)$. Since $\lev_\C(f_1)=k$, the set $I_1= (I\cap C_k)+f_1$ is lexicographically better than $I$.
Then, the lex-maximality of $I$ implies that $I_1$ must violate the partition matroid constraint, 
i.e., there exists $e_1\in I_1$ such that $e_1 \ne f_1$ and $e_1, f_1\in \delta(v_1)$ for some $v_1\in V$.

We have $\lev_{\C}(e_1)\leq k$ as $e_1\in I_1\setminus\{f_1\}=I\cap C_k$.
Since $f_1\in \delta(v_1)\cap A$ and $|\delta(v_1)\cap A|\leq 1$, we have $e_1\not\in A$.
Note that $f_1\in E(\D)$ implies $\lev_{\D}(f_1)\geq \lev_{\D}(e_1)-1$ and $e_1\in E(\C)$ implies $\lev_{\C}(e_1)\geq \lev_{\C}(f_1)-1$. 
Note also that, for any element $e\in E$, we have $\lev_{\D}(e)\geq \lev_{\C}(e)$ because $D_i \subseteq C_i$ for all $i$.
\begin{itemize}
    \item If $f_1\succ_{v_1} e_1$, then $\lev_{\C}(e_1)>\lev_{\C}(f_1)$ by $e_1\in E(\C)$,%
    \footnote{Actually, the case $f_1\succ_{v_1} e_1$ is impossible because $\lev_{\C}(e_1)>\lev_{\C}(f_1)$ contradicts $\lev_{\C}(e_1)\leq k=\lev_{\C}(f_1)$. We write the proof in this form because the proofs of Claims~\ref{claim-e2-f2} and \ref{claim-next} refer to the argument here to apply it to $e_j, f_j$, where $\lev_{\C}(f_j)=k$ is not assumed.}% 
    and hence  $\lev_{\D}(f_1)\geq \lev_{\D}(e_1)-1\geq \lev_{\C}(e_1)-1\geq \lev_{\C}(f_1)$.  As we have $\lev_{\D}(f_1)=\lev_{\C}(f_1)$, all the equalities hold.

    \item If $e_1\succ_{v_1}f_1$, then $\lev_{\D}(f_1)>\lev_{\D}(e_1)$ by $f_1\in E(\D)$, and hence 
$\lev_{\D}(f_1)\geq \lev_{\D}(e_1)+1\geq \lev_{\C}(e_1)+1\geq \lev_{\C}(f_1)$. 
As we have $\lev_{\D}(f_1)=\lev_{\C}(f_1)$, all the equalities hold.
 
    \item  If $f_1 \sim_{v_1} e_1$, then $\lev_{\C}(e_1) \geq \lev_{\C}(f_1)$ by $e_1\in E(\C)$;
    also $\lev_{\D}(f_1) \ge \lev_{\D}(e_1)$ by $f_1 \in E(\D)$. %and the construction of $E(\D)$. 
    Hence, we have $\lev_{\D}(f_1) \geq \lev_{\D}(e_1) \ge \lev_{\C}(e_1)\geq \lev_{\C}(f_1)$. Since $\lev_{\D}(f_1) = \lev_{\C}(f_1)$, all the equalities hold.
\end{itemize}
Thus in all the cases, we have $\lev_{\C}(e_1)=\lev_{\D}(e_1)\leq k$ and $e_1\in I_1\setminus A$.
\end{proof}

Our next claim is the following. Recall that $I_1:=(I\cap C_k)+f_1 \in \I$.

\begin{claim}
 \label{claim-e2-f2}
    There exists $v_2 \in V$ such that there are $e_2, f_2 \in \delta(v_2)$ satisfying the following properties: 
    \begin{enumerate}
        \item $f_2 \in A\setminus I_1$,
        %\ $\lev_{\C}(f_2)=\lev_{\D}(f_2)$, 
        \ $I_{2}:=I_1-e_1+f_2 \in \I$, \ $I_2\subseteq E(\C)$, and $\lev_{\C}(e_1)=\lev_{\C}(f_2)$, 
        \item $e_2 \in I_2 \setminus A$ and $\lev_{\C}(e_2)=\lev_{\D}(e_2)\leq k$.
    \end{enumerate}
\end{claim}
\begin{proof}
We know from Claim~\ref{claim-e1-f1} that %$f_1 \notin \spn(I \cap C_k)$.
$I_1=(I\cap C_k)+f_1\in \I$.
The set $I_1$ satisfies 
$\spn(I_1\cap C_i)=\spn(I\cap C_i)=C_i$ for each $1\leq i\leq k-1$; this is because $I_1\cap C_i=I\cap C_i$ for each $i\leq k-1$.
Let us apply the exchange axiom in Fact~\ref{fact:exchange} to $I_{1}, A \in \I$ and $e_1\in I_1\setminus A$. 
Since $A$ is maximal in $\I$, we have $A+e_1\not\in \I$, and hence there exists $f_2\in A\setminus I_1$ such that $I_1-e_1+f_2$ and $A+e_1-f_2$ are in $\I$. 

 Using that $\spn(A\cap D_i)=D_i$ for $1\leq i\leq q$, from $e_1 \notin \spn(A-f_2)$ we obtain $\lev_{\D}(f_2)\leq \lev_{\D}(e_1)$: indeed, 
assuming $\lev_{\D}(f_2)=\ell \geq 2$ we get $D_{\ell-1}=\spn(A \cap D_{\ell-1}) \subseteq \spn(A-f_2)$, which implies $e_1 \notin D_{\ell-1}$ and hence also $\lev_{\D}(e_1) \geq \ell = \lev_{\D}(f_2)$.
 Similarly,  from $f_2 \notin \spn(I_1 - e_1)$, $\lev_{\C}(e_1)\leq k$, and  $\spn(I_1\cap C_i)=C_i$ for $1\leq i\leq k-1$, 
we obtain $\lev_{\C}(e_1)\leq \lev_{\C}(f_2)$.
Thus we have
$\lev_{\C}(e_1)\leq \lev_{\C}(f_2)\leq \lev_{\D}(f_2)\leq \lev_{\D}(e_1)=\lev_{\C}(e_1)$,
implying  all the equalities. Hence we have
\[f_2\in A\setminus I_1,\quad \lev_{\C}(f_2)=\lev_{\D}(f_2),\quad \lev_{\C}(e_1)=\lev_{\C}(f_2).\]
As $f_2\in A\subseteq E(\D)$, Claim~\ref{claim:1} implies $f_2\in E(\C)$.

Observe that $I_{2}:=I_1-e_1+f_2=(I\cap C_k)+f_1-e_1+f_2\subseteq E(\C)$, and recall $I_2 \in \I$. Since $\lev_{\C}(e_1)=\lev_{\C}(f_2)$ and $\lev_{\C}(f_1)=k$, $I_2$ is lexicographically better than $I$. This implies that $I_2$ must violate the partition matroid constraint.
By the same argument as used in Claim~\ref{claim-e1-f1} to show $\lev_{\C}(e_1)=\lev_{\D}(e_1)$,  
we see that 
there exists $e_2$ such that $e_2, f_2\in \delta(v_2)$ for some $v_2\in V$, satisfying  
\[e_2\in I_2\setminus A, \quad \lev_{\C}(e_2)=\lev_{\D}(e_2)\leq k.\]
This completes the proof of this claim. 
\end{proof}

Note that $f_2 \ne f_1$ since $f_1 \in I_1$ and $f_2\in A\setminus I_1$.
Let $t \in \Zp$. As shown in Claim~\ref{claim-e2-f2} for $t = 3$, suppose we have constructed for $2 \le j \le t-1$:
\begin{enumerate}
    \item $f_j\in A\setminus I_{j-1}$,
    %\ $\lev_{\C}(f_j)=\lev_{\D}(f_j)$, and 
    \ $I_j := I_{j-1} -e_{j-1}+f_j \in \I$, \ $I_j\subseteq E(\C)$, and
    $\lev_{\C}(e_{j-1})=\lev_{\C}(f_j)$,
    \item $e_j\in I_j\setminus A$ and $\lev_{\C}(e_j)=\lev_{\D}(e_j)\leq k$.
\end{enumerate}

For each $j$ with $2 \leq j \leq t-1$, note that $I_j$ satisfies $\spn(I_j\cap C_i)=\spn(I\cap C_i)=C_i$ 
for each $1\leq i\leq k-1$. Indeed, since $\lev_{\C}(e_{j-1})=\lev_{\C}(f_j)$, we have $|I_j\cap C_i|=|I\cap C_i|=\rank(C_i)$ for each $i\leq k-1$. This implies $\spn(I_j\cap C_i)=C_i$. Claim~\ref{claim-next} generalizes Claim~\ref{claim-e2-f2} for any $t \ge 3$.

\begin{claim}
 \label{claim-next}
    There exists $v_t \in V$ such that there are $e_t, f_t \in \delta(v_t)$ satisfying the following properties: 
    \begin{enumerate}
            \item $f_t \in A\setminus I_{t-1}$, 
            %$\lev_{\C}(f_t)=\lev_{\D}(f_t)$ and 
            \ $I_t:=I_{t-1}-e_{t-1}+f_t \in \I$,\  $I_t\subseteq E(\C)$, and 
            $\lev_{\C}(e_{t-1})=\lev_{\C}(f_t)$,
            \item $e_t \in I_t \setminus A$ and $\lev_{\C}(e_t)=\lev_{\D}(e_t)\leq k$.
    \end{enumerate}
\end{claim}
\begin{proof}
    Let us apply the exchange axiom in Fact~\ref{fact:exchange} to $I_{t-1}, A \in \I$ and $e_{t-1}\in I_{t-1}\setminus A$. 
Since $A+e_{t-1}\not\in \I$, there exists $f_t\in A\setminus I_{t-1}$ such that $I_{t-1}-e_{t-1}+f_t$ and $A+e_{t-1}-f_t$ 
are in $\I$. 

By the conditions $\spn(A\cap D_i)=D_i$ for $1\leq i\leq q$ we have $\lev_{\D}(f_t)\leq \lev_{\D}(e_{t-1})$,  
and by
$\spn(I_{t-1}\cap C_i)=C_i$ for $1\leq i\leq k-1$ and $\lev_{\C}(e_{t-1})\leq k$ we have $\lev_{\C}(e_{t-1})\leq \lev_{\C}(f_t)$. Then
$\lev_{\C}(e_{t-1})\leq \lev_{\C}(f_t)\leq \lev_{\D}(f_t)\leq \lev_{\D}(e_{t-1})=\lev_{\C}(e_{t-1})$,
and hence all the equalities hold. 

So we have
$f_t\in A\setminus I_{t-1},\lev_{\C}(f_t)=\lev_{\D}(f_t)$, and $\lev_{\C}(e_{t-1})=\lev_{\C}(f_t)$.
As $f_t\in A\subseteq E(\D)$, Claim~\ref{claim:1} implies $f_t\in E(\C)$.

Observe that $I_t:=I_{t-1}-e_{t-1}+f_t=(I\cap C_k)+f_1-e_1+\ldots + f_{t-1}-e_{t-1}+f_t\subseteq E(\C)$, and recall $I_t \in \I$. 
Since $\lev_{\C}(e_{j-1})=\lev_{\C}(f_{j})$ for $2 \leq j \le t$ and $\lev_{\C}(f_1)=k$, the set $I_t$ is lexicographically better than~$I$. This implies 
that $I_t$ must violate the partition matroid constraint.
By the same argument as used in Claim~\ref{claim-e1-f1} to show $\lev_{\C}(e_1)=\lev_{\D}(e_1)$,  
we see that there exists $e_t$ such that $e_t, f_t\in \delta(v_t)$ for some~$v_t$, satisfying also
$e_t\in I_t\setminus A$ and $\lev_{\C}(e_t)=\lev_{\D}(e_t)\leq k$.
This completes the proof of this claim. 
\end{proof}

Observe that $f_t$ is distinct from $f_1,\ldots,f_{t-1}$ since $\{f_1,\dots, f_{t-1}\}\subseteq I_{t-1}$ 
while $f_t \in A \setminus I_{t-1}$. Thus, for each~$t \in \Zp$, we have shown
distinct elements $f_1,\ldots,f_t$ in $A$, contradicting that $|A| \le n$. 
Therefore,
it has to be the case that $D_k\subseteq \spn(I\cap C_k)$.

This completes the proof of Lemma~\ref{prop:correctness}. 
\end{myproof}

We conclude this section with the proof of Lemma~\ref{lem:duality}, which was postponed in Section~\ref{sec:dual}.

%\lempropduality*
\begin{myproof}{Proof of Lemma~\ref{lem:duality}} % moved from Appendix
The optimal value of \ref{LP1} is at least~$0$ since $\wt_A(A) = 0$. Thus if there exists a feasible solution 
$(\vec{y},\vec{\alpha})$ to \ref{LP2} whose objective value is $0$, then $(\vec{y},\vec{\alpha})$ is an optimal solution to \ref{LP2}.
Since the optimal value of \ref{LP2} is 0, $A$ is a popular arborescence in $G$.

If $A$ is a popular arborescence, 
then the optimal value of \ref{LP2} is 0.
We will now show there always exists an optimal solution $(\vec{y},\vec{\alpha})$ to \ref{LP2} that satisfies properties~1-4.

\smallskip

    1. It is a well-known fact on matroid intersection (see \cite[Theorem~41.12]{Schrijver} or \cite[Lecture 12, Claim~2]{Goemans})
that there exists an integral optimal solution to \ref{LP2} such that the support of the dual variables corresponding to the matroid $M$ is a chain. 
Thus property~1 follows.

    \smallskip
    
2. Among all the optimal solutions to \ref{LP2} that satisfy property~1, let $(\vec{y}, \vec{\alpha})$ be the one that
minimizes $\sum_{C\in \C}|\spn(C)\setminus C|$, where $\C$ is the support of $\vec{y}$. We claim that $\spn(A\cap C)=C$ holds for all $C\in \C$. 
Observe that each $C\in \C$ satisfies $y_C>0$, and hence complementary slackness implies that the characteristic vector~$x$ of~$A$ satisfies $\sum_{e\in C} x_e=\rank(C)$, i.e., $|A\cap C|=\rank(C)$. Therefore, to obtain $\spn(A\cap C)=C$ for all $C\in \C$, it suffices to show $\spn(C)=C$ for all $C\in \C$.
Suppose contrary that it does not hold. Then there exists at least one~$C\in \C$ with $\spn(C)\neq C$. Among all such $C$, let $C^*\in \C$ be the maximal one. 

Define $\vec{z}$ as follows: (i)~$z_{\spn(C^*)} = y_{\spn(C^*)}+y_{C^*}$, (ii)~$z_{C^*}=0$, and (iii)~$z_S=y_S$ for all other $S\subseteq E$. 
Then $\C'=(\C\setminus \{C^*\})\cup \{\spn(C^*)\}$ is the support of $\vec{z}$. Note that $\C'$ is again a chain because any $C\in \C$ with $C^*\subsetneq C$ satisfies $\spn(C)=C$ by the choice of $C^*$, hence $\spn(C^*)\subseteq \spn(C)=C$. %
    
Observe that $(\vec{z}, \vec{\alpha})$ is a feasible solution to \ref{LP2}.
Moreover, since $\rank(C^*)=\rank(\spn(C^*))$, it does not change the objective value. Thus 
$(\vec{z}, \vec{\alpha})$ is an optimal solution to \ref{LP2} that satisfies property~1 and $\sum_{C\in \C'}|\spn(C)\setminus C|<\sum_{C\in \C}|\spn(C)\setminus C|$. This contradicts the choice of $(\vec{y},\vec{\alpha})$. 

\medskip

3.
Suppose $(\vec{y}, \vec{\alpha})$ satisfies properties~1--2 but not property~3. If $\emptyset\in \C$, then remove $\emptyset$ from $\C$ and modify $\vec{y}$ by setting $y_{\emptyset}=0$. 
This does not change the objective value and does not violate feasibility constraints.

If $E\not\in \C$, then add $E$ to $\C$ and modify $(\vec{y}, \vec{\alpha})$ by (i)~setting $y_E=1$ and (ii)~decreasing every $\alpha_v$ value by $1$. 
Since $\rank(E)=|V|$, the objective value does not change. Also, all constraints in \ref{LP2} are preserved. 
Hence the new solution satisfies properties~1--3.

\medskip
4. Among all the optimal solutions to \ref{LP2} that satisfy properties~1--3, let $(\vec{y}, \vec{\alpha})$ be the one that
minimizes $\sum_{S\subseteq  E}y_S$ and let $\C$ be the support of $y$. 
Note that $\alpha_v = -\sum_{C\in \C: A(v)\in C}y_C$ holds for any $v\in V$ by complementary slackness (observe that $x_{A(v)}>0$ for $A$'s characteristic vector $x$).

Suppose $y_{C^*}\geq 2$ for some $C^*\in \C$. Define $(\vec{z}, \vec{\beta})$ as follows: $z_{C*}=y_{C^*}-1$ and $z_S=y_S$ for
every other $S\subseteq E$. For any $v \in V$, let $\beta_v = -\sum_{C\in \C: A(v)\in C}z_C$. We will show below that
$(\vec{z}, \vec{\beta})$ is a feasible solution to \ref{LP2}. Let us first see what is the objective value attained by 
$(\vec{z}, \vec{\beta})$. 

\smallskip

This value is $\sum_{C\in \C} \rank(C)\cdot z_C+\sum_{v\in V}\beta_v$. 
When compared to $\sum_{C\in \C} \rank(C)\cdot y_C+\sum_{v\in V}\alpha_v$, the first term has decreased by 
$\rank(C^*)$ and the second term has increased by $|\set{v\in V: A(v)\in C^*}|=|A\cap C^*|\leq \rank(C^*)$. Thus the objective value does not increase.

We will now show that $(\vec{z},\vec{\beta})$ is a feasible solution to \ref{LP2}, that is, $\sum_{C\in \C:e\in C}z_C+\beta_v\geq \wt_A (e)$ for each~$e\in \delta(v)$, $v\in V$. 
Since $(\vec{y}, \vec{\alpha})$ is feasible and the first term $\sum_{C\in \C:e\in C}z_C$  decreases by at most $1$ and the second term $\beta_v=-\sum_{C\in \C:A(v)\in C}z_C$ never decreases, the only case we need to worry about is when the first term decreases and the second term  does not increase. This implies that $e\in C^*$ and $A(v)\not\in C^*$; hence 
$\sum_{C\in \C:e\in C}z_C+\beta_v=\sum_{C\in \C:e\in C}z_C-\sum_{C\in \C:A(v)\in C}z_C\geq z_{C*}\geq 1\geq \wt_A(e)$. Thus 
$(\vec{z},\vec{\beta})$ is a feasible solution to \ref{LP2}; furthermore, it is an optimal solution to \ref{LP2}. Since $\sum_{S\subseteq  E}z_S<\sum_{S\subseteq  E}y_S$, this contradicts the choice of~$(\vec{y}, \vec{\alpha})$.

Thus, we have shown that $(\vec{y}, \vec{\alpha})$ satisfies properties~1--3 and $y_C=1$ for all $C\in \C$.
Since we have $\alpha_v= -\sum_{C\in \C: A(v)\in C}y_C$, it follows that $\alpha_v = -|\set{C\in \C:A(v)\in C}|$ for each $v\in V$.
\end{myproof}

\section{Popular Colorful Forests}
\label{sec:colorful}
This section proves Corollary~\ref{cor:pop-common} and Theorem~\ref{thm:polytope} (in terms of the \myproblem{popular colorful forest} problem). 
Let $H=(U_H, E_H)$ be an undirected graph where $E_H = E_1\cupdot \cdots \cupdot E_n$, i.e., $E_H$ is partitioned into $n$ color classes. Equivalently,
there are $n$ agents $1,\ldots,n$ where agent~$i$ owns the elements in $E_i$. For each $i$, there is a 
%weak ranking 
% wo->po
partial order $\succ_i$ 
over elements in $E_i$. 

Recall that $S \subseteq E_H$ is a {\em colorful forest} if 
(i)~$S$ is a forest in $H$ and (ii)~$|S\cap E_i|\leq 1$ for every $i \in \{1,\ldots,n\}$. 
We refer to Section~\ref{sec:intro} on how every agent compares any pair of colorful forests; for any pair of colorful
forests $F$ and $F'$, let $\phi(F,F')$ be the number of agents that prefer $F$ to $F'$.

\begin{Definition}
  \label{def:color-forest}
   A colorful forest $F$ is popular if $\phi(F,F') \ge \phi(F',F)$ for any colorful forest $F'$.
\end{Definition}

The \myproblem{popular colorful forest} problem is to decide if a given instance $H$ admits a popular colorful forest or not.
We will now show that Algorithm~\ref{alg:pop-arb} solves the popular colorful forest problem.

Observe that a popular colorful forest is a popular common independent set in the intersection of the partition matroid defined by 
$E_H = E_1\cupdot \cdots \cupdot E_n$ and the graphic matroid of $H$. In order to
use the popular common {\em base} algorithm to solve this problem, we will augment the ground set $E_H$.

\medskip

\paragraph{An auxiliary instance $G$.} 
For each $i\in \{1,\ldots,n\}$, add a dummy edge $e_i = (u_i,v_i)$ with endpoints $u_i, v_i$, where $u_i$ and $v_i$ are new vertices 
that we introduce; call the resulting graph $G$. 
The vertex and edge sets of~$G=(U,E)$ are given by $U=U_H \cup \bigcup_{i=1}^n\{u_i,v_i\}$ and  $E = E_H  \cup \bigcup_{i=1}^n\{e_i\}$.
Furthermore, for each $i$, the edge $e_i$ will be the {\em worst} element in $i$'s preference order
$\succ_i$, i.e., every $f\in E_i$ satisfies $f\succ_i e_i$. 

In the setting of general matroids, $n$ dummy elements $e_1,\ldots,e_n$ are being introduced into the ground set $E$ as 
{\em free} elements, i.e., for any~$i$, no set $S \subseteq E$ such that $e_i \notin S$ can span $e_i$. 
The partitions in the constructed matroid are $E_i \cup \{e_i\}$ for all $i \in \{1,\ldots,n\}$.

Observe that there exists a one-to-one correspondence between colorful forests in $H$ and colorful forests of size~$n$ in $G$.
Suppose $F_H$ is a colorful forest in $H$ and let $C \subseteq \{1,\ldots,n\}$ be the set of colors missing in $F_H$, i.e., $F_H \cap E_i = \emptyset$
exactly if $i \in C$. Let $F_G = F_H \cup \bigcup_{i\in C}\{e_i\}$. 
Then
$F_G$ is a colorful forest of size~$n$ in $G$. Conversely, given a colorful forest $F_G$ of
size~$n$ in $G$, we can obtain a colorful forest $F_H$ in $H$ by %
deleting the dummy elements.

\medskip

%\noindent{\bf Colorful forests in $G$.}
\paragraph{Colorful forests in $G$.}
Let $F_H$ and $F'_H$ be colorful forests in $H$ and let $F_G$ and $F'_G$ be the corresponding forests (of size~$n$) in $G$.
Observe that $\phi(F_H, F'_H)=\phi(F_G, F'_G)$. Thus popular colorful forests in $H$ correspond to popular colorful forests 
of size~$n$ in $G$ and vice-versa. We want popular colorful forests of size~$n$ to be popular common {\em bases} in the
intersection of the partition matroid and the graphic matroid of $G$.

Hence we will consider the $n$-truncation of the graphic matroid of $G$, i.e., all sets of size larger than $n$ will be deleted 
from the graphic matroid of $G$. 
The function 
$\rank(\cdot)$ now denotes the rank function of the truncation and we have $\rank(E) = n$. 
Thus solving the popular common base problem in the intersection of the 
partition matroid defined by the color classes on $E$ and the truncated graphic matroid of $G$ solves the popular colorful forest problem in $H$. 
Observe that such a reduction holds for the \myproblem{popular common independent set} problem; hence
Corollary~\ref{cor:pop-common} follows.

\medskip

\paragraph{The popular colorful forest polytope.}
We will henceforth refer to a colorful forest of size~$n$ in the auxiliary instance $G$ as a {\em colorful base} in $G$.
Every popular colorful base $F$ in $G$ has a dual certificate as given in Lemma~\ref{lem:duality}\footnote{In \ref{LP1} and \ref{LP2} defined with respec to~$F$, the set~$\delta(v)$ for $v \in V$ will be replaced by $E_i \cup \{e_i\}$ for $i \in \{1,\ldots,n\}$, and in the definition of~$\wt_F$, the edge $A(v)$ will be replaced by the unique element in $F\cap (E_i \cup \{e_i\})$,  %will be 
denoted by $F(i)$.} and Lemma~\ref{lem:chain}.
We will now show these dual certificates are even more special than what is given in Lemma~\ref{lem:chain}---along with 
%properties~1-4,
the properties described there,
the following property is also satisfied.

\begin{lemma}
\label{lem:level-two}
Let $F$ be a popular colorful base in the auxiliary instance  $G$ and let $\C = \{C_1,\ldots,C_p\}$ be a dual certificate for $F$. 
Then $p \le 2$.
\end{lemma}
\begin{proof}
Suppose not, i.e., $p \ge 3$. From the definition of a dual certificate $\C$, we have $\emptyset\subsetneq C_1\subsetneq C_2\subsetneq \cdots \subsetneq C_p = E$ 
(see Lemma~\ref{lem:chain}).
We will now show that $F \cap C_1 = \emptyset$. Since $\spn(F \cap C_1) = C_1$, this means $C_1 = \emptyset$;
however this contradicts $C_1 \ne \emptyset$. This will give us the desired contradiction, 
%hence it has to be the case that 
proving $p \le 2$.

In order to show that $F \cap C_1 = \emptyset$, it suffices to prove that for each $i\in \{1,\ldots,n\}$, 
the unique element in~$F\cap (E_i \cup \{e_i\})$, denoted by  $F(i)$, is not contained in $C_1$. 
\begin{itemize}
\item If $F(i) \neq e_i$, then the dummy edge $e_i$ is not in $F$. Since $e_i$ is not spanned by any set $S\subseteq E$ with $e_i\not\in S$ and $\rank(S)<n$, the condition $\spn(F \cap C_j)=C_j$, yielding also $|F\cap C_j|=\rank(C_j)$, for all $j=1,2,\dots,p$ implies that $e_i\not\in C_j$ for any $j < p$. Hence $\lev_{\C}(e_i)=p$,  
which implies that
every edge in $E(\C)\cap E_i$ has level either $p$ or $p-1$. Because $p \ge 3$, this means 
that no edge of $C_1$ is present in $E(\C)\cap E_i$. 
Thus we have 
$F(i)\not\in C_1$.  
\item If $F(i)=e_i$, then $e_i\in E(\C)$. This implies $\lev_{\C}(e_i) > 1$ because $e_i$ is the worst element in $E_i \cup \{e_i\}$. Hence $F(i)$ is not in $C_1$.
\end{itemize}
In both cases,
$F(i)\not\in C_1$ for any $i \in \{1,\ldots,n\}$. 
Thus
we have $F\cap C_1=\emptyset$, as desired.
 \end{proof}

Lemma~\ref{lem:level-two} shows that any dual certificate $\C$ for a popular colorful base $F$ in $G$ has length at most $2$, i.e., 
$F$ has a dual certificate either of the form $\C = \{E\}$ or of the form $\C = \{C, E\}$. Let $F$ be the popular colorful base computed 
by Algorithm~\ref{alg:pop-arb} in $G$ and let $\C$ be a dual certificate for $F$. The following 
%theorem 
lemma shows that if preferences are weak rankings, then $\C$ is a dual certificate for all popular colorful bases. 
%The proof relies on the assumption that the preferences are weak rankings.
Note that this proof crucially uses the fact that preferences are weak rankings---recall that we use this assumption in Theorem~\ref{thm:polytope} as well. 
Indeed, assuming weak rankings is indispensable there, since the \myproblem{min-cost popular colorful forest} problem for partial order preferences is $\mathsf{NP}$-hard, due to the $\mathsf{NP}$-hardness of its special case, the  \myproblem{min-cost popular branching} problem with partial order preferences~\cite{KKMSS20}.

\begin{lemma}
    \label{lem:face}
    Assume that preferences are weak rankings and suppose that $F$ is the popular colorful base computed 
by Algorithm~\ref{alg:pop-arb} in the auxiliary instance~$G$, and $\C$ is a dual certificate for $F$.
    Then for any arbitrary popular colorful base $F'$ in $G$, we have (i)~$F'\subseteq E(\C)$ and (ii)~if 
    $\C = \{C,E\}$, then $|F'\cap C|=\rank(C)$. 
\end{lemma}
\begin{proof}
%Recall that $F$ is the popular colorful base computed by Algorithm~\ref{alg:pop-arb} in $G$.
Let $(\vec{y}, \vec{\alpha})$ be the dual variables defined from $\C$ as given in Lemma~\ref{lem:duality}. 
That is, $y_{\hat{C}}=1$ for each $\hat{C}\in \C$ and $y_S=0$ for any other $S\subseteq E$, and $\alpha_i=-|\set{\hat{C}\in \C: F(i)\in \hat{C}}|$ for every 
$i\in \{1,\ldots,n\}$.
Note that the length of $\C$ is at most two by Lemma~\ref{lem:level-two}.

Consider \ref{LP1} and \ref{LP2} defined with respect to $F$. 
Since both $F$ and $F'$ are popular, their characteristic vectors are both optimal solutions to \ref{LP1}. 
Since $(\vec{y}, \vec{\alpha})$ is an optimal solution to \ref{LP2}, if $\C = \{C,E\}$ then we have $|F'\cap C|=\rank(C)$ by complementary slackness. Then, what is left is to show $F' \subseteq E(\C)$. We consider the cases where the length of $\C$ is one and two.

\begin{enumerate}
    \item Suppose $\C = \{E\}$. Let $\D$ be a dual certificate of $F'$ as described in Lemma~\ref{lem:chain}. 
Then $F'\subseteq E(\D)$. 
Assume that $\D=\{D, E\}$ (otherwise $\D = \{E\} = \C$). 

Take any $i \in \{1,\ldots,n\}$. We now show $F'(i)\in E(\C)$.
If $F'(i) \in D$ then $\lev_{\D}(F'(i)) = 1 = \lev_{\C}(F'(i))$; %
along with $F'(i)\in E(\D)$, this implies 
$F'(i)\in E(\C)$ by Claim~\ref{claim:1}.
We thus assume that $F'(i) \notin D$. 

Since the characteristic vector $x$ of $F$ and $x'$ of $F'$ are optimal solutions to \ref{LP1} (defined with respect to $F$) and 
$(\vec{y}, \vec{\alpha})$ is an optimal solution to \ref{LP2} (its dual LP), we will use complementary slackness.
Because $x_{F(i)}=1$, we have $\sum_{\hat{C}\in \C:F(i)\in \hat{C}}y_{\hat{C}} + \alpha_i = \wt_F(F(i))\, (=0)$. %
Similarly, because $x'_{F'(i)}=1$, we have $\sum_{\hat{C}\in \C:F'(i)\in \hat{C}}y_{\hat{C}}+\alpha_i = \wt_F(F'(i))$. 
By subtracting the former from the latter, we obtain 
\begin{equation}
\sum_{\hat{C}\in \C: F'(i)\in \hat{C}}y_{\hat{C}} \ \ -\sum_{\hat{C}\in \C: F(i)\in \hat{C}}y_{\hat{C}} \ = \ \wt_F(F'(i)).
\label{eq:CS}
\end{equation}
Since $\C = \{E\}$, the left hand side is $1-1 = 0$. 
By this
$\wt_F(F'(i)) = 0$, which implies $F(i) \sim_i F'(i)$. 
The fact $F(i)\in E(\C)$ implies that $F(i)$ is maximal with respect to $\succ_i$ in $E_i\cup \{e_i\}$. Because $\succ_i$ is a weak ranking, $F(i) \sim_i F'(i)$ means that $F'(i)$ is also maximal, and hence $F'(i) \in E(\C)$ follows.
%Thus we have $F'(i) \in E(\C)$ by the definition of $E(\C)$.

\item Suppose $\C = \{C, E\}$. Let $\D$ be a dual certificate of $F'$. 
Then we have $\D = \{D,E\}$ and $D \subseteq C$ (by Lemma~\ref{prop:correctness}).
Take any $i \in \{1,\ldots,n\}$. We now show $F'(i)\in E(\C)$.
If $F'(i)\not\in C$ (resp., if $F'(i)\in D$), then $F'(i)\not\in D$ (resp., $F'(i)\in C$); hence $\lev_{\C}(F'(i)) = \lev_{\D}(F'(i))$. This fact
along with $F'(i)\in E(\D)$ implies that $F'(i)\in E(\C)$, by Claim~\ref{claim:1}. 
Therefore, let us assume that $F'(i)\in C\setminus D$.

By the same analysis as given in Case~1, Equation~\eqref{eq:CS} holds. Let us also consider \ref{LP1} and \ref{LP2} defined with respect to $F'$ 
(instead of $F$). Let $(\vec{z}, \vec{\beta})$ be the optimal solution of \ref{LP2} corresponding to $\D$. As before, 
the characteristic vectors of $F$ and $F'$ are optimal solutions to \ref{LP1}. 
By the same argument (with $F'$, $F$ and $\D$ taking the places of $F$, $F'$, and $\C$, resp.), we have:
\begin{equation}
\sum_{\hat{D}\in \D:F(i)\in \hat{D}}z_{\hat{D}} \ \ -\sum_{\hat{D}\in \D:F'(i)\in \hat{D}}z_{\hat{D}} \ = \ \wt_{F'}(F(i)).
\label{eq:CS2}
\end{equation}

Since $F'(i)\in C$, the left hand side of \eqref{eq:CS} is $1$ or $0$, and so is $\wt_F(F'(i))$, which implies that we have $F'(i)\succ_i F(i)$ 
or $F'(i) \sim_i F(i)$. Furthermore, since $F'(i)\notin D$, the left hand side of~\eqref{eq:CS2} is $1$ or~$0$, and so is~$\wt_{F'}(F(i))$, 
which implies that 
$F(i)\succ_i F'(i)$ or $F(i)\sim_i F'(i)$.
Therefore we must have $F'(i) \sim_i F(i)$. Hence $F(i)\in C$ follows from \eqref{eq:CS}. 

We have shown that $F'(i)\sim_i F(i)$ and $F(i)\in C$. We also have $F'(i) \in C$. Since $F(i)\in E(\C)$, we see that $F(i)$ is maximal in $C\cap (E_i\cup \{e_i\})$ and dominates all elements in $(E_i\cup \{e_i\})\setminus C$ with respect to $\succ_i$. Since $\succ_i$ is a weak ranking and $F'(i)\sim_i F(i)$, the element $F'(i)\in C$ also satisfies these conditions, and hence $F'(i)\in E(\C)$.
%Thus it follows from the definition of $E(\C)$ that $F'(i)\in E(\C)$.
\end{enumerate}
Thus we have $F'(i) \in E(\C)$ for every $i \in \{1,\ldots,n\}$. Hence $F' \subseteq E(\C)$.
\end{proof}

By Lemma~\ref{lem:face}, any popular colorful base $F'$ in $G$ satisfies $F'\subseteq E(\C)$ and  $|F'\cap C|=\rank(C)$ if $\C = \{C,E\}$. 
Conversely, any popular colorful base $F'$ in $G$ that satisfies these conditions is popular by Lemma~\ref{lem:chain}.
Therefore the set of all popular colorful bases in $G$ can be described as a face of the matroid intersection polytope.
Since a popular colorful forest in the given instance $H$ is obtained by 
deleting the dummy elements from popular colorful bases in $G$,
Theorem~\ref{thm:polytope} follows. 

We also state this result explicitly in Theorem~\ref{thm:colorful-forest} in the setting of popular colorful forests.
Let $\C = \{C,E\}$ be a dual certificate for the popular colorful base $F$ in $G$ computed by Algorithm~\ref{alg:pop-arb}.

\begin{theorem}
\label{thm:colorful-forest}
If preferences are weak rankings, an extension of the popular colorful forest polytope of the given instance $H$ is defined by the constraints $\sum_{e\in C}x_e=\rank(C)$ 
and $x_e=0$ for all $e\in E\setminus E(\C)$ along with all the constraints of \ref{LP1}.
\end{theorem}

\section{Min-Cost Popular Arborescence}
\label{sec:hardness-mincost}
We prove Theorem~\ref{thm:min-cost} in this section.
We present a reduction from the \textsc{Vertex Cover} problem, whose input is an undirected graph~$H$ and an integer~$k$, 
and asks whether $H$ admits a set of~$k$ vertices that is a vertex cover, that is, contains an endpoint from each edge in~$H$.

Our reduction is strongly based on the reduction used in~\cite[Theorem 6.3]{KKMSS20} which showed the $\NP$-hardness of the 
\myproblem{min-cost popular branching} problem when vertices have partial order preferences. Recall that the \myproblem{min-cost popular branching} problem
is polynomial-time solvable when vertices have weak rankings~\cite{KKMSS20} (also implied by Theorem~\ref{thm:polytope}).
Note also that neither the hardness of \myproblem{min-cost popular branching} for partial order preferences~\cite{KKMSS20}, nor the hardness of \myproblem{min-cost popular assignment} for strict preferences~\cite{KKMSS22} implies Theorem~\ref{thm:min-cost}, since the \myproblem{min-cost popular arborescence} problem does not contain either of these problems.

To show the $\NP$-hardness of the \myproblem{min-cost popular arborescence} problem when vertices have strict rankings, we construct 
a directed graph~$G=(V \cup \{r\},E=E_1 \cup E_2 \cup E_3)$ as follows; see Figure~\ref{fig:min-cost-hardness} for an illustration. 
We set
\begin{align*}
V&=\{w\} \cup \{v_0,v_1:v \in V(H) \} \cup \{e_u,e_v:e=uv \in E(H) \},\\
E_1 & = \{(e_u,e_v),(e_v,e_u),(e_u,w),(e_v,w) : e=uv \in E(H) \} \\
 & \qquad \cup \{(v_0,v_1),(v_1,v_0): v \in V(H)\},  \\
E_2 & = \{(r,w)\} \cup \{ (w,x): x \in V(G) \setminus \{r,w\} \}, \\
E_3 &= \{(r,v_1):  v \in V(H)\} \cup \{(u_0,e_u),(v_0,e_v): e=uv \in E(H) \}.
\end{align*}

To define the preferences of each vertex in~$G$, we let 
all vertices prefer edges of~$E_1$ to edges of~$E_2$, which in turn are preferred to edges of~$E_3$. 
Whenever some vertex has more than one incoming edge in some $E_i$, $i \in \{1,2,3\}$, then it orders them in some arbitrarily fixed strict order.
We set the cost of each edge in $E_3$, as well as the cost of all edges entering $w$ except for~$(r,w)$ as $\infty$.
We set the cost of $(w,v_1)$  as~$1$ for each~$v \in V(H)$, and we set the cost of all remaining edges as~$0$.
We define our budget to be~$k$, finishing the construction of our instance of \myproblem{min-cost popular arborescence}.

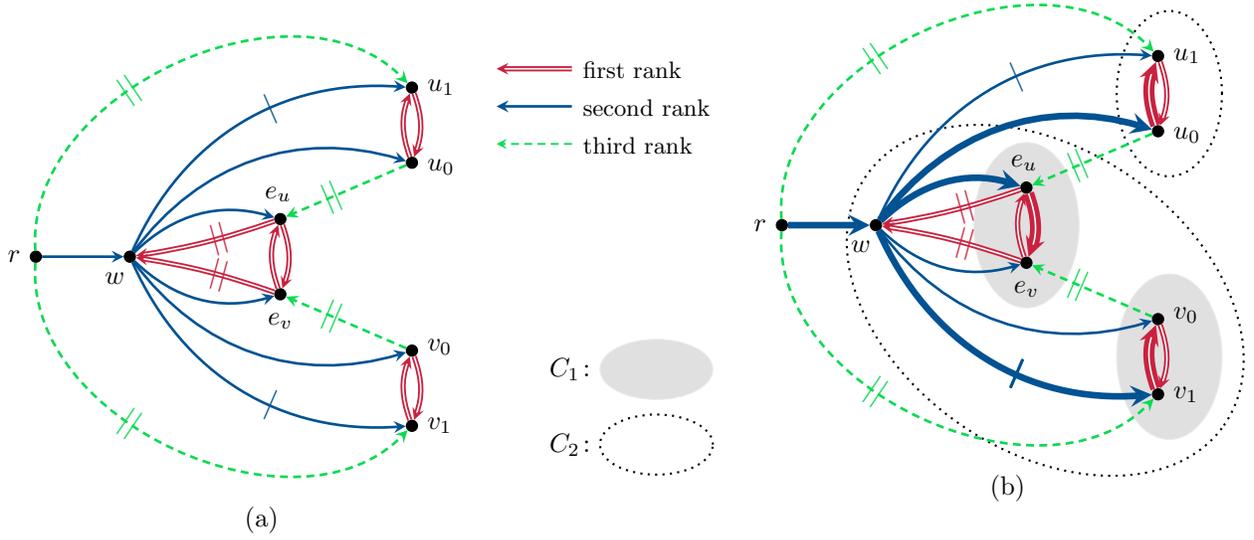
\begin{figure}[t]
  \centering
  \begin{subfigure}[b]{0.58\textwidth}
    \centering  
    \input{tikz_instance_mincost_1}
  \end{subfigure} 
  \hspace{-5pt}
  \begin{subfigure}[b]{0.4\textwidth}
    \centering  
    \input{tikz_instance_mincost_2}  
  \end{subfigure}
  \caption{Illustration of the reduction in the proof of Theorem~\ref{thm:min-cost}. Figure~(a) illustrates the construction showing a subgraph of~$G$, assuming that the input graph~$H$ contains an edge~$e=uv$. Edges in $E_1$, $E_2$, and $E_3$ are depicted with double red, single blue, and dashed green lines, respectively. 
  Edges marked with two, one, and zero crossbars have cost $\infty$, $1$, and~$0$, respectively. Figure~(b) illustrates the popular arborescence~$A$ in bold, assuming $v \in S$ and $u \notin S$. The chain $C_1 \subsetneq C_2 \subsetneq C_3=E$ certifying the popularity of~$A$ is shown using grey and dotted ellipses for edges in~$C_1$ and~$C_2$, respectively.
  }
  \label{fig:min-cost-hardness}
\end{figure}

\medskip
We are going to show that $H$ admits a vertex cover of size at most~$k$ if and only if $G$ has a popular arborescence of cost at most~$k$. 

\medskip
Suppose first that $A$ is a popular arborescence in~$G$ with cost at most~$k$. 
We prove that 
the set $S=\{v \in V(H): (w,v_1) \in A\}$ is a vertex cover in~$H$. Since each edge $(w,v_1)$ has cost~$1$, our budget implies $|S| \leq k$. 

For a vertex $v \in V(H)$ and an edge $e=uv \in E(H)$, let $A_v=A \cap (\delta(v_0) \cup \delta(v_1))$ and 
$A_e=A \cap (\delta(e_u) \cup \delta(e_v))$, respectively. 
We note that any $v \in V(H)$ satisfies that 
$A_v$ is either $\{(w,v_0), (v_0,v_1)\}$ or $\{(w,v_1), (v_1,v_0)\}$. 
Indeed, if it is not the case, we have $A_v=\{(w,v_0), (w,v_1)\}$, since $A$ is an arborescence with finite cost.  
However, this contradicts the popularity of $A$, since $A\setminus \{(w,v_1)\} \cup \{(v_0,v_1)\}$ is more popular than $A$. 
We can similarly show that each $e=uv \in E(H)$ satisfies that 
$A_e$ is either $\{(w,e_u), (e_u,e_v)\}$ or $\{(w,e_v), (e_v,e_u)\}$. 
Note also that $(r,w) \in A$, as all other edges entering~$w$ have infinite cost.

Assume for the sake of contradiction that $S$ is not a vertex cover of $H$, i.e.,  
 there exists an edge $e=uv \in E(H)$ such that neither $(w,u_1)$ nor $(w,v_1)$ is  contained in $A$. 
Then we have 
$
A_u =\{(w,u_0), (u_0,u_1)\}$ and $A_v= (w,v_0), (v_0,v_1)\}. $
By symmetry, we assume without loss of generality that  $A_e=\{(w,e_u), (e_u,e_v)\}$. 
Define an edge set $A'$ by 
\[
A'=(A \setminus  (A_e\cup A_v \cup \{(r,w)\})) \cup \{ (r,v_1), (v_1,v_0), (v_0,e_v), (e_v,e_u), (e_u, w)\}. 
\]
We can see that $A'$ is an arborescence and is more popular than $A$, since 
three vertices, $v_0$, $e_u$,  and $w$, prefer $A'$ to $A$, while 
two vertices, $v_1$ and $e_v$, prefer $A$ to $A'$, 
and all others are indifferent between them. 
This proves that $S$ is a vertex cover of $H$.

\medskip
For the other direction, assume that $S$ is a vertex cover in~$H$. We construct a popular arborescence~$A$ of cost~$|S|$ in~$G$.
For each $e \in E(H)$ we fix an endpoint $\sigma(e)$ of~$e$ that is contained in~$S$, and we denote by $\bar\sigma(e)$ the other endpoint of~$e$ (which may or may not be in~$S$). Let
\begin{align*}
 A = \{(r,w)\} & \cup \{(w,v_1),(v_1,v_0):v  \in S\} \\   
 & \cup  \{(w,v_0),(v_0,v_1):v  \in V(H) \setminus S\}  \\
 & \cup \{(w,e_{\bar{\sigma}(e)}),(e_{\bar{\sigma}(e)},e_{\sigma(e)}):e  \in E(H)\}. 
\end{align*}
It is straightforward to verify that $A$ is an arborescence and its cost is exactly $|S|$.  
Hence it remains to prove its popularity, which is done by showing a dual certificate~$\C$ for~$A$.

To define~$\C$, let us first define a set $X=\{w\} \cup \{e_u,e_v:e=uv \in E(H)\} \cup \{v_0,v_1:v \in S\}$ of vertices in~$G$. Then we set $\C=\{C_1,C_2,C_3\}$ where 
\begin{align*}
    C_1= & \, \{(e_u,e_v),(e_v,e_u): e=uv \in E(H)\} \cup \{ (v_0,v_1),(v_1,v_0):v \in S \}, \\
    C_2= & \, \{f \in E(H):f  \textrm{ has two endpoints in~$X$} \} \cup  \{ (v_0,v_1),(v_1,v_0):v \in V(H) \setminus S \}, \\
    C_3 = & \, E. 
\end{align*}
Let us first check that $\rank(C_i)=|A \cap C_i|$ for each $C_i \in \C$. Clearly, $C_1$ consists of mutually vertex-disjoint 2-cycles, 
and $A$ contains an edge from each of them. 
Thus $\rank(C_1)=|A \cap C_1|$ follows. 
The edge set~$C_2$ consists of all edges induced by the vertices of~$X$, together with another set of mutually vertex-disjoint 2-cycles that share no vertex with~$X$.
It is easy to verify that $A\cap C_2$ contains an edge from each of the 2-cycles in question, as well as 
a directed tree containing all vertices of~$X$. Thus, $\rank(C_2)=|A \cap C_2|$ holds. 
Since~$A$ is an arborescence, $\rank(C_3)=\rank(E)=|V|=|A \cap C_3|$ is obvious. 
Observe that for each $i \in \{1,2,3\}$ we have $\spn(C_i)=C_i$, 
and hence
$\rank(C_i)=|A \cap C_i|$ implies $\spn(A \cap C_i)=C_i$.

It remains to see that $A \subseteq E(\C)$. First, $A(w)=(r,w)$ is the unique incoming edge of~$w$ with $\C$-level~$3$. 
For some $v \in S$, $\lev^*_\C(v_0)=2$ while $\lev_\C^*(v_1)=3$, and by their preferences  both $A(v_0)=(v_1,v_0)$ and $A(v_1)=(w,v_1)$ are in~$E(\C)$. 
For some $v \in V(H) \setminus S$, $\lev^*_\C(v_0)=\lev_\C^*(v_1)=3$, 
and hence both
$A(v_0)=(w,v_0)$ and $A(v_1)=(v_0,v_1)$ are in~$E(\C)$. 
Finally, consider an edge $e=uv \in E(H)$ with $\sigma(e)=v \in S$. As $\lev^*_\C(e_u) \leq 3$, and since $e_u$ prefers $(w,e_u)$ to~$(u_0,e_u)$, 
we know that the edge $A(e_u)=(w,e_u) \in C_2$ is contained in~$E(\C)$.
By contrast, since $v \in S$ implies $v_0 \in X$, we obtain $\lev^*_\C(e_v) =2$, and therefore the edge $A(e_v)=(e_u,e_v) \in C_1$ is contained in~$E(\C)$.
By Lemma~\ref{lem:chain}, this proves that $A$ is indeed a popular arborescence.

\section{Popular Arborescences with Forced/Forbidden Edges}
\label{sec:forbidden}
We prove Theorem~\ref{thm:forced-forbidden} in this section. Observe that the problem of deciding if there exists a popular 
arborescence~$A$ such that $A \supseteq E^+$ for a given set $E^+ \subseteq E$ of {\em forced} edges can be reduced to the
problem of deciding if there exists a popular arborescence $A$ such that certain edges are {\em forbidden} for $A$. 

Let $V' \subseteq V$ be the set of those vertices $v$ such that $\delta(v) \cap E^+ \ne \emptyset$; clearly, we may assume $|\delta(v) \cap E^+|=1$ for each $v \in V'$. Let 
$E' = \bigcup_{v \in V'} (\delta(v)\setminus E^+)$. Since $A \supseteq E^+$ if and only if $A \cap E' = \emptyset$,
it follows that the problem of deciding if there exists a popular arborescence $A$ such that $E^+ \subseteq A$ and 
$E^- \cap A = \emptyset$ reduces to the problem of deciding if there exists a popular arborescence $A$ such that
$A \cap E_0 = \emptyset$ for a set $E_0 \subseteq E$ of forbidden edges.

\medskip

\paragraph{Forbidden edges.} We present our algorithm that decides if $G$ admits a popular arborescence
that avoids $E_0$ for a given subset $E_0$ of $E$ as Algorithm~\ref{alg:pop-arb-forbidden}. 
The only difference from the original popular arborescence algorithm (Algorithm~\ref{alg:pop-arb}) is in line~4: the new
algorithm finds a lexicographically maximal branching in the set $E(\C)\setminus E_0$ instead of $E(\C)$.
Recall that $\rank(E) = |V| = n$.

\begin{algorithm}
\caption{The popular arborescence algorithm with the forbidden edge set $E_0$}\label{alg:pop-arb-forbidden}
\begin{algorithmic}[1]
\State Initialize $p=1$ and $C_1 = E$.
\Comment{Initially we set $\C = \{E\}$.}
\While{$p \le n$}
\State Compute the edge set $E(\C)$ from the current multichain $\C$.
\State Find a branching  $I\subseteq E(\C)\setminus E_0$ that lexicographically maximizes $(|I\cap C_1|, \dots, |I\cap C_{p}|)$.
\If{$|I\cap C_i|=\rank(C_i)$ for every $i=1,\dots, p$} return $I$. \EndIf
\State Let $k$ be the minimum index such that $|I\cap C_k|<\rank(C_k)$.  
\State Update $C_k \leftarrow \spn(I \cap C_k)$.
\If{$k = p$} $p \leftarrow p+1$, $C_p \leftarrow E$, and $\C \leftarrow \C \cup \{C_p\}$. \EndIf
\EndWhile
\State Return ``$G$ has no popular arborescence that avoids $E_0$''.  
\end{algorithmic}
\end{algorithm}

\begin{theorem}
Let $E_0 \subseteq E$. The instance $G = (V \cup \{r\}, E)$ admits a popular arborescence $A$ such that 
$A \cap E_0 = \emptyset$ if and only if Algorithm \ref{alg:pop-arb-forbidden} returns a popular arborescence 
with no edge of $E_0$.
\end{theorem}
\begin{proof}
The easy side is to show that if Algorithm~\ref{alg:pop-arb-forbidden} returns an arborescence $I$, then (i)~$I$ is popular
and (ii)~$I \cap E_0 = \emptyset$. As done in Section~\ref{sec:algo}, let us prune the multichain $\C$ into a chain $\C'$.
Because $I \subseteq E(\C)\setminus E_0$ and $E(\C) \subseteq E(\C')$, we have $I \subseteq E(\C') \setminus E_0$. Since 
$I \subseteq E(\C')$ and $|I\cap C'_i|=\rank(C'_i)$ (and hence $\spn(I\cap C'_i)=C'_i$) for every $C'_i \in \C'$, it follows from  Lemma~\ref{lem:chain} that $I$ is a popular arborescence.

We now show the converse. 
Suppose that $G$ admits a popular arborescence $A$ with $A \cap E_0 = \emptyset$. 
Let
$\D=\{D_1, \dots, D_q\}$ be a dual certificate for $A$.
Then we have $A\subseteq E(\D)\setminus E_0$.
It suffices to show that Algorithm~\ref{alg:pop-arb-forbidden} maintains the following invariant: the multichain $\C=\{C_1, \dots, C_p\}$ maintained in the algorithm satisfies $p\leq q$ and $D_i\subseteq C_i$ for any $i=1, 2, \dots, p$. 

We can show a variant of Lemma~\ref{prop:correctness}, i.e., we can show that when $C_k$ is updated in the algorithm, $D_k\subseteq \spn(I\cap C_k)$ holds where $I$ is a lexicographically maximal branching in $E(\C)\setminus E_0$.
The proof of Lemma~\ref{prop:correctness} works almost as it is. %
Recall that we sequentially find elements $f_1, e_1, f_2, e_2,\dots$ in the proof of Lemma~\ref{prop:correctness}.
For each $j=1,2,\dots$, in addition to the condition $f_j\in E(\C)$, we have $f_j\not\in E_0$ since $f_j\in A\subseteq E\setminus E_0$. 
By this, 
$I_{j}=(I\cap C_k)+f_1-e_1+f_2\cdots -e_{j-1}+f_j$ satisfies $I_j\subseteq E(C)\setminus E_0$ for each~$j$. Hence the
proof of Lemma~\ref{prop:correctness} works with ``lex-maximality subject to $I\subseteq E(\C)\setminus E_0$'' replacing
``lex-maximality subject to $I\subseteq E(\C)$''.
\end{proof}

\switch{}{   %%%%%%%%%%%%%%%%%%%%% the beginning of switch (2)
\section{Minimum Unpopularity Margin Arborescence}
\label{sec:hardness-min-unpop-margin}
We prove Theorem~\ref{thm:min-unpop-margin} in this section.
%We defer the proofs for some claims used in the proof to Appendix~\ref{app:proofs}; we mark all such results by an asterisk ($\star$).
It is easy to see that the problem is in $\NP$, since given an arborescence~$A$ we can verify $\mu(A)\leq k$ efficiently, assuming that a dual certificate for~$A$ (i.e., a solution for~\ref{LP2} with objective value~$k$) is provided. 

To prove $\NP$-hardness,
we present a reduction from the following $\NP$-hard variant of the \textsc{Exact 3-Cover} problem~\cite{Gonzalez85}. The input contains a set~$U$ of size $3n$ and a set family $\calS=\{S_1,\dots, S_{3n}\}$ where $S_i \subseteq U$ and $|S_i|=3$ for each $S_i \in \calS$, and each $u \in U$ is contained in exactly three sets from~$\calS$. 
The task is to decide whether there exist $n$ sets in~$\calS$ whose union is $U$. 

Our reduction draws inspiration from the reduction used in~\cite[Theorem 4.6]{KKMSS20} which proved the $\NP$-hardness of the 
\myproblem{$k$-unpopularity margin branching} problem when vertices have partial order preferences.  
Recall that this problem was shown to be polynomial-time solvable when vertices have weak rankings~\cite{KKMSS20}.
Note also that Theorem~\ref{thm:min-unpop-margin} does not follow from the $\NP$-hardness of either the \myproblem{$k$-unpopularity margin branching} problem~\cite{KKMSS20} or the \myproblem{$k$-unpopularity margin assignment} problem~\cite{KKMSS22}.

To show the $\NP$-hardness of the \myproblem{$k$-unpopularity margin arborescence} problem when vertices have strict rankings,
we construct a directed graph~$G=(V \cup \{r\},E=E_1 \cup E_2 \cup E_3)$ as follows; see Figure~\ref{fig:unpop-hardness} for an illustration. 
For each $u \in U$ we construct a gadget $G_u$ whose vertex set is $\{u_0,u_1\} \cup A_u \cup B_u$ where $A_u=\{a_{u,1},a_{u,2},a_{u,3}\}$ and $B_u=\{b_{u,1},b_{u,2},b_{u,3}\}$. First we add four 2-cycles, with all their edges in~$E_1$, on vertex sets $\{a_{u,i},b_{u,i}\}$ for each $i=1,2,3$, as well as on $\{u_0,u_1\}$; these $8|U|$ edges comprise~$E_1$. 
We next add edges of~$E_2$: first, we stitch together the three 2-cycles on~$A_u \cup B_u$ 
with edges $(a_{u,3},b_{u,2})$, $(a_{u,2},b_{u,1})$, and  
$(a_{u,1},b_{u,3})$;
second, we add all possible edges between $\{u_0,u_1\}$ and~$A_u$,
creating a bidirected $K_{2,3}$.
We denote the unique 6-cycle on $A_u \cup B_u$ as $C_u$.
This finishes the construction of our gadget~$G_u$. 
To complete the definition of~$G$, it remains to define $E_3$.
To this end, for each $u \in U$ we fix an arbitrary ordering over the three sets of~$\calS$ containing~$u$, and denote them as $S(u,1)$, $S(u,2)$, and $S(u,3)$.
We then let
\[E_3=\{(r,u_0),(r,u_1):u \in U\}
\cup \{(b_{u,i},b_{v,j}): \exists S \in \calS \textrm{ s.t. } S=S(u,i)=S(v,j)\}.
\]

To define the preferences of each vertex in~$G$, we let all vertices prefer edges of~$E_1$ to those in~$E_2$, which in turn are preferred to edges of~$E_3$. 
Whenever some vertex has more than one incoming
edge in some $E_i$, $i \in \{1,2,3\}$, then it orders them in some fixed strict order with the only constraint that edges from~$u_0$ are preferred to edges from~$u_1$ for each $u \in U$. 

\begin{figure}[t]
  \centering
  \input{tikz_instance_margin}  
  \caption{Illustration of a gadget $G_u$ in the proof of Theorem~\ref{thm:min-unpop-margin}. Preferences are encoded using line types and colors as in Figure~\ref{fig:unpop-hardness}.
  }
  \label{fig:unpop-hardness}
\end{figure}
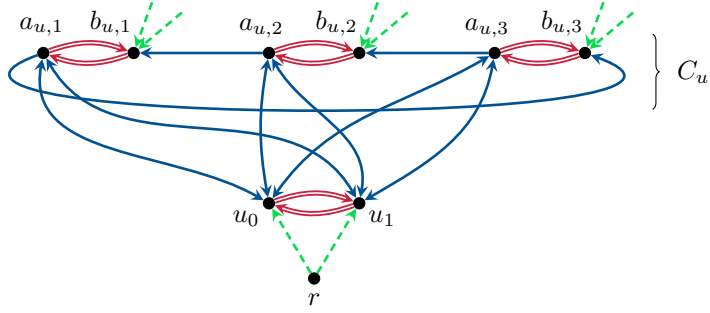

\medskip
We are going to show that our instance of \textsc{Exact 3-Cover} is solvable if and only if $G$ admits an arborescence with $\mu(A)\leq 2n$. 
\medskip

First, assume that there exists some $\T \subseteq \calS$ of size~$n$ that covers each $u \in U$ exactly once. Let $\sigma(u)$ denote the index in $\{1,2,3\}$ for which  $S(u,\sigma(u)) \in \T$. We then let 
\[A=\bigcup_{u \in U} \{
(r,u_0),(u_0,u_1),(u_0,a_{u,\sigma(u)}),
(a_{u,\sigma(u)},b_{u,\sigma(u)})\} \cup
(C_u \setminus \{e \in C_u: e \textrm{ is incident to }b_{u,\sigma(u)}\}).
\]

Note that $A$ is an arborescence in~$G$. To prove that the unpopularity margin of~$A$ is at most~$2n$, we will use the fact that, by definition, $\mu(A)=\max_{A'}\phi(A',A)-\phi(A,A')$ is the optimal value of~\ref{LP1}. Therefore, to show that $\mu(A)\leq 2n$ it suffices to give a dual feasible solution with objective value~$2n$. 
To this end, we define a chain $\C=\{C_1,C_2,C_3\}$ with $C_1 \subsetneq C_2 \subsetneq C_3=E$ by setting 
\begin{align*}
C_1 = &\,\, \{ (a_{u,i},b_{u,i}),(b_{u,i},a_{u,i}): u \in U, i \in \{1,2,3\} \}, \\
C_2= & \,\,\bigcup_{\{u,v,z\} \in \T} \{ e \in E: e \textrm{ has both endpoints in } V(G_u) \cup V(G_v) \cup V(G_z) \}.
\end{align*}
Note that 
    $\rank(C_1)=3|U|$, $\rank(C_2)=(3\cdot 8 -1)n = 7|U|+2n$, and $\rank(C_3)=8|U|$. 

To define a feasible solution $(\vec{y},\vec{\alpha})$ for~\ref{LP2}, for each $S \subseteq E$ we let $y_S=1$ if $S \in \C$, and $y_S=0$ otherwise;  
we also set
\begin{align*}
  \alpha_{a_{u,i}}&=
  \left\{\begin{array}{ll}
  -3 \quad & \textrm{if $i \neq \sigma(u)$,} \\
  -2 & \textrm{if $i = \sigma(u)$,} \end{array}\right. \qquad \qquad \alpha_{u_0}=-1, \\
  \alpha_{b_{u,i}}&=\left\{\begin{array}{ll}
  -2 \quad & \textrm{if $i \neq \sigma(u)$,} \\
  -3 & \textrm{if $i = \sigma(u)$,} \end{array}\right.  \qquad \qquad \alpha_{u_1}=-2, 
\end{align*}
for each $u \in U$.
See Figure~\ref{fig:unpop-dualsol} for an illustration.
The objective value of $(\vec{y},\vec{\alpha})$ is 
\[
\sum_{C_i \in \C} \rank(C_i)+\sum_{v \in V} \alpha_v = 3|U|+7|U|+2n+8|U| - 18|U|=2n.
\]
Therefore, to prove  that $A$ has unpopularity margin at most~$2n$, it suffices to show that $(\vec{y},\vec{\alpha})$ is a feasible solution for~\ref{LP2}, as stated by Claim~\ref{clm:margin_0} below. %whose proof we defer to Appendix~\ref{app:proofs}.
The proofs of claims marked by an asterisk ($\star$) are deferred to Appendix~\ref{app:proofs}.

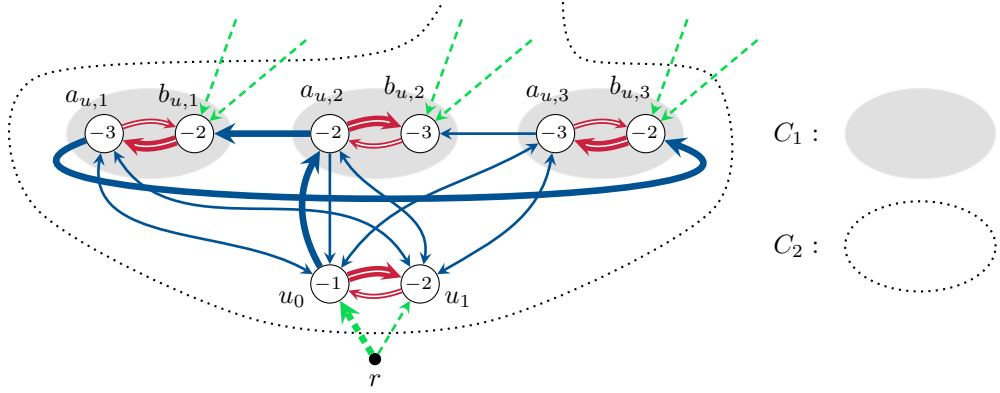
\begin{figure}[t]
    \centering
    \input{tikz_instance_margin_2}
    \caption{Illustration of the arborescence~$A$ in the proof of Theorem~\ref{thm:min-unpop-margin}, shown in bold, together with a feasible dual solution $(\vec{y},\vec{\alpha})$  certifying  $\mu(A) \leq 2n$. The figure assumes $\sigma(u)=2$. 
    The chain $C_1 \subsetneq C_2 \subsetneq C_3=E$ is shown using grey and dotted ellipses for edges in~$C_1$ and~$C_2$, respectively, while the values~$\alpha_v$, $v \in V$, are written within the corresponding vertices.}
    \label{fig:unpop-dualsol}
\end{figure}

\begin{restatable}{claim}{clmmargindual}[$\star$]
\label{clm:margin_0} 
     $(\vec{y},\vec{\alpha})$ is a feasible solution for~\ref{LP2}.
\end{restatable}

\medskip
For the other direction, 
assume that $G$ admits an arborescence~$A$ with $\mu(A) \leq 2n$. Let $B$ be an arborescence that yields an optimal solution for~\ref{LP1}, maximizing  $\phi(B,A)-\phi(A,B) \leq 2n$. First note that we can assume that~$A$ is Pareto-optimal in the sense that there is no arborescence that is weakly preferred by all vertices to~$A$, and strictly preferred by at least one vertex to~$A$. Similarly, we can choose~$B$ to be Pareto-optimal as well. 
Consequently, for any two edges $e,e' \in E_1$ forming a 2-cycle, both~$A$ and~$B$  uses at least one of $e$ and~$e'$.

For some~$X \subseteq V$ and two arborescences~$A'$ and~$A''$, let $\phi_X(A',A'')$ denote the number of vertices in~$X$ that prefer~$A'$ to~$A''$.
We say that a gadget~$G_u$ is \emph{clean}, if $\phi_{V(G_u)}(B,A)-\phi_{V(G_u)}(A,B) \leq 0$.
By $\mu(A)=\phi(B,A)-\phi(A,B)=\sum_{u \in U} \phi_{V(G_u)}(B,A)-\phi_{V(G_u)}(A,B) \leq 2n$, we know that 
there are at least $|U|-2n=n$ clean gadgets. Let $U^\star = \{u: G_u \textrm{ is clean}\}$.

\begin{restatable}{claim}{clmmarginunique}[$\star$]
\label{clm:margin_1}
If~$G_u$ is clean, then a unique edge of~$A$ enters~$G_u$, and it comes from~$r$.    
\end{restatable}

By Claim~\ref{clm:margin_1}, for each $u \in U^\star$ there exists a vertex~$\hat{u} \in \{u_0,u_1\}$ for which $(r,\hat{u}) \in A$.
We can also assume w.l.o.g.~that $A$ and $B$ coincide on $G_u$, since otherwise we can replace~$B$ with the arborescence $B^\star=B \setminus \{\delta(x):x \in V(G_u), u \in U^\star\} \cup \{A(x): x \in V(G_u), u \in U^\star\}$, since $B^\star$ is also optimal for \ref{LP2}.
Furthermore, we get that for each $u \in U^\star$ there exists some~$i \in \{1,2,3\}$ for which $A(a_{u,i})$ comes from $\{u_0,u_1\}$; let $\sigma(u)$ denote this index.
\begin{restatable}{claim}{clmmargindescendant}[$\star$]
\label{clm:margin_2}
    If $u \in U^\star$, then the tail of each edge $f \in \delta(b_{u,\sigma(u)}) \cap E_3$ is a descendant of $\hat{u}$ in~$B$.
\end{restatable}

We claim that $\T=\{S(u,\sigma(u)):u \in U^\star\}$ is a solution to our instance of \textsc{Exact 3-Cover}. 
First observe that $\T$ contains at least $n$ sets by $|U^\star|\geq n$. 
It remains to show that the sets in~$\T$ are pairwise disjoint. 
We say that $G_v$ is \emph{assigned} to $u \in U^\star$, if $v \in S(u,\sigma(u))$. 
It suffices to show that no gadget $G_v$ can be assigned to more than one vertices in~$U^\star$. 

Assume for the sake of contradiction that $G_v$ is assigned to both~$u$ and $w$ for 
two distinct vertices $u,w \in U^\star$.
Then by Claim~\ref{clm:margin_2} there are two vertices in $B_v$, one of them a descendant of $\hat{u}$, the other a descendant of~$\hat{w}$.
Note that neither~$\hat{u}$ nor~$\hat{w}$ is a descendant of the other, since both $(r,\hat{u})$ and $(r,\hat{w})$ are edges in~$A$, and hence, in~$B$ (recall that $A$ and $B$ coincide on $G_u$ and on~$G_w$). This means that there are two distinct edges entering~$B_v$, one from a descendant of~$\hat{u}$, the other from a descendant of~$\hat{w}$. Thus for some $j \in \{1,2,3\}$, the edges $B(b_{v,j})$ and $B(b_{v,j+1})$ are both in~$E_3$, implying also 
$(b_{v,j},a_{v,j}) \in B$ and $(b_{v,j+1},a_{v,j+1}) \in B$, where indices are taken modulo~$3$ (so $a_{v,4}=a_{v,1}$ and $b_{v,4}=b_{v,1}$). However, this contradicts the Pareto-optimality of~$B$, since replacing $B(b_{v,j})$ with $(a_{v,j+1},b_{v,j})$ in~$B$ results in an arborescence that $b_{v,j}$ prefers to~$B$, with all other vertices being indifferent between the two.

This shows that any two sets in~$\T$ are disjoint, proving the correctness of our reduction.
}  %%%%%%%%%%%%%%%%%%%%  the end of switch (2)

\section{Conclusions}
We considered the \myproblem{popular arborescence} problem, which asks to determine whether a given directed rooted graph, in which vertices have preferences over incoming edges, admits a popular arborescence or not and to find one if so.
We provided a polynomial-time algorithm to solve this problem, which affirmatively answers an open problem posed in 2019~\cite{Kir19}.
Our algorithm and its correctness proof work in the generality of matroid intersection (where one of the matroids is a partition matroid), which means that we also solved the \myproblem{popular common base} problem.
Furthermore, we observed that the \myproblem{popular common independent set} problem, which includes the \myproblem{popular colorful forest} problem as a special case, can be reduced to the \myproblem{popular common base} problem, and hence can be solved by our algorithm. Utilizing structural observations, we also proved that the \myproblem{min-cost popular common independent set} problem is tractable 
% wo->po
if preferences are weak rankings.

% wo->po

\begin{comment}
For the sake of simplicity, we assumed throughout the paper that preferences are represented by weak rankings. In fact, however, all our tractability results hold for partial order preferences, with the exception of the \myproblem{min-cost popular common independent set} problem, whose tractability relies on Theorem~\ref{thm:polytope} which does not hold for partial orders.
In particular, our polynomial-time algorithm for the \myproblem{popular common base} problem (Theorem~\ref{thm:pop-largest})  works even for partial order preferences, and therefore generalizes the tractability results for both the \myproblem{popular branching} problem in~\cite{KKMSS20} and for the \myproblem{popular assignment} problem in~\cite{KKMSS22}.
\end{comment}

On the intractability side, we proved that the \myproblem{min-cost popular arborescence} problem and the \myproblem{$k$-unpopularity margin arborescence} problem are both $\NP$-hard even 
for strict preferences. 
Note that the min-cost problem is $\NP$-hard for 
\emph{popular common bases} (a fact implied by the $\NP$-hardness of the \myproblem{popular assignment} problem shown in~\cite{KKMSS22}, as well as by Theorem~\ref{thm:min-cost}), while it is tractable for \emph{popular common independent sets} by Theorem~\ref{thm:polytope}. By analogy, one may expect 
the problem of finding a common independent set with unpopularity margin at most~$k$ to be polynomial-time solvable. 
However, %
this is not the case (unless $\mathsf{P}=\NP$), since
the \myproblem{$k$-unpopularity matching} problem is $\NP$-hard even for strict rankings \cite{McCutchen}. Note that the \myproblem{$k$-unpopularity margin branching} problem is %
polynomial-time solvable when preferences are weak rankings, as shown in~\cite{KKMSS20}, but this does not contradict the above fact:
branchings and matchings are both special cases of common independent sets (where one matroid is a partition matroid), but neither of them includes the other.

\section*{Acknowledgments}
We are grateful for inspiring discussions on the popular arborescence problem to Chien-Chung Huang, Satoru Iwata, Tam\'as Kir\'aly, Jannik Matuschke, and Ulrike Schmidt-Kraepelin.
We thank the anonymous reviewers for their valuable comments.
Telikepalli Kavitha is supported by the Department of Atomic Energy, Government  of India, under project no.\ RTI4001.
Kazuhisa Makino is partially supported by JSPS KAKENHI Grant Numbers JP20H05967, JP19K22841, and JP20H00609. 
Ildik\'o Schlotter is supported by the Hungarian Academy of Sciences under its Momentum Programme (LP2021-2) and its J\'anos Bolyai Research Scholarship, and by the Hungarian Scientific Research Fund (OTKA grant K124171).
Yu Yokoi is supported by JST PRESTO Grant Number JPMJPR212B.
This work was partially supported by the joint project of Kyoto University and Toyota Motor Corporation, titled ``Advanced Mathematical Science for Mobility Society.''

\begin{appendices}
\section{Examples of Algorithm Execution}
\label{app:examples}
We illustrate how %
Algorithm~\ref{alg:pop-arb} works using some examples.
We provide three instances of the \myproblem{popular arborescence} problem. 
In all of these instances, a digraph is given as $G=(V\cup \{r\}, E)$ with $V=\{a, b, c, d\}$, and each node~$v \in V$ has a strict preference on $\delta(v)$.
For better readability, for a multichain~$\C=\{C_1,\dots,C_p\}$ with $C_1 \subseteq \dots \subseteq C_p$ we will also use the notation $\langle C_1,\dots,C_p\rangle$.

\subsection{Example~1.}
This instance is similar to the one illustrateded in Section~\ref{sec:intro}; the only difference is that now the edge $(r, d)$ is deleted.
In contrast to the case where $(r,d)$ exists, this instance admits a popular arborescence, which is found by Algorithm~\ref{alg:pop-arb} as follows.

\begin{minipage}[c]{0.58\textwidth}
The preference orders for the four vertices are as follows:
		\begin{align*}
                & (b,a) \succ_a (c,a) \succ_a (r,a) \\
                & (a,b) \succ_b (d,b) \succ_b (r,b)\\
                & (d,c) \succ_c (a,c) \succ_c (r,c)\\
                & (c,d) \succ_d (b,d).\\
                 \end{align*}
\end{minipage}
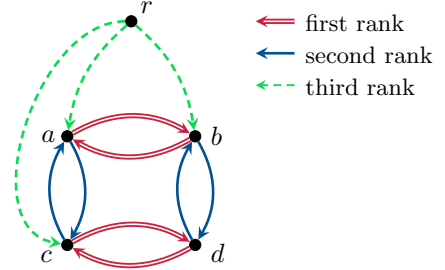
\begin{minipage}{0.35\textwidth}
\vspace{-2mm}
\input{tikz_instance1}
\end{minipage}

%This instance is similar to the one described in Section~\ref{sec:intro}; the only difference is that now the edge $(r, d)$ is deleted.
%In contrast to the case where $(r,d)$ exists, this instance admits a popular arborescence, which is found by Algorithm~\ref{alg:pop-arb} as follows. 

For convenience, we denote by $E_1$, $E_2$, and $E_3$ the sets of the first, second and third choice edges, respectively. That is, 
$E_1=\{(b,a), (a, b), (d,c), (c,d)\}$, $E_2=\{(c,a), (d,b), (a, c), (b,d)\}$, and $E_3=\{(r, a), (r, b), (r,c)\}$.

\medskip

\paragraph{Algorithm Execution.} Below we describe the steps in our algorithm.
\begin{enumerate}
\item $p=1$ and $C_1=E$. Then $E(\C)=E_1$ and $I=\{(a, b), (c, d)\}$ is a lex-maximal branching in $E(\C)$.
Since $|I\cap C_1|=2<4=\rank(C_1)$, the set $C_1$ is updated to $\spn(I\cap C_1)=E_1$. Since $C_1=C_p$ is updated, $p$ is incremented and $E$ is added to $\C$ as $C_2$.

\item $p=2$ and  $\langle C_1, C_2\rangle=\langle E_1, E\rangle$. 
Then $E(\C)=E_1\cup E_2$ and $I=\{(a, b), (c, d), (a, c)\}$ is a lex-maximal branching in $E(\C)$.
Since $|I\cap C_1|=2=\rank(C_1)$ and $|I\cap C_2|=3<4=\rank(C_2)$, the set $C_2$ is updated to $\spn(I\cap C_2)=E_1\cup E_2$. Since $C_2=C_p$ is updated, $p$ is incremented and $E$ is added to $\C$ as $C_3$.

\item $p=3$ and  $\langle C_1, C_2, C_3\rangle=\langle E_1,\, E_1\cup E_2,\, E\rangle$. 
Then $E(\C)=\{(c,d)\}\cup E_2\cup E_3$ and $I=\{(c, d), (c, a), (d, b), (r, c)\}$ is a lex-maximal branching in $E(\C)$.
Since $|I\cap C_1|=1<2=\rank(C_1)$, the set $C_1$ is updated to $\spn(I\cap C_1)=\{(c, d), (d, c) \}$. 

\begin{minipage}{0.6\textwidth}
\item $p=3$ and  $\langle C_1, C_2, C_3\rangle =\langle \{(c,d), (d,c)\}, \,E_1\cup E_2,\, E\rangle$. 
Then we have $E(\C)=\{(r, a),(b,a), (r,b), (a, b), (r,c), (a,c), (c,d), (b,d)\}$ (all edges on the figure to the right)
and $I=\{(r, a), (a, b), (a, c), (c, d)\}$ (thick edges on the figure to the right) is a lex-maximal branching in~$E(\C)$.
Since $|I\cap C_i|=\rank(C_i)$ holds for $i=1,2,3$, the algorithm returns~$I$. 
\end{minipage}
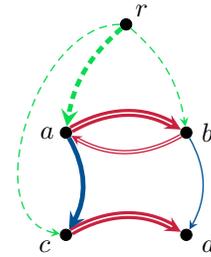
\begin{minipage}{0.2\textwidth}
\vspace{-4mm}~~~~~~~~
\input{tikz_instance1_2}
\end{minipage}
\end{enumerate}
Note that $I'=\{(r, b), (b,a), (a,c), (c,d)\}$ is also a possible output of the algorithm. Indeed, both $I$ and $I'$ are popular arborescences.

\subsection{Example~2.} We next demonstrate how the algorithm works for an instance that admits no popular arborescences.

\begin{minipage}{0.55\textwidth}
%We next demonstrate how the algorithm works for an instance which admits no popular arborescences.

\quad Consider the instance illustrated in the introduction. For the reader's convenience, we include the same figure again. As observed there, this instance has no popular arborescence.

\quad We denote by $E_1$, $E_2$, and $E_3$ the sets of the first, second and third rank edges, respectively. Note that, unlike in Example 1, here $E_3$ contains $(r, d)$.
\end{minipage}\quad
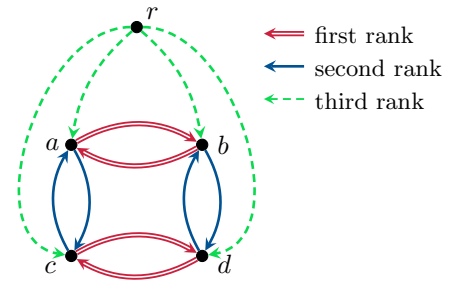
\begin{minipage}{0.4\textwidth}
\input{tikz_instance_intro}
\end{minipage}

\paragraph{Algorithm Execution}
\begin{enumerate}
\item The first step is the same as Step 1 in Example 1. That is, $p=1$, $C_1=E$, $E(\C)=E_1$, and $I=\{(a, b), (c, d)\}$ is found as a lex-maximal branching in $E(\C)$. Then,  $C_1$ is updated to $\spn(I\cap C_1)=E_1$, $p$ is incremented, and $E$ is added to $\C$ as $C_2$.

\item The second step is also the same as Step 2 in Example 1. That is, 
$p=2$, $\langle C_1,C_2\rangle=\langle E_1, E\rangle$, $E(\C)=E_1\cup E_2$, and $I=\{(a, b), (c, d), (a, c)\}$ is found as a lex-maximal branching in $E(\C)$.
Then, $C_2$ is updated to $\spn(I\cap C_2)=E_1\cup E_2$, $p$ is incremented, and $E$ is added to $\C$ as $C_3$.

\item $p=3$ and  $\langle C_1, C_2, C_3\rangle=\langle E_1, E_1\cup E_2, E\rangle$.
Then $E(\C)=E_2\cup E_3$ (compared to Example~1, here $(r,d)$ is included while $(c,d)$ is excluded) and $I=\{(a, c), (b, d), (r, a), (r,b)\}$ is a lex-maximal branching in $E(\C)$. Since $|I\cap C_1|=0<2=\rank(C_1)$, the set $C_1$ is updated to $\spn(I\cap C_1)=\emptyset$.

\item $p=3$ and  $\langle C_1, C_2, C_3\rangle =\langle \emptyset, E_1\cup E_2, E\rangle$.
Then $E(\C)=E_1\cup E_3$ and $I=\{(a, b), (c, d), (r, a), (r,c)\}$ is a lex-maximal branching in $E(\C)$. Since $|I\cap C_1|=\rank(C_1)$ and $|I\cap C_2|=2<3=\rank(C_2)$, the set $C_2$ is updated to $\spn(I\cap C_2)=E_1$.

\item $p=3$ and  $\langle C_1, C_2, C_3\rangle=\langle\emptyset, E_1, E\rangle$.
Then $E(\C)=E_1\cup E_2$ and $I=\{(a, b), (c, d), (a, c)\}$ is a lex-maximal branching in $E(\C)$. (Observe that these $E(\C)$ and $I$ are the same as Step 2.) Since $|I\cap C_i|=\rank(C_i)$ for~$i=1,2$ and $|I\cap C_3|=3<4=\rank(C_3)$, the set $C_3$ is updated to $\spn(I\cap C_3)=E_1\cup E_2$, $p$ is incremented, and $E$ is added to $\C$ as $C_4$.

\item $p=4$ and $\langle C_1, C_2, C_3, C_4\rangle=\langle\emptyset, E_1, E_1\cup E_2, E\rangle$.
Then, as in Step 3, $E(\C)=E_2\cup E_3$ and $I=\{(a, c), (b, d), (r, a), (r,b)\}$ is a lex-maximal branching in $E(\C)$. Since $|I\cap C_1|=\rank(C_1)$ and $|I\cap C_2|=0<2=\rank(C_2)$, the set $C_2$ is updated to $\spn(I\cap C_2)=\emptyset$.

\item $p=4$ and $\langle C_1, C_2, C_3, C_4\rangle=\langle\emptyset, \emptyset, E_1\cup E_2, E\rangle$. By the same argument as in Step 4, the set $C_3$ is updated to $E_1$.

\item $p=4$ and $\langle C_1, C_2, C_3, C_4\rangle=\langle\emptyset, \emptyset, E_1, E\rangle$. By the same argument as in Step 5, the set $C_4$ is updated to~$E_1\cup E_2$, $p$ is incremented, and 
$E$ is added to $\C$ as $C_5$.

\item $p=5$ and $\langle C_1, C_2, C_3, C_4, C_5\rangle=\langle\emptyset, \emptyset, E_1, E_1\cup E_2, E\rangle$. Since $p=5>4=n=|V|$, the algorithm halts with returning ``$G$ has no popular arborescence.''
\end{enumerate}

The reader might observe that whenever $C_1$ becomes empty in the algorithm, then by Lemma~\ref{prop:correctness} we can conclude that the instance admits no popular arborescence, since the dual certificate contains only non-empty sets (Lemma~\ref{lem:chain}) 
and hence $D_1 \subseteq C_1=\emptyset$ is not possible. 
Therefore, we could in fact stop the algorithm already in Step~3 when $C_1$ gets updated to~$\emptyset$. Nevertheless, the algorithm will reach a correct answer even without using this observation, as illustrated by the above example.

\subsection{Example 3.}
We next provide an example that shows the importance of considering multichains. 
During the algorithm's execution on this instance, $\C$ does become a multichain that is not a chain. 

\medskip
\begin{minipage}[c]{0.58\textwidth}
\quad The preferences of the four vertices are as follows:
		\begin{align*}
                & (b,a) \succ_a  (r,a) \\
                & (c,b)  \succ_b (a,b)\\
                & (d,c) \succ_c (b,c)\\
                & (c,d)\\            
		\end{align*}
\end{minipage}\quad
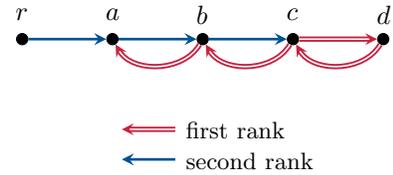
\begin{minipage}{0.35\textwidth}
\input{tikz_instance3}
\end{minipage}

\noindent where $(c,d)$ is the unique incoming edge of $d$. 
For convenience, we denote  by $E_{abcd}$, $E_{bcd}$, and $E_{cd}$ the edge sets of the induced subgraphs for the vertex sets $\{a,b,c,d\}$, $\{b,c,d\}$, and $\{c,d\}$, respectively.
That is, $E_{abcd}=E\setminus \{(r,a)\}$, $E_{bcd}=\{(b,c), (c,b), (c,d), (d,c)\}$, and $E_{cd}=\{(c,d), (d,c)\}$.
Note that $\{(r,a), (a,b), (b,c), (c,d)\}$ is the unique arborescence in this instance, and hence it is a popular arborescence.

\medskip

\paragraph{Algorithm Execution}
\begin{enumerate}
\item $p=1$ and $C_1=E$. Then $E(\C)=\{(b,a), (c,b), (d,c), (c,d)\}$ and $I=\{(b, a), (c, b), (c,d)\}$ is a lex-maximal branching in $E(\C)$.
Since $|I\cap C_1|=3<4=\rank(C_1)$, the set $C_1$ is updated to $\spn(I\cap C_1)=E_{abcd}$. Since $C_1=C_p$ is updated, $p$ is incremented and $E$ is added to $\C$ as $C_2$.

\medskip

\begin{minipage}{0.58\textwidth}
\item $p=2$ and  $\langle C_1, C_2\rangle=\langle E_{abcd}, E\rangle$ (shown by braces on the right). Then $E(\C)=\{ (r,a), (b,a), (c,b), (d,c), (c,d)\}$ (all edges on the right) and $I=\{(b, a), (c, b), (c,d)\}$ (thick edges on the right) is a lex-maximal branching in $E(\C)$.
Since $|I\cap C_1|=\rank(C_1)$ and $|I\cap C_2|=3<4=\rank(C_2)$,  $C_2$ is updated to $\spn(I\cap C_2)=E_{abcd}$. Since $C_2=C_p$ is updated, $p$ is incremented and $E$ is added to $\C$ as $C_3$.
\end{minipage}\quad
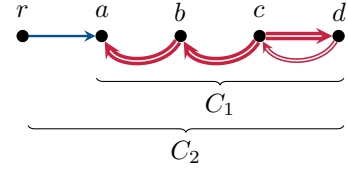
\begin{minipage}{0.35\textwidth}
\input{tikz_instance3_1}
\end{minipage}

\medskip

\begin{minipage}{0.58\textwidth}
\item $p=3$ and  $\langle C_1, C_2, C_3\rangle=\langle E_{abcd},\, E_{abcd},\, E\rangle$ (so \mbox{$C_1=C_2$}). 
Then $E(\C)=\{(r,a), (c,b), (d,c), (c,d)\}$. Note that $(b,a)$ is not in $E(\C)$ as $\lev_{\C}((b,a))=1$ while $\lev_{\C}((r,a))=3$.  $I=\{(r, a), (c, b), (c, d)\}$ is a lex-maximal branching in $E(\C)$.
Since $|I\cap C_1|=2<3=\rank(C_1)$, the set $C_1$ is updated to $\spn(I\cap C_1)=E_{bcd}$. 
\end{minipage}\quad
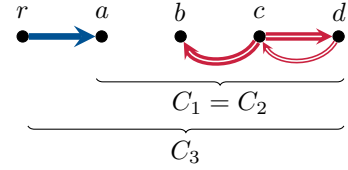
\begin{minipage}{0.35\textwidth}
\input{tikz_instance3_2}
\end{minipage}

\medskip

\begin{minipage}{0.58\textwidth}
\item $p=3$ and  $\langle C_1, C_2, C_3\rangle=\langle E_{bcd},\, E_{abcd},\, E\rangle$. Then,
$E(\C)=E\setminus \{(b,c)\}$ and $I=\{(b, a), (c, b), (c, d)\}$ is a lex-maximal branching in $E(\C)$.
Since $|I\cap C_i|=\rank(C_i)$ for $i=1,2$ and $|I\cap C_3|=3<4=\rank(C_3)$, the set $C_3$ is updated to $\spn(I\cap C_3)=E_{abcd}$. 
Since $C_3=C_p$ is updated, $p$ is incremented and $E$ is added to $\C$ as $C_4$.
\end{minipage}\quad
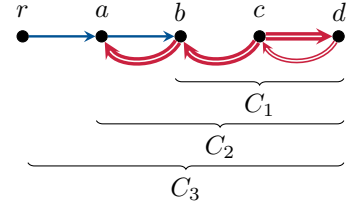
\begin{minipage}{0.35\textwidth}
\input{tikz_instance3_3}
\end{minipage}
\medskip

\begin{minipage}[t]{0.58\textwidth}
\vspace{0pt}
\item $p=4$ and  $\langle C_1, C_2, C_3, C_4\rangle=\langle E_{bcd},\, E_{abcd},\, E_{abcd}, E\rangle$. 
 $E(\C)=E\setminus \{(b,a), (b,c)\}$ and $I=\{(r, a), (c, b), (c, d)\}$ is a lex-maximal branching in $E(\C)$.
Since $|I\cap C_1|=\rank(C_1)$ and $|I\cap C_2|=2<3=\rank(C_2)$, the set $C_2$ is updated to $\spn(I\cap C_2)=E_{bcd}$. 
\end{minipage}\quad
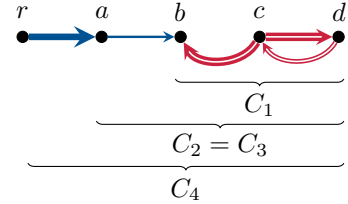
\begin{minipage}[t]{0.35\textwidth}
\vspace{-4pt}
\input{tikz_instance3_4}
\end{minipage}

\medskip
\begin{minipage}[t]{0.58\textwidth}
\vspace{0pt}
\item $p=4$ and  $\langle C_1, C_2, C_3, C_4\rangle=\langle E_{bcd},\, E_{bcd},\, E_{abcd}, E\rangle$. Then 
 $E(\C)=E\setminus \{(b,c), (c,b)\}$ and $I=\{(r,a), (a, b), (c, d)\}$ is a lex-maximal branching in $E(\C)$.
Since $|I\cap C_1|=1<2=\rank(C_1)$, the set $C_1$ is updated to $\spn(I\cap C_1)=E_{cd}$. 
\end{minipage}\quad
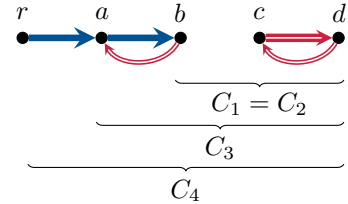
\begin{minipage}[t]{0.35\textwidth}
\vspace{-4pt}
\input{tikz_instance3_5}
\end{minipage}

\medskip
\begin{minipage}[b]{0.58\textwidth}
\item $p=4$ and  $\langle C_1, C_2, C_3, C_4\rangle=\langle E_{cd},\, E_{bcd},\, E_{abcd}, E\rangle$. Then 
 $E(\C)=E$ and $I=\{(r,a),(a,b), (b,c),(c, d)\}$ is a lex-maximal branching in $E(\C)$.
Since $|I\cap C_i|=\rank(C_i)$ holds for $i=1,2,3,4$, the algorithm returns~$I$. 
\end{minipage}\quad
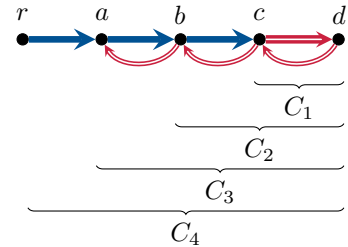
\begin{minipage}{0.35\textwidth}
\input{tikz_instance3_6}
\end{minipage}

\end{enumerate}

\switch{}{  %%%%%%%%%%%%%%%%%%%%%%%%%%%%% the beginning of switch (1)
\section{Deferred proofs}
\label{app:proofs}

Here we present all proofs that we omitted from the main body of the paper, namely, from Section~\ref{sec:hardness-min-unpop-margin}. 
For the convenience of the reader, we re-state each claim before providing its proof.

%We next present all missing proofs from Section~\ref{sec:hardness-min-unpop-margin}.

\clmmargindual*
\begin{proof}
We need to verify that 
\begin{equation}
\label{eqn:LP2-feasible}
|\{C \in \C: e \in C\}| + \alpha_v \geq \wt_A(e)    
\end{equation}
 holds for each edge~$e$ entering some vertex~$v$.
First assume $e \in C_1$, in which case $e$ is contained in three sets of~$\C$. If~$e \in C_1\cap A$, then $\wt_A(e)=0$ and $\alpha_v = -3$ ensures~(\ref{eqn:LP2-feasible}). 
If $e \in C_1 \setminus A$, then $\wt_A(e)=1$ but $\alpha_v=-2$, so~(\ref{eqn:LP2-feasible}) is again satisfied. 

Second, assume $e \in C_2$, in which case $e$ is contained in two sets from~$\C$. 
If~$e \in C_2 \cap A$, then $\alpha_v=-2$, 
which implies (\ref{eqn:LP2-feasible}). 
If $e \in C_2 \setminus A$, then we distinguish between two cases: if $v = u_0$ for some $u \in U$, then $\wt_A(e)=1$ and $\alpha_v=-1$; otherwise $\wt_A(e)=-1$ and $\alpha_v \geq -3$ (note that here we used that all vertices $a_{u,i}$ prefer $(u_0,a_{u,i})$ to~$(u_1,a_{u,i})$). Hence, $e$ again satisfies~(\ref{eqn:LP2-feasible}).

Third, assume $e \in C_3$, in which case $e$ is contained in one set from~$\C$. Let $G_u$ be the gadget entered by~$e$. 
If $e=(r,u_0) \in A$, then $\wt_A(e)=0$ and $\alpha_v=\alpha_{u_0}=-1$, and thus (\ref{eqn:LP2-feasible}) holds. 
Otherwise $\wt_A(e)=-1$.
Let $T$ be the set in~$\T$ containing~$u$.
Note that either $v=u_1$ or $v=b_{u,j}$ for some $j \neq \sigma(u)$, 
because all edges entering $b_{u,\sigma(u)}$ are contained in $C_2$, since they each originate in some gadget $G_v$ with $v \in T$.
Therefore, we have $\alpha_v = -2$ in both cases, 
which implies  (\ref{eqn:LP2-feasible}) for~$e$.
\end{proof}

\clmmarginunique*
\begin{proof}
Assume for the sake of contradiction that the claim does not hold for some $u \in U^\star$; this means that $A$ must reach~$G_u$ through an edge $e \in E_3$ pointing to some vertex of~$B_u$; let $b_{u,j}$ denote this vertex. 
Let $u_h$ be the vertex where $B$ enters~$\{u_0,u_1\}$; then $(u_h,u_{1-h}) \in B$. 
Define~$B'$ as follows: 
\begin{align*}
    B'= & B \setminus \{\delta(x): x \in V(G)\} \\
    & \cup 
\{ (r,u_{1-h}),(u_{1-h},u_h), (u_0,a_{u,j+1}), (a_{u,j+1},b_{u,j+1})\} \\
    & \cup (C_u \setminus \delta(a_{u,j+1}) \setminus \delta(b_{u,j+1}))
\end{align*}
where indices are taken modulo~3 (so $a_{u,4}=a_{u,1}$ and $b_{u,4}=b_{u,1}$).

Observe that $B'$ is an arborescence.
If $(u_0,a_{u,j+1}) \notin A$, then $(b_{u,j+1},a_{u,j+1}) \in A$, and thus vertices $u_h$, $b_{u,j}$, and $b_{u,j+1}$ all prefer~$B'$ to~$A$, while vertices $u_{1-h}$ and $a_{u,j+1}$ prefer~$A$ to~$B'$. 
If $(u_0,a_{u,j+1}) \in A$, then vertices $u_h$ and $b_{u,j}$ prefer~$B'$ to~$A$, vertex $u_{1-h}$ prefers~$A$ to~$B'$, while $a_{u,j+1}$ and $b_{u,j+1}$ are indifferent between them; note $(a_{u,j+1},b_{u,j+1}) \in A \cap B'$. 
Furthermore, if $b_{u,j-1}$ prefers~$A$ to~$B'$, then $(a_{u,j-1},b_{u,j-1}) \in A$, and therefore $a_{u,j-1}$ prefers~$B'$ to~$A$. Summing up all these facts, we obtain $\phi_{V(G_u)}(B',A)-\phi_{V(G_u)}(A,B')\geq 1$ which in turn implies
$\phi(B',A)-\phi(A,B')>\phi(B,A)-\phi(A,B)$, a contradiction to our choice of~$B$. 
\end{proof}

\clmmargindescendant*
\begin{proof}    
 Define $B_f$ as follows:
\begin{align*}
    B_f=& B  \setminus \{\delta(x): x \in V(C_u) \textrm{ or } x=\hat{u} \}  
     \cup  
\{ f, (a_{u,\sigma(u)},\hat{u})\} \cup (C_u \setminus \delta(b_{u,\sigma(u)})).
\end{align*}
Observe that there is an edge from $\hat{u}$ to the other vertex of~$\{u_0,u_1\}$ shared by~$A$ and~$B_f$.

Suppose that $B_f$ is an arborescence.
Note that vertices~$\hat{u}$ and $a_{u,\sigma(u)}$ prefer~$B_f$ to~$A$, while vertex~$b_{u,\sigma(u)}$ prefers~$A$ to~$B_f$ (because $(a_{u,\sigma(u)},b_{u,\sigma(u)}) \in A$).  Furthermore, if some $b_{u,i}$, $i \neq \sigma(u)$, prefers~$A$ to~$B_f$, then $a_{u,i}$ prefers~$B_f$ to~$A$. 
Hence, $\phi_{V(G_u)}(B_f,A)-\phi_{V(G_u)}(A,B_f)\geq 1$, which implies also $\phi(B_f,A)-\phi(A,B_f)>\phi(B,A)-\phi(A,B)$, contradicting our choice of~$B$. Hence, $B_f$ cannot be an arborescence, which can only happen if the tail of~$f$ is a descendant of~$\hat{u}$ in~$B$.
\end{proof}

\section{Extensions and Related Results}
\label{sec:discussion}

\paragraph{Popularity under size constraints.}
As mentioned in the introduction, the \myproblem{popular largest common independent set} problem can be reduced to the \myproblem{popular common base} problem. 
More generally, we can reduce the \myproblem{popular size $[\ell, u]$ common independent set} problem to the \myproblem{popular common base} problem, where the goal of the former problem is to find a common independent set that is popular within the set of all common independent sets whose size is at least $\ell$ and at most $u$ (if such a solution exists).

We now describe the reduction. Let $E=E_1\cupdot \cdots\cupdot E_n$ be a given partition of $E$ and $M=(E, \I)$ be a given matroid. 
We define a new instance as follows. For each $i\in \{1, \dots, n\}$, we create a new element $e_i$ and extend the domain of $\succ_i$ to $E_i\cup \{e_i\}$, where $e_i$ is the unique worst element. The new partition is defined as $E':=\bigcupdot_{i=1}^{n}(E_i\cup\{e_i\})$.
We define a new matroid $M'=(E', \I')$ by \[\I'=\{\,X\subseteq E':X\cap E\in \I,\, |X\cap E|\leq u,\, |X\cap \{e_1, \dots, e_n\}|\leq n-\ell,\, |X|\leq n\,\}.\] Note that we can assume that the rank of $M$ (i.e., the size of a base in $M$) is at least $\ell$ since otherwise the given instance clearly has no solution. Therefore, the rank of $M'$ is $n$. 

There exists a one-to-one correspondence between common independent sets of sizes in $[\ell, u]$ in the original instance and common bases of the new instance. 
Suppose $I$ is a common independent set with $\ell\leq |I|\leq u$ in the original instance. Let $B$ be obtained from $I$ by adding $e_i$ for any unassigned $i\in \{1,\dots, n\}$ (that is, where~$I \cap E_i = \emptyset$). Then, $B$ is a common base in the new instance. Conversely, given a common base of the new instance, we can obtain a common independent set satisfying the size constraint by projecting out the dummy elements. 
Furthermore, $\phi(I, I')=\phi(B, B')$ holds for any common independent sets $I$ and $I'$ of the original instance and their corresponding bases $B$ and $B'$. Thus, the reduction is completed.

The reduction used in Section~\ref{sec:colorful} (to reduce the popular colorful forest problem to the popular colorful base problem) is a special case of this reduction where $u=n$ and $\ell=0$.

\medskip

\paragraph{Popularity under category-wise size constraints.}
We can also use our popular common base algorithm (Algorithm~\ref{alg:pop-arb}) to solve the problem of finding a common independent set that is popular under a kind of diversity constraints.

Similarly to the above, suppose that a partition $E=E_1\cupdot \cdots\cupdot E_n$ and a matroid $M=(E, \I)$ are given. We regard $\{1, \dots, n\}$ as the set of agents. 
Suppose that the set $\{1, \dots n\}$ is partitioned into $q$ categories $P_1\cupdot \cdots \cupdot P_q$, and each category $P_k$ is associated with integers $\ell_k$ and $u_k$ where $\ell_k\leq u_k$. We call a common independent set $X\subseteq E$ {\em admissible} if, for each $k=1,\dots, q$, we have $\ell_k\leq |\set{i\in P_k: E_i\cap X\neq \emptyset}|\leq u_k$. 
That is, a set $X$ is admissible if, among the agents in each category $P_k$, at least $\ell_k$ and at most $u_k$ agents are assigned an element. 

The problem of finding a common independent set that is popular within the set of admissible common independent sets can be reduced to the \myproblem{popular common base} problem as follows.
Similarly to the case of size constraints above, for each $i\in \{1,\dots, n\}$, we introduce a dummy element $e_i$ that is worst in $i$'s preferences. 
Moreover, for each category $P_k$, we create a set~$D_k$ of dummy agents with $|D_k|=u_k-\ell_k$.
%$g_k:=u_k-\ell_k$ dummy agents and denote by $D_k$ the set of these agents. 
With each agent $j\in D_k$ we associate a set $\{f_{j}, g_{j}\}$ of two new elements, and these are tied in the preferences of~$j$, 
that is, $f_j \not\succ_j g_j$ and $g_j \not\succ_j f_j$.
%$\succeq_{j}$.
Thus, the new ground set is $E^*=\bigcup_{i=1}^n(E_i\cup \{e_i\})\cup \bigcup_{j\in D_1\cup \cdots \cup D_q}\{f_j, g_j\}$, and its partitions are the sets~$E_i \cup \{e_i\}$ for $i \in \{1,\dots,n\}$ and the sets $\{f_j,g_j\}$ for $j \in D_k$ and $k \in \{1,\dots,q\}$. 

We define a matroid on~$E^*$. 
First, for $k=1,\dots, q$, let $F_k:=\set{e_i:i\in P_k}\cup \set{f_j:j\in D_k}$ and let $(F_k, \I_k)$ be a uniform matroid defined by $\I_k=\{\,X\subseteq F_k: |X|\leq |P_k|-\ell_k\,\}$.

Next, let $E':=E\cup \set{g_j:j\in D_1\cup \cdots \cup D_q}$ and define a matroid $(E', \I')$ 
as the truncation of the direct sum of $M$ and the free matroid on $\set{g_j:j\in D_1\cup \cdots \cup D_q}$, that is, 
$\I':=\{\,X\subseteq E':X\cap E\in \I, |X|\leq \sum_{k=1}^q u_k\,\}$.
Let $(E^*, \I^*)$ be the direct sum of all these matroids, i.e., $\I^*$ is defined as \[\textstyle \I^*=\{\,X\subseteq E^*:X\cap E\in \I, ~|X\cap E'|\leq \sum_{k=1}^q u_k, ~|X\cap F_k|\leq |P_k|-\ell_k \, \textrm{ for }k=1,\dots, q\,\}.\]

We can assume that the size of a base in $(E, \I)$ is at least $\sum_{k=1}^q \ell_k$ since otherwise the instance clearly has no admissible set. As we have $|\set{g_j:j\in D_1\cup \cdots \cup D_q}|=\sum_{k=1}^q (u_k-\ell_k)$, the size of a base in the matroid~$(E', \I')$ is exactly $\sum_{k=1}^q u_k$. Also, the size of a base in each $(F_k, \I_k)$ is $|P_k|-\ell_k$ (since $|F_k|=|P_k|+u_k-\ell_k$). Thus, the size of a base of the matroid $(E^*, \I^*)$ is $\sum_{k=1}^q (|P_k|+u_k-\ell_k)$, which equals the number of agents.

We now explain how to transform admissible common independent sets of the original instance to common bases of the new instance, and vise versa.
Let $I$ be an admissible common independent set of the original instance. 
For each $k=1, \dots, q$, let $C_k \subseteq P_k$ be the set of agents $i$ in $P_k$ with $I \cap E_i = \emptyset$. Since $I$ is admissible, $|P_k|-u_k\leq |C_k|\leq |P_k|-\ell_k$.
Set $B=I$ and augment $B$ by adding elements in the following manner. 
For all agents in $C_k$, add the corresponding $e_i$ elements to~$B$. 
Note that $|P_k|-\ell_k-|C_k|$ is at least~$0$ and at most~$u_k-\ell_k$. Take $|P_k|-\ell_k-|C_k|$ agents~$j$ from $D_k$ arbitrarily and add the corresponding~$f_j$ elements to $B$. For the remaining $|D_k|-(|P_k|-\ell_k-|C_k|)=u_k-|P_k|+|C_k|$ agents~$j$ in $D_k$, we add the corresponding~$g_j$ elements to $B$. Thus, all agents are assigned elements. Furthermore, we see that the set $B$ satisfies $B\cap E\in \I$, $|B\cap E'|=|I|+\sum_{i=1}^n(u_k-|P_k|+|C_k|)=\sum_{k=1}^q u_k$ (note that $\sum_{k=1}^q(|P_k|-|C_k|)=|I|$), and $|B\cap F_k|= |P_k|-\ell_k$ for each $k=1,\dots, q$. Thus, $B$ is a common base in the new instance.

Conversely, let $B$ be a common base of the new instance and $I$ be obtained by deleting all dummy elements in $B$. Clearly $I$ is a common independent set of the original instance. 
As $B$ is a base in $\I^*$, we have $|B\cap F_k|= |P_k|-\ell_k$ for each $k=1,\dots, q$. Since $F_k=\set{e_i:i\in P_k}\cup \set{f_j:j\in D_k}$, this implies $|B\cap \set{e_i:i\in P_k}|\leq |P_k|-\ell_k$. As $|\set{f_j:j\in D_k}|=u_k-\ell_k$, it also follows that $|B\cap \set{e_i:i\in P_k}|\geq |P_k|-u_k$. Thus, we have $|P_k|-u_k\leq |B\cap \{\,e_i:i\in P_k\}|\leq |P_k|-\ell_k$, which is equivalent to $\ell_k\leq |\set{i\in P_k:B\cap E_i\neq \emptyset}|\leq u_k$. Thus, $I$ is admissible in the original instance.

We can also observe that $\phi(I, I')=\phi(B, B')$ holds for any admissible common independent sets $I$ and $I'$ of the original instance and their corresponding bases $B$ and $B'$ in the new instance. Therefore, a popular admissible common independent set in the original instance corresponds to a popular common base of the new instance.

\medskip

\paragraph{Popular fractional solutions.}
The notion of popularity can be extended to fractional solutions, or equivalently, probability distributions over integral
solutions. A fractional/mixed solution $x$ is popular if there is no fractional (in fact, integral) solution more popular than $x$. 

It was shown in \cite{KMN09} using the minimax theorem that popular mixed matchings always exist and such a fractional/mixed
matching can be computed in polynomial time. The same proof shows that a popular fractional (largest) common
independent set always exists and such a fractional solution can be computed in polynomial time
by optimizing over the matroid intersection polytope. 

An integral solution $I$ is \emph{strongly popular} if $\phi(I,I') > \phi(I',I)$ for all solutions $I' \ne I$. 
As observed in \cite{BB20} in the context of matchings, if a strongly popular solution exists, then it has to be a unique popular 
fractional solution. Thus there is a polynomial-time algorithm for the strongly popular (largest) common independent set problem. 
}  %%%%%%%%%%%%%%%%%%%%%%%%%%%%% the end of switch (1)

\end{appendices}

\bibliographystyle{abbrv}
\bibliography{sample}

\end{document}

%% file: tikz_instance_intro.tex
\begin{tikzpicture}[xscale=0.87, yscale=0.87]
  % vertices
  \node[draw, circle, fill=black, inner sep=1.5pt] (r) at (0, 0) {};
  \node[draw, circle, fill=black, inner sep=1.5pt] (a) at (-1, -1.8) {};
  \node[draw, circle, fill=black, inner sep=1.5pt] (b) at (1, -1.8) {};
  \node[draw, circle, fill=black, inner sep=1.5pt] (c) at (-1, -3.5) {};
  \node[draw, circle, fill=black, inner sep=1.5pt] (d) at (1, -3.5) {};

  % names of vertices
  \node[right, yshift=1mm] at (r.north) {$r$};
  \node[left,   xshift=0.5mm] at (a.west) {$a$};
  \node[right] at (b.east) {$b$};
  \node[left,   yshift=-1.5mm] at (c.west) {$c$};
  \node[right,yshift=-1mm] at (d.east) {$d$};

  % arcs
  \draw[->, >=stealth, line width=1.0pt, densely dashed, green] (r) to[bend right=20] (a); %node[pos=0.95, above, xshift=0.5mm] {\scriptsize 3} (a)
  \draw[->, >=stealth, line width=1.0pt, densely dashed, green] (r) to[bend left=20] (b);
  \draw[->, >=stealth, line width=1.0pt, densely dashed, green] (r) to[out=180, in=170] (c);
  \draw[->, >=stealth, line width=1.0pt, densely dashed, green] (r) to[out=0, in=370] (d);

  \draw[->, >=stealth, line width=0.7pt, double, red] (a) to[bend left=30] (b);
  \draw[->, >=stealth, line width=0.7pt, double, red] (b) to[bend left=30] (a);
  \draw[->, >=stealth, line width=0.7pt, double, red] (c) to[bend left=30] (d);
  \draw[->, >=stealth, line width=0.7pt, double, red] (d) to[bend left=30] (c);

  \draw[->, >=stealth, line width=1.0pt, blue] (a) to[bend left=30] (c);
  \draw[->, >=stealth, line width=1.0pt, blue] (c) to[bend left=30] (a);
  \draw[->, >=stealth, line width=1.0pt, blue] (b) to[bend left=30] (d);
  \draw[->, >=stealth, line width=1.0pt, blue] (d) to[bend left=30] (b);

  % explanation of arcs
  \node (f1) at (2.7, -0.1) {};
  \node (f2) at (1.8, -0.1) {};
  \node[right, xshift=-6pt, yshift=-4pt] at (f1.north) {\small ~first rank};
  \draw[->, >=stealth, line width=0.7pt, double, red] (f1) -- (f2) ;
  \node (s1) at (2.7, -0.6) {};
  \node (s2) at (1.8, -0.6) {};
  \node[right, xshift=-6pt, yshift=-4pt] at (s1.north) {\small ~second rank};
  \draw[->, >=stealth, line width=1.0pt, blue] (s1) -- (s2) ;
  \node (t1) at (2.7, -1.1) {};
  \node (t2) at (1.8, -1.1) {};
  \node[right, xshift=-6pt, yshift=-4pt] at (t1.north) {\small ~third rank};
  \draw[->, >=stealth, line width=0.7pt, densely dashed, green] (t1) -- (t2) ;

\end{tikzpicture}

%% file: tikz_instance_mincost_1.tex
    \begin{tikzpicture}[xscale=0.5, yscale=0.5]
    
  % vertices
  \node[draw, circle, fill=black, inner sep=1.5pt] (r) at (0, 0) {};
  \node[draw, circle, fill=black, inner sep=1.5pt] (w) at (2.5, 0) {};
  \node[draw, circle, fill=black, inner sep=1.5pt] (eu) at (6.5, 1) {};
  \node[draw, circle, fill=black, inner sep=1.5pt] (ev) at (6.5, -1) {};
  \node[draw, circle, fill=black, inner sep=1.5pt] (u1) at (10, 4.5) {};
  \node[draw, circle, fill=black, inner sep=1.5pt] (u0) at (10, 2.5) {};  
  \node[draw, circle, fill=black, inner sep=1.5pt] (v0) at (10, -2.5) {};    
  \node[draw, circle, fill=black, inner sep=1.5pt] (v1) at (10, -4.5) {};      
  
  % names of vertices
  \node[left] at (r.west) {$r$};
  \node[xshift=-2mm, yshift=-2mm] at (w.south) {$w$};
  \node[xshift=-0.3mm, yshift=2.5mm] at (eu.north) {$e_u$};
  \node[xshift=0mm, yshift=-2.7mm] at (ev.south) {$e_v$};
  \node[right,  yshift=-0.5mm] at (u0.east) {$u_0$};
  \node[right,  yshift=0mm] at (u1.east) {$u_1$};
  \node[right,  yshift=0.5mm] at (v0.east) {$v_0$};
  \node[right] at (v1.east) {$v_1$};
  
  % subfigure label
  \node at (6, -7){{(a)}};

  % arcs    
    \draw[->, >=stealth, line width=1.0pt, densely dashed, green] (u0) to node{} node[pos=0.62, sloped] {\large $\mathbf{||}$} (eu);
    \draw[->, >=stealth, line width=1.0pt, densely dashed, green] (v0) to node{}  node[pos=0.62, sloped] {\large $\mathbf{||}$} (ev);
    \draw[->, >=stealth, line width=1.0pt, densely dashed, green, bend left=70] (r) to node{} node[pos=0.4, sloped] {\large $\mathbf{||}$} (u1);
    \draw[->, >=stealth, line width=1.0pt, densely dashed, green, bend right=70] (r) to node{} node[pos=0.4, sloped] {\large $\mathbf{||}$} (v1);
    
    \draw[->, >=stealth, line width=0.7pt, double, red, bend left=5] (eu) to node{}  node[pos=0.4, sloped] {\large $\mathbf{||}$} (w);    
    \draw[->, >=stealth, line width=0.7pt, double, red, bend right=5] (ev) to node{}  node[pos=0.4, sloped] {\large $\mathbf{||}$} (w);    
    \draw[->, >=stealth, line width=0.7pt, double, red, bend left=20] (ev) to (eu);
    \draw[->, >=stealth, line width=0.7pt, double, red, bend left=20] (eu) to (ev);  
    \draw[->, >=stealth, line width=0.7pt, double, red, bend left=20] (u1) to  (u0);
    \draw[->, >=stealth, line width=0.7pt, double, red, bend left=20] (u0) to (u1);
    \draw[->, >=stealth, line width=0.7pt, double, red, bend left=20] (v0) to  (v1);
    \draw[->, >=stealth, line width=0.7pt, double, red, bend left=20] (v1) to   (v0);

    \draw[->, >=stealth, line width=1.0pt, blue] (r) to (w);
    \draw[->, >=stealth, line width=1.0pt, blue, bend left=30] (w) to  (eu);
    \draw[->, >=stealth, line width=1.0pt, blue, bend right=30] (w) to (ev);    
    \draw[->, >=stealth, line width=1.0pt, blue, bend left=35] (w) to node{} node[pos=0.6, sloped] {\large $\mathbf{|}$} (u1);
    \draw[->, >=stealth, line width=1.0pt, blue, bend left=35] (w) to  (u0);
    \draw[->, >=stealth, line width=1.0pt, blue, bend right=35] (w) to node{}  node[pos=0.6, sloped] {\large $\mathbf{|}$} (v1);
    \draw[->, >=stealth, line width=1.0pt, blue, bend right=35] (w) to  (v0);
    
  % explanation of arcs
  \node (f1) at (14.5, 5) {};
  \node (f2) at (12, 5) {};
  \node[right, xshift=-6pt, yshift=-4pt] at (f1.north) {\small ~first rank};
  \draw[->, >=stealth, line width=0.7pt, double, red] (f1) -- (f2) ;
  \node (s1) at (14.5, 4) {};
  \node (s2) at (12, 4) {};
  \node[right, xshift=-6pt, yshift=-4pt] at (s1.north) {\small ~second rank};
  \draw[->, >=stealth, line width=1.0pt, blue] (s1) -- (s2) ;
  \node (t1) at (14.5, 3) {};
  \node (t2) at (12, 3) {};
  \node[right, xshift=-6pt, yshift=-4pt] at (t1.north) {\small ~third rank};
  \draw[->, >=stealth, line width=0.7pt, densely dashed, green] (t1) -- (t2) ;
    
  \node at (14.2, -3){{$C_1 \!:$}};    
  \draw[dual] (16.5, -3) ellipse (1.5cm and 0.8cm);
  \node at (14.2, -5){{$C_2 \!:$}};    
  \draw[thick, dotted] (16.5, -5) ellipse (1.5cm and 0.8cm) ;

  \end{tikzpicture}

%% file: tikz_instance_mincost_2.tex
\begin{tikzpicture}[xscale=0.5, yscale=0.5]

  % chain sets
  \draw[dual] (6.5, 0) ellipse (1.4cm and 2.2cm);
  \draw[dual] (10.3, -3.5) ellipse (1.4cm and 2.2cm);
  \draw[thick,dotted] (10.3, 3.5) ellipse (1.4cm and 2.2cm);
  \draw[thick,dotted,rotate around={-35:(7,-2)}] (7, -2) ellipse (5.8cm and 4cm);

  % vertices
  \node[draw, circle, fill=black, inner sep=1.5pt] (r) at (0, 0) {};
  \node[draw, circle, fill=black, inner sep=1.5pt] (w) at (2.5, 0) {};
  \node[draw, circle, fill=black, inner sep=1.5pt] (eu) at (6.5, 1) {};
  \node[draw, circle, fill=black, inner sep=1.5pt] (ev) at (6.5, -1) {};
  \node[draw, circle, fill=black, inner sep=1.5pt] (u1) at (10, 4.5) {};
  \node[draw, circle, fill=black, inner sep=1.5pt] (u0) at (10, 2.5) {};  
  \node[draw, circle, fill=black, inner sep=1.5pt] (v0) at (10, -2.5) {};    
  \node[draw, circle, fill=black, inner sep=1.5pt] (v1) at (10, -4.5) {};      
  
  % names of vertices
  \node[left] at (r.west) {$r$};
  \node[xshift=-2mm, yshift=-2mm] at (w.south) {$w$};
  \node[xshift=-0.3mm, yshift=2.5mm] at (eu.north) {$e_u$};
  \node[xshift=0mm, yshift=-2.7mm] at (ev.south) {$e_v$};
  \node[right,  yshift=-0.5mm] at (u0.east) {$u_0$};
  \node[right,  yshift=0mm] at (u1.east) {$u_1$};
  \node[right,  yshift=0.5mm] at (v0.east) {$v_0$};
  \node[right] at (v1.east) {$v_1$};  
  
  % subfigure label
  \node at (6, -7){{(b)}};
     
  % arcs    
    \draw[->, >=stealth, line width=1.0pt, densely dashed, green] (u0) to node{} node[pos=0.62, sloped] {\large $\mathbf{||}$}  (eu);
    \draw[->, >=stealth, line width=1.0pt, densely dashed, green] (v0) to node{}  node[pos=0.62, sloped] {\large $\mathbf{||}$} (ev);
    \draw[->, >=stealth, line width=1.0pt, densely dashed, green, bend left=70] (r) to node{} node[pos=0.4, sloped] {\large $\mathbf{||}$} (u1);
    \draw[->, >=stealth, line width=1.0pt, densely dashed, green, bend right=70] (r) to node{} node[pos=0.4, sloped] {\large $\mathbf{||}$} (v1);
    
    \draw[->, >=stealth, line width=0.7pt, double, red, bend left=5] (eu) to node{}  node[pos=0.4, sloped] {\large $\mathbf{||}$} (w);    
    \draw[->, >=stealth, line width=0.7pt, double, red, bend right=5] (ev) to node{}  node[pos=0.4, sloped] {\large $\mathbf{||}$} (w);    
    \draw[->, >=stealth, line width=1.8pt, double, red, bend left=20] (eu) to  (ev);
    \draw[->, >=stealth, line width=0.7pt, double, red, bend left=20] (ev) to (eu);
    \draw[->, >=stealth, line width=0.7pt, double, red, bend left=20] (u1) to  (u0);
    \draw[->, >=stealth, line width=1.8pt, double, red, bend left=20] (u0) to (u1);
    \draw[->, >=stealth, line width=0.7pt, double, red, bend left=20] (v0) to  (v1);
    \draw[->, >=stealth, line width=1.8pt, double, red, bend left=20] (v1) to   (v0);

    \draw[->, >=stealth, line width=2.5pt, blue] (r) to (w);
    \draw[->, >=stealth, line width=2.5pt, blue, bend left=30] (w) to  (eu);
    \draw[->, >=stealth, line width=1.0pt, blue, bend right=30] (w) to (ev);    
    \draw[->, >=stealth, line width=1.0pt, blue, bend left=35] (w) to node{} node[pos=0.6, sloped] {\large $\mathbf{|}$} (u1);
    \draw[->, >=stealth, line width=2.5pt, blue, bend left=35] (w) to  (u0);
    \draw[->, >=stealth, line width=2.5pt, blue, bend right=35] (w) to node{}  node[pos=0.6, sloped] {\scalebox{3}[1.2]{$\mathbf{|}$}} (v1);
    \draw[->, >=stealth, line width=1.0pt, blue, bend right=35] (w) to  (v0);
    
  \end{tikzpicture}

%% file: tikz_instance_margin.tex
    \begin{tikzpicture}[xscale=0.5, yscale=0.5]
    
  % vertices
  \node[draw, circle, fill=black, inner sep=1.5pt] (a1) at (-7.2, 0) {};
  \node[draw, circle, fill=black, inner sep=1.5pt] (b1) at (-4.8, 0) {};
  \node[draw, circle, fill=black, inner sep=1.5pt] (a2) at (-1.2, 0) {};
  \node[draw, circle, fill=black, inner sep=1.5pt] (b2) at (1.2, 0) {};
  \node[draw, circle, fill=black, inner sep=1.5pt] (a3) at (4.8, 0) {};
  \node[draw, circle, fill=black, inner sep=1.5pt] (b3) at (7.2, 0) {};  
  \node[draw, circle, fill=black, inner sep=1.5pt] (u0) at (-1.2, -4) {};    
  \node[draw, circle, fill=black, inner sep=1.5pt] (u1) at (1.2, -4) {};      
  \node[draw, circle, fill=black, inner sep=1.5pt] (r) at (0, -6) {};      
  \path (b1) ++(40:2)   node (t1) {};
  \path (b1) ++(70:1.7) node (s1) {};
  \path (b2) ++(40:2)   node (t2) {};
  \path (b2) ++(70:1.7) node (s2) {};
  \path (b3) ++(40:2)   node (t3) {};
  \path (b3) ++(70:1.7) node (s3) {};
      
  % names of vertices
  \node[above left, xshift=4mm, yshift=1mm] at (a1) {$a_{u,1}$};
  \node[above left, xshift=1mm, yshift=1mm] at (b1) {$b_{u,1}$};
  \node[above, xshift=-1mm] at (a2.north) {$a_{u,2}$};
  \node[above left, xshift=1mm, yshift=1mm] at (b2) {$b_{u,2}$};
  \node[above, xshift=-1mm] at (a3.north) {$a_{u,3}$};
  \node[above left, xshift=1mm, yshift=1mm] at (b3) {$b_{u,3}$};
  \node[left, yshift=-2mm] at (u0) {$u_0$};
  \node[right, yshift=-2mm] at (u1) {$u_1$};  
  \node[below] at (r.south) {$r$};
  
  % braces for Q sets  and C_u
  %\draw [decorate, decoration = {brace}] (-7.9,1.5) --  (-4.6,1.5) node[above,pos=0.5] {$Q_{u,1}$};
  %\draw [decorate, decoration = {brace}] (-1.9,1.5) --  (1.4,1.5) node[above,pos=0.5] {$Q_{u,2}$};
  %\draw [decorate, decoration = {brace}] (4.3,1.5) --  (7.4,1.5) node[above,pos=0.5] {$Q_{u,3}$};    
  \draw [decorate, decoration = {brace}] (9,0.5) --  (9,-1.5) node[right, xshift = 2mm, pos=0.5] {$C_u$};    
  
%  \path (b3) ++(0,-1.8) node (Cu) {$C_u$};

  % arcs
  \draw[->, >=stealth, line width=1.0pt, densely dashed, green] (t1) to  (b1);
  \draw[->, >=stealth, line width=1.0pt, densely dashed, green] (t2) to  (b2);  
  \draw[->, >=stealth, line width=1.0pt, densely dashed, green] (t3) to  (b3);    
  \draw[->, >=stealth, line width=1.0pt, densely dashed, green] (s1) to  (b1);
  \draw[->, >=stealth, line width=1.0pt, densely dashed, green] (s2) to  (b2);  
  \draw[->, >=stealth, line width=1.0pt, densely dashed, green] (s3) to  (b3);    
  \draw[->, >=stealth, line width=1.0pt, densely dashed, green] (r) to  (u0);      
  \draw[->, >=stealth, line width=1.0pt, densely dashed, green] (r) to  (u1);        

  \draw[->, >=stealth, line width=0.7pt, double, red, bend left=20] (a1) to (b1);  
  \draw[->, >=stealth, line width=0.7pt, double, red, bend left=20] (b1) to (a1);    
  \draw[->, >=stealth, line width=0.7pt, double, red, bend left=20] (a2) to (b2);  
  \draw[->, >=stealth, line width=0.7pt, double, red, bend left=20] (b2) to (a2);    
  \draw[->, >=stealth, line width=0.7pt, double, red, bend left=20] (a3) to (b3);  
  \draw[->, >=stealth, line width=0.7pt, double, red, bend left=20] (b3) to (a3);    
  \draw[->, >=stealth, line width=0.7pt, double, red, bend left=20] (u0) to (u1);  
  \draw[->, >=stealth, line width=0.7pt, double, red, bend left=20] (u1) to (u0);    
  
  \draw[<->, >=stealth, line width=1.0pt, blue, out= -100, in=150] (a1) to (u0);  
  \draw[<->, >=stealth, line width=1.0pt, blue, out= -100, in = 100] (a2) to (u0);  
  \draw[<->, >=stealth, line width=1.0pt, blue, out=55, in=-150] (u0) to (a3);  
  \draw[<->, >=stealth, line width=1.0pt, blue, in= -60, out = 120] (u1) to (a1);  
  \draw[<->, >=stealth, line width=1.0pt, blue, out=80, in=-60] (u1) to (a2);  
  \draw[<->, >=stealth, line width=1.0pt, blue, out=30, in=-100] (u1) to (a3);  
  \draw[->, >=stealth, line width=1.0pt, blue] (a3) to (b2);  
  \draw[->, >=stealth, line width=1.0pt, blue] (a2) to (b1);    
  \draw[bend left=-160, ->, >=stealth, line width=1.0pt, blue, bend left=-160] (a1) to (b3);    
    
  \end{tikzpicture}

%% file: tikz_instance_margin_2.tex
    \begin{tikzpicture}[xscale=0.5, yscale=0.5]
    
  % certificate:
  \draw[dual] (-6, 0) ellipse (2.2cm and 1.2cm);
  \draw[dual] (-0, 0) ellipse (2.2cm and 1.2cm);
  \draw[dual] (6, 0) ellipse (2.2cm and 1.2cm);
  \draw[dual] (14.5, 0) ellipse (2.0cm and 1.2cm);
  \draw[thick, dotted] (14.5, -3) ellipse (2.0cm and 1.2cm);

  % vertices
  \node[draw, circle, fill=white, inner sep=1.5pt] (a1) at (-7.2, 0) {$\scriptstyle \! -3$};
  \node[draw, circle, fill=white, inner sep=1.5pt] (b1) at (-4.8, 0) {$\scriptstyle \! -2$};;
  \node[draw, circle, fill=white, inner sep=1.5pt] (a2) at (-1.2, 0) {$\scriptstyle \! -2$};;
  \node[draw, circle, fill=white, inner sep=1.5pt] (b2) at (1.2, 0)  {$\scriptstyle \! -3$};;
  \node[draw, circle, fill=white, inner sep=1.5pt] (a3) at (4.8, 0)  {$\scriptstyle \! -3$};
  \node[draw, circle, fill=white, inner sep=1.5pt] (b3) at (7.2, 0)  {$\scriptstyle \! -2$};  
  \node[draw, circle, fill=white, inner sep=1.5pt] (u0) at (-1.2, -4) {$\scriptstyle \! -1$};  
  \node[draw, circle, fill=white, inner sep=1.5pt] (u1) at (1.2, -4) {$\scriptstyle \! -2$};        
  \node[draw, circle, fill=black, inner sep=1.5pt] (r) at (0, -6) {};      
  \path (b1) ++(40:4.2)   node (t1) {};
  \path (b1) ++(70:3.5) node (s1) {};
  \path (b2) ++(40:4.2)   node (t2) {};
  \path (b2) ++(70:3.5) node (s2) {};
  \path (b3) ++(40:4.2)   node (t3) {};
  \path (b3) ++(70:3.5) node (s3) {};
  \path (b3) ++(4,0)  node (C1) {$C_1:$};
  \path (C1) ++(0,-3) node (C2) {$C_2:$};  
  \draw[thick,dotted] plot [smooth] coordinates {(5, 3.5) (5.5,2) (9,1.5) (9,-2) (3,-5) (-3,-5) (-9,-2) (-9,1.5) (-1,2) (1, 2.5) (1.8,3.5)};
      
  % names of vertices
  \node[above left, xshift=2mm, yshift=2mm] at (a1) {$a_{u,1}$};
  \node[above left, xshift=2mm, yshift=2mm] at (b1) {$b_{u,1}$};
  \node[above, xshift=-1mm] at (a2.north) {$a_{u,2}$};
  \node[above left, xshift=2mm, yshift=3mm] at (b2) {$b_{u,2}$};
  \node[above, xshift=-1mm] at (a3.north) {$a_{u,3}$};
  \node[above left, xshift=2mm, yshift=3mm] at (b3) {$b_{u,3}$};
  \node[left, yshift=-2mm, xshift=-2mm] at (u0) {$u_0$};
  \node[right, yshift=-2mm, xshift=2mm] at (u1) {$u_1$};  
  \node[below] at (r.south) {$r$};

  % braces for Q sets  and C_u
  %\draw [decorate, decoration = {brace}] (-7.9,1.5) --  (-4.6,1.5) node[above,pos=0.5] {$Q_{u,1}$};
  %\draw [decorate, decoration = {brace}] (-1.9,1.5) --  (1.4,1.5) node[above,pos=0.5] {$Q_{u,2}$};
  %\draw [decorate, decoration = {brace}] (4.3,1.5) --  (7.4,1.5) node[above,pos=0.5] {$Q_{u,3}$};    
  %\draw [decorate, decoration = {brace}] (9,0.5) --  (9,-1.5) node[right, xshift = 2mm, pos=0.5] {$C_u$};    
  
%  \path (b3) ++(0,-1.8) node (Cu) {$C_u$};

  % arcs
  \draw[->, >=stealth, line width=1.0pt, densely dashed, green] (t1) to  (b1);
  \draw[->, >=stealth, line width=1.0pt, densely dashed, green] (t2) to  (b2);  
  \draw[->, >=stealth, line width=1.0pt, densely dashed, green] (t3) to  (b3);    
  \draw[->, >=stealth, line width=1.0pt, densely dashed, green] (s1) to  (b1);
  \draw[->, >=stealth, line width=1.0pt, densely dashed, green] (s2) to  (b2);  
  \draw[->, >=stealth, line width=1.0pt, densely dashed, green] (s3) to  (b3);    
  \draw[->, >=stealth, line width=2.5pt, densely dashed, green] (r) to  (u0);      
  \draw[->, >=stealth, line width=1.0pt, densely dashed, green] (r) to  (u1);        

  \draw[->, >=stealth, line width=0.7pt, double, red, bend left=20] (a1) to (b1);  
  \draw[->, >=stealth, line width=1.7pt, double, red, bend left=20] (b1) to (a1);    
  \draw[->, >=stealth, line width=1.7pt, double, red, bend left=20] (a2) to (b2);  
  \draw[->, >=stealth, line width=0.7pt, double, red, bend left=20] (b2) to (a2);    
  \draw[->, >=stealth, line width=0.7pt, double, red, bend left=20] (a3) to (b3);  
  \draw[->, >=stealth, line width=1.7pt, double, red, bend left=20] (b3) to (a3);    
  \draw[->, >=stealth, line width=1.7pt, double, red, bend left=20] (u0) to (u1);  
  \draw[->, >=stealth, line width=0.7pt, double, red, bend left=20] (u1) to (u0);    
  
  \draw[<->, >=stealth, line width=1.0pt, blue, out= -100, in=150] (a1) to (u0);  
  \draw[<-, >=stealth, line width=2.5pt, blue, out= -120, in = 120] (a2) to (u0);  
  \draw[->, >=stealth, line width=1.0pt, blue, out= -90, in = 90] (a2) to (u0);  
  \draw[<->, >=stealth, line width=1.0pt, blue, out=55, in=-150] (u0) to (a3);  
  \draw[<->, >=stealth, line width=1.0pt, blue, in= -60, out = 120] (u1) to (a1);  
  \draw[<->, >=stealth, line width=1.0pt, blue, out=80, in=-60] (u1) to (a2);  
  \draw[<->, >=stealth, line width=1.0pt, blue, out=30, in=-100] (u1) to (a3);  
  \draw[->, >=stealth, line width=1.0pt, blue] (a3) to (b2);  
  \draw[->, >=stealth, line width=2.5pt, blue] (a2) to (b1);    
  \draw[bend left=-160, ->, >=stealth, line width=2.5pt, blue, bend left=-160] (a1) to (b3);

  \end{tikzpicture}

%% file: tikz_instance1.tex
\begin{tikzpicture}[xscale=0.85, yscale=.85]
  % vertices
  \node[draw, circle, fill=black, inner sep=1.5pt] (r) at (0, 0) {};
  \node[draw, circle, fill=black, inner sep=1.5pt] (a) at (-1, -1.8) {};
  \node[draw, circle, fill=black, inner sep=1.5pt] (b) at (1, -1.8) {};
  \node[draw, circle, fill=black, inner sep=1.5pt] (c) at (-1, -3.5) {};
  \node[draw, circle, fill=black, inner sep=1.5pt] (d) at (1, -3.5) {};

  % names of vertices
  \node[right, yshift=1mm] at (r.north) {$r$};
  \node[left,   xshift=0.5mm] at (a.west) {$a$};
  \node[right] at (b.east) {$b$};
  \node[left,   yshift=-1.5mm] at (c.west) {$c$};
  \node[right,yshift=-1mm] at (d.east) {$d$};

  % arcs
  \draw[->, >=stealth, line width=1.0pt, densely dashed, green] (r) to[bend right=20] (a); %node[pos=0.95, above, xshift=0.5mm] {\scriptsize 3} (a)
  \draw[->, >=stealth, line width=1.0pt, densely dashed, green] (r) to[bend left=20] (b);
  \draw[->, >=stealth, line width=1.0pt, densely dashed, green] (r) to[out=180, in=170] (c);

  \draw[->, >=stealth, line width=0.7pt, double, red] (a) to[bend left=30] (b);
  \draw[->, >=stealth, line width=0.7pt, double, red] (b) to[bend left=30] (a);
  \draw[->, >=stealth, line width=0.7pt, double, red] (c) to[bend left=30] (d);
  \draw[->, >=stealth, line width=0.7pt, double, red] (d) to[bend left=30] (c);

  \draw[->, >=stealth, line width=1.0pt, blue] (a) to[bend left=30] (c);
  \draw[->, >=stealth, line width=1.0pt, blue] (c) to[bend left=30] (a);
  \draw[->, >=stealth, line width=1.0pt, blue] (b) to[bend left=30] (d);
  \draw[->, >=stealth, line width=1.0pt, blue] (d) to[bend left=30] (b);

  % explanation of arcs
  \node (f1) at (2.7, 0) {};
  \node (f2) at (1.8, 0) {};
  \node[right, xshift=-6pt, yshift=-4pt] at (f1.north) {\small ~first rank};
  \draw[->, >=stealth, line width=0.7pt, double, red] (f1) -- (f2) ;
  \node (s1) at (2.7, -0.5) {};
  \node (s2) at (1.8, -0.5) {};
  \node[right, xshift=-6pt, yshift=-4pt] at (s1.north) {\small ~second rank};
  \draw[->, >=stealth, line width=1.0pt, blue] (s1) -- (s2) ;
  \node (t1) at (2.7, -1) {};
  \node (t2) at (1.8, -1) {};
  \node[right, xshift=-6pt, yshift=-4pt] at (t1.north) {\small ~third rank};
  \draw[->, >=stealth, line width=0.7pt, densely dashed, green] (t1) -- (t2) ;
\end{tikzpicture}

%% file: tikz_instance1_2.tex
\begin{tikzpicture}[xscale=0.8, yscale=.8]
  % vertices
  \node[draw, circle, fill=black, inner sep=1.5pt] (r) at (0, 0) {};
  \node[draw, circle, fill=black, inner sep=1.5pt] (a) at (-1, -1.8) {};
  \node[draw, circle, fill=black, inner sep=1.5pt] (b) at (1, -1.8) {};
  \node[draw, circle, fill=black, inner sep=1.5pt] (c) at (-1, -3.5) {};
  \node[draw, circle, fill=black, inner sep=1.5pt] (d) at (1, -3.5) {};

  % names of vertices
  \node[right, yshift=1mm] at (r.north) {$r$};
  \node[left,   xshift=0.5mm] at (a.west) {$a$};
  \node[right] at (b.east) {$b$};
  \node[left,   yshift=-1.5mm] at (c.west) {$c$};
  \node[right,yshift=-1mm] at (d.east) {$d$};

  % arcs
  \draw[->, >=stealth, line width=1.8pt, densely dashed, green] (r) to[bend right=20] (a); %node[pos=0.95, above, xshift=0.5mm] {\scriptsize 3} (a)
  \draw[->, >=stealth, line width=0.5pt, densely dashed, green] (r) to[bend left=20] (b);
  \draw[->, >=stealth, line width=0.5pt, densely dashed, green] (r) to[out=180, in=170] (c);
%  \draw[->, >=stealth, line width=0.5pt, densely dashed] (r) to[out=0, in=370] (d);

  \draw[->, >=stealth, line width=1.2pt, double, red] (a) to[bend left=30] (b);
  \draw[->, >=stealth, line width=0.5pt, double, red] (b) to[bend left=30] (a);
  \draw[->, >=stealth, line width=1.2pt, double, red] (c) to[bend left=30] (d);
%  \draw[->, >=stealth, line width=0.4pt, double] (d) to[bend left=30] (c);

  \draw[->, >=stealth, line width=1.8pt, blue] (a) to[bend left=30] (c);
%  \draw[->, >=stealth, line width=0.4pt] (c) to[bend left=30] (a);
  \draw[->, >=stealth, line width=0.5pt, blue] (b) to[bend left=30] (d);
%  \draw[->, >=stealth, line width=0.4pt] (d) to[bend left=30] (b);
%(r, a),(b,a), (r,b), (a, b), (r,c), (a,c), (c,d), (b,d)
\end{tikzpicture}

%% file: tikz_instance3.tex
\begin{tikzpicture}[xscale=0.8, yscale=0.8]
  % vertices
  \node[draw, circle, fill=black, inner sep=1.5pt] (r) at (0, 0) {};
  \node[draw, circle, fill=black, inner sep=1.5pt] (a) at (1.5, 0) {};
  \node[draw, circle, fill=black, inner sep=1.5pt] (b) at (3, 0) {};
  \node[draw, circle, fill=black, inner sep=1.5pt] (c) at (4.5, 0) {};
  \node[draw, circle, fill=black, inner sep=1.5pt] (d) at (6, 0) {};

  % names of vertices
  \node[right, xshift=-6pt, yshift=7pt] at (r.north) {$r$};
  \node[right, xshift=-6pt, yshift=7pt] at (a.north) {$a$};
  \node[right, xshift=-6pt, yshift=7pt] at (b.north) {$b$};
  \node[right, xshift=-6pt, yshift=7pt] at (c.north) {$c$};
  \node[right, xshift=-6pt, yshift=7pt] at (d.north) {$d$};

  % arcs
  \draw[->, >=stealth, line width=1.0pt, blue] (r) to (a); 
  \draw[->, >=stealth, line width=1.0pt, blue] (a) to (b); 
  \draw[->, >=stealth, line width=1.0pt, blue] (b) to (c); 
  \draw[->, >=stealth, line width=0.7pt, double, red] (c) to (d); 

  \draw[->, >=stealth, line width=0.7pt, double, red] (b) to[bend left=60] (a);
  \draw[->, >=stealth, line width=0.7pt, double, red] (c) to[bend left=60] (b);
  \draw[->, >=stealth, line width=0.7pt, double, red] (d) to[bend left=60] (c);

  % explanation of arcs
  \node (f1) at (2.7, -1.5) {};
  \node (f2) at (1.5, -1.5) {};
  \node[right, xshift=-6pt, yshift=-4pt] at (f1.north) {\small ~first rank};
  \draw[->, >=stealth, line width=0.7pt, double, red] (f1) -- (f2) ;
  \node (s1) at (2.7, -2) {};
  \node (s2) at (1.5, -2) {};
  \node[right, xshift=-6pt, yshift=-4pt] at (s1.north) {\small ~second rank};
  \draw[->, >=stealth, line width=1.0pt, blue] (s1) -- (s2) ;
\end{tikzpicture}

%% file: tikz_instance3_1.tex
\begin{tikzpicture}[xscale=0.7, yscale=0.7]
 % braces
  \draw [decorate, decoration = {brace}]  (6.1, -0.8)--(1.4, -0.8) node[below, yshift = -0.8mm, pos=0.5] {$C_1$}; 
  \draw [decorate, decoration = {brace}]  (6.1, -1.7)--(0.1, -1.7) node[below, yshift = -0.8mm, pos=0.5] {$C_2$}; 
  % vertices
  \node[draw, circle, fill=black, inner sep=1.5pt] (r) at (0, 0) {};
  \node[draw, circle, fill=black, inner sep=1.5pt] (a) at (1.5, 0) {};
  \node[draw, circle, fill=black, inner sep=1.5pt] (b) at (3, 0) {};
  \node[draw, circle, fill=black, inner sep=1.5pt] (c) at (4.5, 0) {};
  \node[draw, circle, fill=black, inner sep=1.5pt] (d) at (6, 0) {};

  % names of vertices
  \node[right, xshift=-6pt, yshift=7pt] at (r.north) {$r$};
  \node[right, xshift=-6pt, yshift=7pt] at (a.north) {$a$};
  \node[right, xshift=-6pt, yshift=7pt] at (b.north) {$b$};
  \node[right, xshift=-6pt, yshift=7pt] at (c.north) {$c$};
  \node[right, xshift=-6pt, yshift=7pt] at (d.north) {$d$};

  % arcs
  \draw[->, >=stealth, line width=0.9pt, blue] (r) to (a); 
%  \draw[->, >=stealth, line width=0.9pt, blue] (a) to (b); 
%  \draw[->, >=stealth, line width=0.9pt, blue] (b) to (c); 

  \draw[->, >=stealth, line width=1.3pt, double, red] (c) to (d); 
  \draw[->, >=stealth, line width=1.3pt, double, red] (b) to[bend left=60] (a);
  \draw[->, >=stealth, line width=1.3pt, double, red] (c) to[bend left=60] (b);
  \draw[->, >=stealth, line width=0.6pt, double, red] (d) to[bend left=60] (c);
\end{tikzpicture}

%% file: tikz_instance3_2.tex
\begin{tikzpicture}[xscale=0.7, yscale=0.7]
 % braces
  \draw [decorate, decoration = {brace}]  (6.1, -0.8)--(1.4, -0.8) node[below, yshift = -0.5mm, pos=0.5] {$C_1=C_2$}; 
  \draw [decorate, decoration = {brace}]  (6.1, -1.7)--(0.1, -1.7) node[below, yshift = -0.8mm, pos=0.5] {$C_3$}; 
  % vertices
  \node[draw, circle, fill=black, inner sep=1.5pt] (r) at (0, 0) {};
  \node[draw, circle, fill=black, inner sep=1.5pt] (a) at (1.5, 0) {};
  \node[draw, circle, fill=black, inner sep=1.5pt] (b) at (3, 0) {};
  \node[draw, circle, fill=black, inner sep=1.5pt] (c) at (4.5, 0) {};
  \node[draw, circle, fill=black, inner sep=1.5pt] (d) at (6, 0) {};

  % names of vertices
  \node[right, xshift=-6pt, yshift=7pt] at (r.north) {$r$};
  \node[right, xshift=-6pt, yshift=7pt] at (a.north) {$a$};
  \node[right, xshift=-6pt, yshift=7pt] at (b.north) {$b$};
  \node[right, xshift=-6pt, yshift=7pt] at (c.north) {$c$};
  \node[right, xshift=-6pt, yshift=7pt] at (d.north) {$d$};

  % arcs
  \draw[->, >=stealth, line width=2.5pt, blue] (r) to (a); 
%  \draw[->, >=stealth, line width=0.9pt, blue] (a) to (b); 
%  \draw[->, >=stealth, line width=0.9pt, blue] (b) to (c); 

  \draw[->, >=stealth, line width=1.3pt, double, red] (c) to (d); 
%  \draw[->, >=stealth, line width=1.3pt, double, red] (b) to[bend left=60] (a);
  \draw[->, >=stealth, line width=1.3pt, double, red] (c) to[bend left=60] (b);
  \draw[->, >=stealth, line width=0.6pt, double, red] (d) to[bend left=60] (c);
\end{tikzpicture}

%% file: tikz_instance3_3.tex
\begin{tikzpicture}[xscale=0.7, yscale=0.7]
 % braces
  \draw [decorate, decoration = {brace}]  (6.1, -0.8)--(2.9, -0.8) node[below, yshift = -0.8mm, pos=0.5] {$C_1$}; 
  \draw [decorate, decoration = {brace}]  (6.1, -1.6)--(1.4, -1.6) node[below, yshift = -0.8mm, pos=0.5] {$C_2$}; 
  \draw [decorate, decoration = {brace}]  (6.1, -2.4)--(0.1, -2.4) node[below, yshift = -0.8mm, pos=0.5] {$C_3$}; 
  % vertices
  \node[draw, circle, fill=black, inner sep=1.5pt] (r) at (0, 0) {};
  \node[draw, circle, fill=black, inner sep=1.5pt] (a) at (1.5, 0) {};
  \node[draw, circle, fill=black, inner sep=1.5pt] (b) at (3, 0) {};
  \node[draw, circle, fill=black, inner sep=1.5pt] (c) at (4.5, 0) {};
  \node[draw, circle, fill=black, inner sep=1.5pt] (d) at (6, 0) {};

  % names of vertices
  \node[right, xshift=-6pt, yshift=7pt] at (r.north) {$r$};
  \node[right, xshift=-6pt, yshift=7pt] at (a.north) {$a$};
  \node[right, xshift=-6pt, yshift=7pt] at (b.north) {$b$};
  \node[right, xshift=-6pt, yshift=7pt] at (c.north) {$c$};
  \node[right, xshift=-6pt, yshift=7pt] at (d.north) {$d$};

  % arcs
  \draw[->, >=stealth, line width=0.9pt, blue] (r) to (a); 
  \draw[->, >=stealth, line width=0.9pt, blue] (a) to (b); 
%  \draw[->, >=stealth, line width=0.9pt, blue] (b) to (c); 

  \draw[->, >=stealth, line width=1.3pt, double, red] (c) to (d); 
  \draw[->, >=stealth, line width=1.3pt, double, red] (b) to[bend left=60] (a);
  \draw[->, >=stealth, line width=1.3pt, double, red] (c) to[bend left=60] (b);
  \draw[->, >=stealth, line width=0.6pt, double, red] (d) to[bend left=60] (c);
\end{tikzpicture}

%% file: tikz_instance3_4.tex
\begin{tikzpicture}[xscale=0.7, yscale=0.7]
 % braces
  \draw [decorate, decoration = {brace}]  (6.1, -0.8)--(2.9, -0.8) node[below, yshift = -0.8mm, pos=0.5] {$C_1$}; 
  \draw [decorate, decoration = {brace}]  (6.1, -1.6)--(1.4, -1.6) node[below, yshift = -0.5mm, pos=0.5] {$C_2=C_3$}; 
  \draw [decorate, decoration = {brace}]  (6.1, -2.4)--(0.1, -2.4) node[below, yshift = -0.8mm, pos=0.5] {$C_4$}; 
  % vertices
  \node[draw, circle, fill=black, inner sep=1.5pt] (r) at (0, 0) {};
  \node[draw, circle, fill=black, inner sep=1.5pt] (a) at (1.5, 0) {};
  \node[draw, circle, fill=black, inner sep=1.5pt] (b) at (3, 0) {};
  \node[draw, circle, fill=black, inner sep=1.5pt] (c) at (4.5, 0) {};
  \node[draw, circle, fill=black, inner sep=1.5pt] (d) at (6, 0) {};

  % names of vertices
  \node[right, xshift=-6pt, yshift=7pt] at (r.north) {$r$};
  \node[right, xshift=-6pt, yshift=7pt] at (a.north) {$a$};
  \node[right, xshift=-6pt, yshift=7pt] at (b.north) {$b$};
  \node[right, xshift=-6pt, yshift=7pt] at (c.north) {$c$};
  \node[right, xshift=-6pt, yshift=7pt] at (d.north) {$d$};

  % arcs
  \draw[->, >=stealth, line width=2.5pt, blue] (r) to (a); 
  \draw[->, >=stealth, line width=0.9pt, blue] (a) to (b); 
%  \draw[->, >=stealth, line width=0.9pt, blue] (b) to (c); 

  \draw[->, >=stealth, line width=1.3pt, double, red] (c) to (d); 
%  \draw[->, >=stealth, line width=1.3pt, double, red] (b) to[bend left=60] (a);
  \draw[->, >=stealth, line width=1.3pt, double, red] (c) to[bend left=60] (b);
  \draw[->, >=stealth, line width=0.6pt, double, red] (d) to[bend left=60] (c);
\end{tikzpicture}

%% file: tikz_instance3_5.tex
\begin{tikzpicture}[xscale=0.7, yscale=0.7]
 % braces
  \draw [decorate, decoration = {brace}]  (6.1, -0.8)--(2.9, -0.8) node[below, yshift = -0.5mm, pos=0.5] {$C_1=C_2$}; 
  \draw [decorate, decoration = {brace}]  (6.1, -1.6)--(1.4, -1.6) node[below, yshift = -0.8mm, pos=0.5] {$C_3$}; 
  \draw [decorate, decoration = {brace}]  (6.1, -2.4)--(0.1, -2.4) node[below, yshift = -0.8mm, pos=0.5] {$C_4$}; 
  % vertices
  \node[draw, circle, fill=black, inner sep=1.5pt] (r) at (0, 0) {};
  \node[draw, circle, fill=black, inner sep=1.5pt] (a) at (1.5, 0) {};
  \node[draw, circle, fill=black, inner sep=1.5pt] (b) at (3, 0) {};
  \node[draw, circle, fill=black, inner sep=1.5pt] (c) at (4.5, 0) {};
  \node[draw, circle, fill=black, inner sep=1.5pt] (d) at (6, 0) {};

  % names of vertices
  \node[right, xshift=-6pt, yshift=7pt] at (r.north) {$r$};
  \node[right, xshift=-6pt, yshift=7pt] at (a.north) {$a$};
  \node[right, xshift=-6pt, yshift=7pt] at (b.north) {$b$};
  \node[right, xshift=-6pt, yshift=7pt] at (c.north) {$c$};
  \node[right, xshift=-6pt, yshift=7pt] at (d.north) {$d$};

  % arcs
  \draw[->, >=stealth, line width=2.5pt, blue] (r) to (a); 
  \draw[->, >=stealth, line width=2.5pt, blue] (a) to (b); 
%  \draw[->, >=stealth, line width=0.9pt, blue] (b) to (c); 

  \draw[->, >=stealth, line width=1.3pt, double, red] (c) to (d); 
  \draw[->, >=stealth, line width=0.6pt, double, red] (b) to[bend left=60] (a);
%  \draw[->, >=stealth, line width=1.3pt, double, red] (c) to[bend left=60] (b);
  \draw[->, >=stealth, line width=0.6pt, double, red] (d) to[bend left=60] (c);
\end{tikzpicture}

%% file: tikz_instance3_6.tex
\begin{tikzpicture}[xscale=0.7, yscale=0.7]
 % braces
  \draw [decorate, decoration = {brace}]  (6.1, -0.8)--(4.4, -0.8) node[below, yshift = -0.8mm, pos=0.5] {$C_1$}; 
  \draw [decorate, decoration = {brace}]  (6.1, -1.6)--(2.9, -1.6) node[below, yshift = -0.8mm, pos=0.5] {$C_2$}; 
  \draw [decorate, decoration = {brace}]  (6.1, -2.4)--(1.4, -2.4) node[below, yshift = -0.8mm, pos=0.5] {$C_3$}; 
  \draw [decorate, decoration = {brace}]  (6.1, -3.2)--(0.1, -3.2) node[below, yshift = -0.8mm, pos=0.5] {$C_4$}; 
  % vertices
  \node[draw, circle, fill=black, inner sep=1.5pt] (r) at (0, 0) {};
  \node[draw, circle, fill=black, inner sep=1.5pt] (a) at (1.5, 0) {};
  \node[draw, circle, fill=black, inner sep=1.5pt] (b) at (3, 0) {};
  \node[draw, circle, fill=black, inner sep=1.5pt] (c) at (4.5, 0) {};
  \node[draw, circle, fill=black, inner sep=1.5pt] (d) at (6, 0) {};

  % names of vertices
  \node[right, xshift=-6pt, yshift=7pt] at (r.north) {$r$};
  \node[right, xshift=-6pt, yshift=7pt] at (a.north) {$a$};
  \node[right, xshift=-6pt, yshift=7pt] at (b.north) {$b$};
  \node[right, xshift=-6pt, yshift=7pt] at (c.north) {$c$};
  \node[right, xshift=-6pt, yshift=7pt] at (d.north) {$d$};

  % arcs
  \draw[->, >=stealth, line width=2.5pt, blue] (r) to (a); 
  \draw[->, >=stealth, line width=2.5pt, blue] (a) to (b); 
  \draw[->, >=stealth, line width=2.5pt, blue] (b) to (c); 

  \draw[->, >=stealth, line width=1.3pt, double, red] (c) to (d); 
  \draw[->, >=stealth, line width=0.6pt, double, red] (b) to[bend left=60] (a);
  \draw[->, >=stealth, line width=0.6pt, double, red] (c) to[bend left=60] (b);
  \draw[->, >=stealth, line width=0.6pt, double, red] (d) to[bend left=60] (c);
\end{tikzpicture}